%% file: paper.tex
\documentclass[11pt, oneside]{article}   	% use "amsart" instead of "article" for AMSLaTeX format
\usepackage{geometry}                		% See geometry.pdf to learn the layout options. There are lots.
\usepackage{graphicx}				% Use pdf, png, jpg, or eps§ with pdflatex; use eps in DVI mode

\usepackage{import}
\usepackage{authblk}						% for affiliations
\usepackage{cite}
\usepackage{array}								% for tables
\usepackage{amssymb}
\usepackage{amsmath}
\usepackage{mathtools}
\usepackage{braket}
\usepackage{amsthm}
\usepackage{dsfont}							% for mathbb{1}
\usepackage{verbatim}						% for multi line comments
\usepackage{MnSymbol}						% for \rrangle
\usepackage{mathrsfs}						% for curlier C
\usepackage{mdframed}
\usepackage{caption}
\usepackage[dvipsnames]{xcolor}

\usepackage{float}
\usepackage{placeins}

\RequirePackage{doi}
\usepackage{hyperref}
\hypersetup{
  colorlinks   = true,    % Colours links instead of boxes
  urlcolor     = blue,    % Colour for external hyperlinks
  linkcolor    = blue,    % Colour of internal links
  citecolor    = green      % Colour of citations
}

\DeclareCaptionType{Protocol}

\usepackage{latexdef}

\title{Smooth min-entropy lower bounds for approximation chains}
\author{Ashutosh Marwah\footnote{email: ashutosh.marwah@outlook.com} \ and Fr\'ed\'eric Dupuis}
\affil{\small{D\'epartement d'informatique et de recherche op\'erationnelle,\\ Universit\'e de Montr\'eal,\\ Montr\'eal QC, Canada}}

\date{\today}							% Activate to display a given date or no date

\newcommand{\charFn}[1]{\chi_{#1}}

\begin{document}
\maketitle

\begin{abstract}
    For a state $\rho_{A_1^n B}$, we call a sequence of states $(\sigma_{A_1^k B}^{(k)})_{k=1}^n$ an approximation chain if for every $1 \leq k \leq n$, $\rho_{A_1^k B} \approx_\epsilon \sigma_{A_1^k B}^{(k)}$. In general, it is not possible to lower bound the smooth min-entropy of such a $\rho_{A_1^n B}$, in terms of the entropies of $\sigma_{A_1^k B}^{(k)}$ without incurring very large penalty factors. In this paper, we study such approximation chains under additional assumptions. We begin by proving a simple entropic triangle inequality, which allows us to bound the smooth min-entropy of a state in terms of the R\'enyi entropy of an arbitrary auxiliary state while taking into account the smooth max-relative entropy between the two. Using this triangle inequality, we create lower bounds for the smooth min-entropy of a state in terms of the entropies of its approximation chain in various scenarios. In particular, utilising this approach, we prove approximate versions of the asymptotic equipartition property and entropy accumulation. In the companion paper \cite{Marwah23-src_corr}, we show that the techniques developed in this paper can be used to prove the security of quantum key distribution in the presence of source correlations. 
\end{abstract}

\tableofcontents

\section{Introduction}
\begin{sloppypar}
One-shot information theory investigates the behaviour of tasks in communication and cryptography under general unstructured processes, as opposed to independent and identically distributed (i.i.d) processes, where the states or the tasks themselves have a certain tensor product structure. This is crucial for information theoretically secure cryptography, where one cannot place any kind of assumption on the actions of the adversary (see, for example, \cite{Tomamichel12,Konig12}). To prove security for such protocols, a common strategy is to show that some smooth min-entropy is sufficiently large. For this reason, the smooth min-entropy \cite{Renner06,Renner05} is one of the most important quantities in one-shot information theory. \\

The smooth min-entropy $H^{\epsilon}_{\min}(K|E)_{\rho}$ for the classical-quantum state ${\rho = \sum_{k} p(k)\ket{k}\bra{k}\otimes \rho_{E|k}}$ characterises the amount of randomness one can extract from the classical register $K$ independent of the adversary's register $E$~\cite{tssr10}. It behaves very differently from the von Neumann conditional entropy, which characterises tasks in the i.i.d setting, and the difference between the two can be very large. Roughly speaking, the smooth min-entropy places a much higher weight on the worst possible scenario of the conditioning register, whereas the von Neumann entropy places an equal weight on all possible scenarios. \\

An important and interesting argument, which works with the von Neumann conditional entropy but fails with the smooth min-entropy, is that of proving lower bounds on the entropy using an \emph{approximation chain}. We call a sequence of states\footnote{For $n$ quantum registers $(X_1, X_2, \cdots, X_n)$, the notation $X_i^j$ refers to the set of registers $(X_i, X_{i+1}, \cdots, X_{j})$.} $(\sigma^{(k)}_{A_1^k B})_{k=1}^n$ an $\epsilon$-approximation chain for the state $\rho_{A_1^n B}$ if for every $k$, we can approximate the partial state $\rho_{A_1^k B}$ as $\Vert{\rho_{A_1^k B} - \sigma^{(k)}_{A_1^k B}}\Vert_1 \leq {\epsilon}$. If one can further prove that these states satisfy $H(A_k | A_{1}^{k-1} B)_{\sigma^{(k)}} \geq c$ for some $c>0$ sufficiently large, then the following simple argument shows that $H(A_1^n | B)_{\rho}$ is large:
\begin{align*}
    H(A_1^n | B)_{\rho} &= \sum_{k=1}^n H(A_k | A_{1}^{k-1} B)_{\rho} \\
    &\geq \sum_{k=1}^n \rndBrk{H(A_k | A_{1}^{k-1} B)_{\sigma^{(k)}} - g(\epsilon)} \\
    &\geq n(c- g(\epsilon))
\end{align*}
where we used continuity of the von Neumann conditional entropy in the second line ($g(\epsilon) = O(\epsilon \log \frac{|A|}{\epsilon})$ is a ``small'' function of $\epsilon$). It is well known that a similar argument is not possible with the smooth min-entropy. Consequently, identities for the smooth min-entropy, like the chain rules \cite{Vitanov13}, are much more restrictive. Tools like entropy accumulation \cite{Dupuis20,Metger22} also seem quite rigid, in the sense that they cannot be applied unless certain (Markov chain or non-signalling) conditions apply. It is also not clear how one could relax the conditions for such tools. In this paper, we consider scenarios consisting of approximation chains, similar to the above, along with additional conditions and prove lower bounds on the appropriate smooth min-entropies. \\

We begin by considering the scenario of \emph{approximately independent registers}, that is, a state $\rho_{A_1^n B}$, which for every $1 \leq k \leq n$ satisfies
\begin{align}
    \frac{1}{2}\norm{\rho_{A_1^k B} - \rho_{A_k} \otimes \rho_{A_1^{k-1} B}}_1 \leq \epsilon. 
    \label{eq:scenario1_eqn}
\end{align}
for some small $\epsilon>0$ and arbitrarily large $n$ (in particular $n \gg \frac{1}{\epsilon}$). That is, for every $k$, the system $A_k$ is almost independent of the system $B$ and everything else, which came before it. For simplicity, let us further assume that for all $k$ the state $\rho_{A_k} = \rho_{A_1}$. Intuitively, one expects that the smooth min-entropy (with the smoothing parameter depending on $\epsilon$ and not on $n$)\footnote{The smoothing parameter must depend on $\epsilon$ in such a scenario. This can be seen by considering the probability distribution $P_{A_1^n B}$ such that $B$ is $0$ with probability $\epsilon$ and $1$ otherwise and $A_1^n $ is a random $n$-bit string if $B=1$ and constant if $B=0$. \label{fn:eps_conc_bad}} for such a state will be large and close to $\approx n (H(A_1)-g'(\epsilon))$ (for some small function $g'(\epsilon)$). However, it is not possible to prove this result using techniques, which rely only on the triangle inequality and smoothing. The triangle inequality, in general, can only be used to bound the trace distance between $\rho_{A_1^n B}$ and $\otimes_{k=1}^n \rho_{A_k} \otimes \rho_B$ by $n \epsilon$, which will result in a trivial bound when $n \gg \frac{1}{\epsilon}$ \footnote{Consider the distribution $Q_{A_1^{2n} B_1^{2n}}$, where for every $i \in [2n]$, the bit $B_i$ is chosen independently and is equal to $0$ with probability $\epsilon$ and is $1$ otherwise. The bit $A_i$ is chosen randomly if $B_i=1$, otherwise it is chosen to be equal to $A_{i-1}$. In this case, $Q_{A_k}$ is the uniformly random distribution for bits and Eq. \ref{eq:scenario1_eqn} is satisfied. Let $I = |\{i \in [n]: A_{2i-1} = A_{2i}\}|$. Then, for $Q_{A_1^{2n} B_1^{2n}}$, this value concentrates around $\frac{n(1+\epsilon)}{2}$, whereas for $\prod_{i=1}^{2n} Q_{A_i} \cdot Q_{B_1^{2n}}$, it concentrates around $\frac{n}{2}$. This shows that $\norm{Q_{A_1^{2n} B_1^{2n}}-\prod_{i=1}^{2n} Q_{A_i} \cdot Q_{B_1^{2n}}}_1 \rightarrow 2$.\label{fn:on_avg_eps_bd}}. Instead, we show that a bound on the entropic distance given by the smooth max-relative entropy between these two states can be used to prove a lower bound for the smooth min-entropy in this scenario. \\

While an upper bound of $n\epsilon$ is trivial and meaningless for the trace distance for large $n$, it is still a meaningful bound for the relative entropy between two states, which is unbounded in general. We can show that the above approximation conditions (Eq. \ref{eq:scenario1_eqn}) also imply that relative entropy distance between $\rho_{A_1^n B}$ and $\otimes_{k=1}^n \rho_{A_k} \otimes \rho_B$ is $n f(\epsilon)$ for some small function $f(\epsilon)$. The substate theorem \cite{Jain02} allows us to transform this relative entropy bound into a smooth max-relative entropy bound. For two general states $\rho_{AB}$ and $\eta_{AB}$, such that $d := D^{\delta}_{\max}(\rho_{AB} || \eta_{AB})$, we can easily bound the smooth min-entropy of $\rho$ in terms of the min-entropy of $\eta$ by observing that
\begin{align}
    \rho_{AB} \approx_{\delta} \tilde{\rho}_{AB} \leq 2^d \eta_{AB} \leq 2^{-(H_{\min}(A|B)_{\eta} - d)} \Id_A \otimes \sigma_B
    \label{eq:simple_Hmin_triangle_ineq}
\end{align}
for some state $\sigma_B$, which satisfies $D_{\max}(\eta_{AB} || \Id_A \otimes \sigma_B) = -H_{\min}(A|B)_{\eta}$. This implies that 
$$H_{\min}^{\delta}(A|B)_{\rho} \geq H_{\min}(A|B)_{\eta} - D^{\delta}_{\max}(\rho_{AB} || \eta_{AB}).$$
We call this an \emph{entropic triangle inequality}, since it is based on the triangle inequality property of $D_{\max}$. We can further improve this smooth min-entropy triangle inequality to (Lemma \ref{lemm:Hmin_rho_to_Halpha_sigma_using_Dmax})
\begin{align}
    H_{\min}^{\epsilon+\delta}(A|B)_{\rho} &\geq \tilde{H}^{\uparrow}_{\alpha}(A|B)_{\eta} - \frac{\alpha}{\alpha-1} D^\epsilon_{\max} (\rho_{AB}|| \eta_{AB}) - \frac{g_1(\delta, \epsilon)}{\alpha-1}
    \label{eq:intro_Hmin_trnf_bd}
\end{align}
for some function $g_1$, $\epsilon + \delta <1$ and $1 < \alpha \leq 2$. Our general strategy for the scenarios considered in this paper is to first bound the ``one-shot information theoretic'' distance (the smooth max-relative entropy distance) between the real state $\rho$ ($\rho_{A_1^n B}$ in the above scenario) and a virtual, but \emph{nicer} state, $\eta$ ($\otimes_{k=1}^n \rho_{A_k} \otimes \rho_B$ above) by $nf(\epsilon)$ for some small $f(\epsilon)$. Then, we use Eq. \ref{eq:intro_Hmin_trnf_bd} above to reduce the problem of bounding the smooth min-entropy on state $\rho$ to that of bounding a $\alpha$-R\'enyi entropy on the state $\eta$. Using this strategy, in Corollary \ref{cor:approx_ind_reg}, we prove that for states satisfying the approximately independent registers assumptions, we have for $\delta = O\rndBrk{\epsilon \log\frac{|A|}{\epsilon}}$ that 
\begin{align}
    H_{\min}^{\delta^{\frac{1}{4}}}(A_1^n|B)_{\rho} &\geq n \rndBrk{H (A_1)_{\rho} - O(\delta^{\frac{1}{4}})} - O\rndBrk{\frac{1}{\delta^{3/4}}}.
\end{align}

Another scenario we consider here is that of approximate entropy accumulation. In the setting for entropy accumulation, a sequence of channels $\cM_k: R_{k-1} \rightarrow A_k B_k R_k$ for $1 \leq k \leq n$ sequentially act on a state $\rho^{(0)}_{R_0 E}$ to produce the state $\rho_{A_1^n B_1^n E} = \cM_n \circ \cdots \circ \cM_1(\rho^{(0)}_{R_0 E})$. It is assumed that the channels $\cM_k$ are such that the Markov chain $A_1^{k-1} \leftrightarrow B_1^{k-1} E \leftrightarrow B_k$ is satisfied for every $k$. This ensures that the register $B_k$ does not reveal any additional information about $A_1^{k-1}$ than what was previously revealed by $B_1^{k-1}E$. The entropy accumulation theorem \cite{Dupuis20}, then provides a tight lower bound for the smooth min-entropy $H^{\delta}_{\min}(A_1^n | B_1^n E)$. We consider an approximate version of the above setting where the channels $\cM_k$ themselves do not necessarily satisfy the Markov chain condition, but they can be $\epsilon$-approximated by a sequence of channels $\cM'_k$, which satisfies certain Markov chain conditions. Such relaxations are important to understand the behaviour of cryptographic protocols, like device-independent quantum key distribution \cite{Ekert91,Friedman18}, which are implemented with imperfect devices \cite{Jain23,Tan23}. Once again we can model this scenario as an approximation chain: for every $1 \leq k \leq n$, the state produced in the $k$th step satisfies
\begin{align*}
    \rho_{A_1^k B_1^k E} &= \tr_{R_k} \circ \cM_k \rndBrk{ \cM_{k-1} \circ \cdots \circ \cM_1(\rho^{(0)}_{R_0 E})} \\
    &\approx_{\epsilon} \tr_{R_k} \circ \cM'_k \rndBrk{ \cM_{k-1} \circ \cdots \circ \cM_1(\rho^{(0)}_{R_0 E})} := \sigma^{(k)}_{A_1^k B_1^k E}.
\end{align*}
Moreover, the assumptions on the channel $\cM'_k$ guarantee that the state $\sigma^{(k)}_{A_1^k B_1^k E}$ satisfies the Markov chain condition $A_1^{k-1} \leftrightarrow B_1^{k-1} E \leftrightarrow B_k$, and so the chain rules and bounds used for entropy accumulation apply for it too. Roughly speaking, we use the chain rules for divergences \cite{Fawzi21} to show that the divergence distance between the states $\rho_{A_1^n B_1^n E} = \cM_n \circ \cdots \circ \cM_1(\rho^{(0)}_{R_0 E})$ and the virtual state $\sigma_{A_1^n B_1^n E} = \cM'_n \circ \cdots \circ \cM'_1(\rho^{(0)}_{R_0 E})$ is relatively small, and then reduce the problem of lower bounding the smooth min-entropy of $\rho_{A_1^n B_1^n E}$ to that of lower bounding an $\alpha$-R\'enyi entropy of $\sigma_{A_1^n B_1^n E}$, which can be done by using the chain rules developed for entropy accumulation\footnote{The channel divergence bounds we are able to prove are too weak for this idea to work as stated here. The actual proof is more complicated. However, this idea works in the classical case. }. In Theorem \ref{th:approx_EAT}, we show the following smooth min-entropy lower bound for the state $\rho_{A_1^n B_1^n E}$ for sufficiently small $\epsilon$ and an arbitrary $\delta>0$
\begin{align}
    H_{\min}^{\delta}(A_1^n|B_1^n E)_{\rho} \geq \sum_{k=1}^n \inf_{\omega} & H(A_k | B_k \tilde{R}_{k-1})_{\cM'_k(\omega)} - nO(\epsilon^{\frac{1}{24}}) - O\rndBrk{\frac{1}{\epsilon^{\frac{1}{24}}}}
\end{align}
where the infimum is over all possible input states $\omega_{R_{k-1} \tilde{R}_{k-1}}$ for reference register $\tilde{R}_{k-1}$ isomorphic to $R_{k-1}$, and the dimensions $|A|$ and $|B|$ are assumed constant while using the asymptotic notation. \\

In the companion paper \cite{Marwah23-src_corr}, we use the techniques developed in this paper to provide a solution for the source correlation problem in quantum key distribution (QKD) \cite{Pereira22}. Briefly speaking, the security proofs of QKD require that one of the honest parties produce randomly and independently sampled quantum states in each round of the protocol. However, the states produced by a realistic quantum source will be somewhat correlated across different rounds due to imperfections. These correlations are called source correlations. Proving security for QKD under such a correlated source is challenging and no general satisfying solution was known. In \cite{Marwah23-src_corr}, we use the entropic triangle inequality to reduce the security of a QKD protocol with a correlated source to that of the QKD protocol with a depolarised variant of the perfect source, for which security can be proven using existing techniques.
\end{sloppypar} 

\section{Background and Notation}

For $n$ quantum registers $(X_1, X_2, \cdots, X_n)$, the notation $X_i^j$ refers to the set of registers $(X_i, X_{i+1}, \cdots, X_{j})$. We use the notation [n] to denote the set $\{1,2, \cdots, n\}$. For a register $A$, $|A|$ represents the dimension of the underlying Hilbert space. If $X$ and $Y$ are Hermitian operators, then the operator inequality $X \geq Y$ denotes the fact that $X-Y$ is a positive semidefinite operator and $X>Y$ denotes that $X-Y$ is a strictly positive operator. A quantum state (or briefly just \emph{state}) refers to a positive semidefinite operator with unit trace. At times, we will also need to consider positive semidefinite operators with trace less than equal to $1$. We call these operators subnormalised states. We will denote the set of registers a quantum state describes (equivalently, its Hilbert space) using a subscript. For example, a quantum state on the register $A$ and $B$, will be written as $\rho_{AB}$ and its partial states on registers $A$ and $B$, will be denoted as $\rho_{A}$ and $\rho_{B}$. The identity operator on register $A$ is denoted using $\Id_A$. A classical-quantum state on registers $X$ and $B$ is given by $\rho_{XB} = \sum_{x} p(x) \ket{x}\bra{x} \otimes \rho_{B|x}$, where $\rho_{B|x}$ are normalised quantum states on register $B$.\\

The term ``channel'' is used for completely positive trace preserving (CPTP) linear maps between two spaces of Hermitian operators. A channel $\cN$ mapping registers $A$ to $B$ will be denoted by $\cN_{A \rightarrow B}$. We write $\text{supp}(X)$ to denote the support of the Hermitian operator $X$ and use $X \ll Y$ to denote that $\text{supp}(X) \subseteq \text{supp}(Y)$. \\

The trace norm is defined as $\norm{X}_1 := \tr\big(\rndBrk{X^\dag X}^{\frac{1}{2}}\big)$. The fidelity between two positive operators $P$ and $Q$ is defined as $F(P,Q)= \norm{\sqrt{P}\sqrt{Q}}_1^2$. The generalised fidelity between two subnormalised states $\rho$ and $\sigma$ is defined as 
\begin{align}
    F_\ast(\rho, \sigma) := \rndBrk{\norm{\sqrt{\rho}\sqrt{\sigma}}_1 + \sqrt{(1- \tr\rho)(1- \tr\sigma)}}^2.
\end{align}
The purified distance between two subnormalised states $\rho$ and $\sigma$ is defined as 
\begin{align}
    P(\rho, \sigma) = \sqrt{1- F_{\ast}(\rho, \sigma)}.
\end{align}
We will also use the diamond norm distance as a measure of the distance between two channels. For a linear transform $\cN_{A \rightarrow B}$ from operators on register $A$ to operators on register $B$, the diamond norm distance is defined as 
\begin{align}
    \norm{\cN_{A \rightarrow B}}_{\diamond} := \max_{X_{AR}: \norm{X_{AR}}_1 \leq 1} \norm{\cN_{A \rightarrow B}(X_{AR})}_1
\end{align}
where the supremum is over all Hilbert spaces $R$ (fixing $|R|= |A|$ is sufficient) and operators $X_{AR}$ such that $\norm{X_{AR}}_1 \leq 1$. \\

Throughout this paper, we use base $2$ for both the functions $\log$ and $\exp$. We follow the notation in Tomamichel's book~\cite{TomamichelBook16} for R\'enyi entropies. For $\alpha \in (0,1) \cup (1,2)$, the Petz $\alpha$-R\'enyi relative entropy between the positive operators $P$ and $Q$ is defined as
\begin{align}
    \bar{D}_{\alpha}(P || Q) = \begin{cases}
        \frac{1}{\alpha-1} \log \tr \frac{\rndBrk{P^\alpha Q^{1-\alpha}}}{\tr(P)} & \text{ if } (\alpha<1 \text{ and } P \not\perp Q) \text{ or } (P \ll Q)\\
        \infty & \text{ else}.
    \end{cases}
\end{align}
The sandwiched $\alpha$-R\'enyi relative entropy for $\alpha \in [\frac{1}{2},1) \cup (1,\infty]$ between the positive operator $P$ and $Q$ is defined as 
\begin{align}
    \tilde{D}_{\alpha}(P || Q) = \begin{cases}
        \frac{1}{\alpha-1} \log \frac{\tr(Q^{-\frac{\alpha'}{2}} P Q^{-\frac{\alpha'}{2}})^{\alpha}}{\tr(P)} & \text{ if } (\alpha<1 \text{ and } P \not\perp Q) \text{ or } (P \ll Q)\\
        \infty & \text{ else}.
    \end{cases}
\end{align}
where $\alpha' = \frac{\alpha-1}{\alpha}$. In the limit $\alpha \rightarrow 
\infty$, the sandwiched divergence becomes equal to the max-relative entropy, $D_{\max}$, which is defined as 
\begin{align}
    D_{\max} (P|| Q) := \inf \curlyBrk{\lambda \in \mathbb{R}: P \leq 2^{\lambda}Q}.
\end{align}
In the limit of $\alpha \rightarrow 1$, both the Petz and the sandwiched relative entropies equal the quantum relative entropy, $D(P ||Q)$, which is defined as 
\begin{align}
    D(P|| Q) := \begin{cases}
        \frac{\tr\rndBrk{P \log P - P \log Q}}{\tr(P)} & \text{ if } (P \ll Q)\\
        \infty & \text{ else}.
    \end{cases}
\end{align}
Given any divergence $\mathbb{D}$, we can define the (stabilised) channel divergence based on $\mathbb{D}$ between two channels $\cN_{A \rightarrow B}$ and $\cM_{A \rightarrow B}$ as \cite{Cooney16,Leditzky18}
\begin{align}
    \mathbb{D}(\cN || \cM) := \sup_{\rho_{AR}} \mathbb{D}(\cN_{A \rightarrow B}(\rho_{AR}) || \cM_{A \rightarrow B}(\rho_{AR})) 
\end{align}
where $R$ is reference register of arbitrary size ($|R| = |A|$ can be chosen when $\mathbb{D}$ satisfies the data processing inequality). \\

We can use the divergences defined above to define the following conditional entropies for the subnormalised state $\rho_{AB}$:
\begin{align*}
    \bar{H}_{\alpha}^{\uparrow} (A|B)_{\rho} &:= \sup_{\sigma_B} - \bar{D}_{\alpha}(\rho_{AB} || \Id_A \otimes \sigma_B) \\
    \tilde{H}_{\alpha}^{\uparrow} (A|B)_{\rho} &:= \sup_{\sigma_B} - \tilde{D}_{\alpha}(\rho_{AB} || \Id_A \otimes \sigma_B) \\
    \bar{H}_{\alpha}^{\downarrow} (A|B)_{\rho} &:= - \bar{D}_{\alpha}(\rho_{AB} || \Id_A \otimes \rho_B) \\
    \tilde{H}_{\alpha}^{\downarrow} (A|B)_{\rho} &:= - \tilde{D}_{\alpha}(\rho_{AB} || \Id_A \otimes \rho_B)
\end{align*}
for appropriate $\alpha$ in the domain of the divergences. The supremum in the definition for $\bar{H}_{\alpha}^{\uparrow}$ and $ \tilde{H}_{\alpha}^{\uparrow}$ is over all quantum states $\sigma_B$ on register $B$. \\

For $\alpha \rightarrow 1$, all these conditional entropies are equal to the von Neumann conditional entropy $H(A|B)$. $\tilde{H}_{\infty}^{\uparrow} (A|B)_{\rho}$ is usually called the min-entropy. The min-entropy is usually denoted as $H_{\min}(A|B)_{\rho}$ and for a subnormalised state can also be defined as 
\begin{align}
    H_{\min}(A|B)_{\rho} &:= \sup \curlyBrk{\lambda \in \mathbb{R}: \text{ there exists state }\sigma_B \text{ such that } \rho_{AB} \leq 2^{-\lambda} \Id_A \otimes \sigma_B}.
\end{align}  
For the purpose of smoothing, define the $\epsilon$-ball around the subnormalised state $\rho$ as the set
\begin{align}
    B_{\epsilon}(\rho) = \{ \tilde{\rho} \geq 0 : P(\rho, \tilde{\rho}) \leq \epsilon \text{ and } \tr\tilde{\rho} \leq 1\}.
\end{align}
We define the smooth max-relative entropy as 
\begin{align}
    D_{\max}^{\epsilon}(\rho || \sigma) = \min_{\tilde{\rho} \in B_{\epsilon}(\rho)} D_{\max}(\tilde{\rho} || \sigma)
\end{align}
The smooth min-entropy of $\rho_{AB}$ is defined as 
\begin{align}
    H_{\min}^{\epsilon}(A|B)_{\rho} = \max_{\tilde{\rho} \in B_{\epsilon}(\rho)} H_{\min}(A|B)_{\tilde{\rho}}.
\end{align} 

\section{Entropic triangle inequality for the smooth min-entropy}
\label{sec:D_max_based_triangle_ineq}

In this section, we derive a simple entropic triangle inequality (Lemma \ref{lemm:Hmin_rho_to_Halpha_sigma_using_Dmax}) for the smooth min-entropy of the form in Eq. \ref{eq:intro_Hmin_trnf_bd}. This Lemma is a direct consequence of the following triangle inequality for $\tilde{D}_{\alpha}$ (see \cite[Theorem 3.1]{Christandl17} for a triangle inequality, which changes the second argument of $\tilde{D}_{\alpha}$). 

\begin{lemma}
    Let $\rho$ and $\eta$ be subnormalised states and $Q$ be a positive operator, then for $\alpha>1$, we have 
    \begin{align*}
        \tilde{D}_{\alpha}(\rho || Q) \leq  \tilde{D}_{\alpha}(\eta || Q) + \frac{\alpha}{\alpha-1} D_{\max}(\rho || \eta) + \frac{1}{\alpha-1} \log \frac{\tr(\eta)}{\tr(\rho)}
    \end{align*}
    and for $\alpha<1$ if one of $\tilde{D}_{\alpha}(\eta || Q)$ and $D_{\max}(\rho || \eta)$ is finite (otherwise we cannot define their difference), we have
    \begin{align*}
        \tilde{D}_{\alpha}(\rho || Q) \geq  \tilde{D}_{\alpha}(\eta || Q) - \frac{\alpha}{1- \alpha} D_{\max}(\rho || \eta)- \frac{1}{1-\alpha} \log \frac{\tr(\eta)}{\tr(\rho)}.
    \end{align*}
    \label{lemm:Dalpha_triangle_ineq}
\end{lemma}
\begin{proof}
    If $D_{\max}(\rho || \eta)= \infty$, then both statements are true trivially. Otherwise, we have that $\rho \leq 2^{D_{\max}(\rho || \eta)} \eta$ and also $\rho \ll \eta$. Now, if $\rho \not\ll Q$ then $\eta \not\ll Q$. Hence, for $\alpha>1$ if $\tilde{D}_{\alpha}(\rho || Q) = \infty$, then $\tilde{D}_{\alpha}(\eta || Q) = \infty$, which means the Lemma is also satisfied in this condition. For $\alpha<1$, if $\tilde{D}_{\alpha}(\rho || Q) = \infty$, then the Lemma is also trivially satisfied. For the remaining cases we have,  
    \begin{align*}
        2^{(\alpha-1) \tilde{D}_{\alpha}(\rho || Q)} &= \frac{\tr\rndBrk{Q^{-\frac{\alpha-1}{2\alpha}} \rho Q^{-\frac{\alpha-1}{2\alpha}}}^{\alpha}}{\tr(\rho)} \\
        &\leq \frac{\tr\rndBrk{Q^{-\frac{\alpha-1}{2\alpha}} 2^{D_{\max}(\rho || \eta)} \eta Q^{-\frac{\alpha-1}{2\alpha}}}^{\alpha}}{\tr(\rho)}\\
        &= \frac{\tr(\eta)}{\tr(\rho)} 2^{\alpha D_{\max}(\rho || \eta)} 2^{(\alpha-1) \tilde{D}_{\alpha}(\eta || Q)}
    \end{align*}
    where we used the fact that $\tr(f(X))$ is monotone increasing if the function $f$ is monotone increasing. Dividing by $(\alpha-1)$ now gives the result.
\end{proof}

We define smooth $\alpha$-R\'enyi conditional entropy as follows to help us amplify the above inequality.

\begin{definition}[$\epsilon$-smooth $\alpha$-R\'enyi conditional entropy]
    \label{defn:smooth_renyi_cond_ent}
    For $\alpha \in (1, \infty]$ and $\epsilon \in [0,1]$, we define the $\epsilon$-smooth $\alpha$-R\'enyi conditional entropy as 
    \begin{align}
        \tilde{H}^{\uparrow}_{\alpha, \epsilon}(A|B)_{\rho}:= \max_{\tilde{\rho}_{AB} \in B_{\epsilon}(\rho_{AB})} \tilde{H}^{\uparrow}_{\alpha}(A|B)_{\tilde{\rho}}.
    \end{align}
\end{definition}

\begin{lemma}
    \label{lemm:using_smooth_Dmax_to_lower_bd_renyi_cond_ent}
    For $\alpha \in (1, \infty]$ and $\epsilon \in [0,1)$, and states $\rho_{AB}$ and $\eta_{AB}$ we have
    \begin{align*}
        \tilde{H}^{\uparrow}_{\alpha, \epsilon}(A|B)_{\rho} \geq \tilde{H}^{\uparrow}_{\alpha}(A|B)_{\eta} - \frac{\alpha}{\alpha-1} D^\epsilon_{\max} (\rho_{AB}|| \eta_{AB}) - \frac{1}{\alpha-1} \log \frac{1}{1-\epsilon^2}.
    \end{align*}
\end{lemma}
\begin{proof}
    Let $\tilde{\rho}_{AB} \in B_{\epsilon}(\rho_{AB})$ be a subnormalised state such that $D_{\max}(\tilde{\rho}_{AB}|| \eta_{AB})= D^\epsilon_{\max} (\rho_{AB}|| \eta_{AB})$. Using Lemma \ref{lemm:Dalpha_triangle_ineq} for $\alpha>1$, we have that for every state $\sigma_B$, we have 
    \begin{align}
        \tilde{D}_{\alpha} (\tilde{\rho}_{AB} || \Id_A \otimes \sigma_B) &\leq \tilde{D}_{\alpha} (\eta_{AB} || \Id_A \otimes \sigma_B) + \frac{\alpha}{\alpha-1} D^{\epsilon}_{\max}({\rho}_{AB} || \eta_{AB}) + \frac{1}{\alpha-1}\log\frac{1}{1- \epsilon^2}
    \end{align}
    where we used the fact that $\tilde{\rho}_{AB} \in B_{\epsilon}(\rho_{AB})$ which implies that $\tr(\tilde{\rho}_{AB}) \geq 1- \epsilon^2$. Since, the above bound is true for arbitrary states $\sigma_B$, we can multiply it by $-1$ and take the supremum to derive
    \begin{align*}
        \tilde{H}^{\uparrow}_{\alpha}(A|B)_{\tilde{\rho}} \geq \tilde{H}^{\uparrow}_{\alpha}(A|B)_{\eta} - \frac{\alpha}{\alpha-1} D^\epsilon_{\max}(\rho_{AB}|| \eta_{AB}) - \frac{1}{\alpha-1} \log \frac{1}{1-\epsilon^2}.
    \end{align*}
    The desired bound follows by using the fact that $\tilde{H}^{\uparrow}_{\alpha, \epsilon}(A|B)_{{\rho}} \geq \tilde{H}^{\uparrow}_{\alpha}(A|B)_{\tilde{\rho}}$. 
\end{proof}
\begin{lemma}
    \label{lemm:lower_bd_smooth_min_ent_with_smooth_cond_ent}
    For a state $\rho_{AB}$, $\epsilon \in [0,1)$, and $\delta \in (0,1)$ such that $\epsilon+ \delta < 1$ and $\alpha \in (1,2]$, we have 
    \begin{align*}
        H_{\min}^{\epsilon+\delta}(A|B)_{\rho} \geq \tilde{H}^{\uparrow}_{\alpha, \epsilon}(A|B)_{\rho} - \frac{g_0(\delta)}{\alpha-1}
    \end{align*}
    where $g_0(x):= - \log(1- \sqrt{1-x^2})$. 
\end{lemma}
\begin{proof}
    First, note that 
    \begin{align}
        H_{\min}^{\epsilon+\delta}(A|B)_{\rho} &\geq \sup_{\tilde{\rho} \in B_{\epsilon}(\rho_{AB})} H_{\min}^{\delta}(A|B)_{\tilde{\rho}}.
        \label{eq:Hmin_break_down_eps_plus_delta}
    \end{align}
    To prove this, consider a $\tilde{\rho}_{AB} \in B_{\epsilon}(\rho_{AB})$ and $\rho'_{AB} \in B_{\delta}(\tilde{\rho}_{AB})$ such that $H_{\min}(A|B)_{{\rho'}} = H^{\delta}_{\min}(A|B)_{\tilde{\rho}}$. Then, using the triangle inequality for the purified distance, we have 
    \begin{align*}
        P(\rho_{AB}, \rho'_{AB}) &\leq P(\rho_{AB}, \tilde{\rho}_{AB})  + P(\tilde{\rho}_{AB}, \rho'_{AB}) \\
        &\leq \epsilon + \delta
    \end{align*}
    which implies that $H_{\min}^{\epsilon+\delta}(A|B)_{\rho} \geq H_{\min}(A|B)_{{\rho'}} = H^{\delta}_{\min}(A|B)_{\tilde{\rho}}$. Since, this is true for all $\tilde{\rho} \in B_{\epsilon}(\rho_{AB})$ the bound in Eq. \ref{eq:Hmin_break_down_eps_plus_delta} is true.\\

    Using this, we have
    \begin{align*}
        H_{\min}^{\epsilon+\delta}(A|B)_{\rho} &\geq \sup_{\tilde{\rho} \in B_{\epsilon}(\rho_{AB})} H_{\min}^{\delta}(A|B)_{\tilde{\rho}}\\
        &\geq \sup_{\tilde{\rho} \in B_{\epsilon}(\rho_{AB})} \curlyBrk{\tilde{H}^{\uparrow}_{\alpha}(A|B)_{\tilde{\rho}} - \frac{g_0(\delta)}{\alpha-1}}\\
        &= \tilde{H}^{\uparrow}_{\alpha, \epsilon}(A|B)_{\rho} - \frac{g_0(\delta)}{\alpha-1}
    \end{align*}
    where we have used \cite[Lemma B.10]{Dupuis20}\footnote{This Lemma is also valid for subnormalised states as long as $\delta \in (0, \sqrt{2 \tr(\tilde{\rho}) - \tr(\tilde{\rho})^2})$ according to \cite[Lemma B.4]{Dupuis20}.} (originally proven in \cite{Tomamichel09}) in the second step.
\end{proof}
\noindent We can combine these two lemmas to derive the following result. 
\begin{lemma}
    For $\alpha \in (1,2]$, $\epsilon \in [0,1)$, and $\delta \in (0,1)$ such that $\epsilon + \delta < 1$ and two states $\rho$ and $\eta$, we have 
    \begin{align}
        H_{\min}^{\epsilon+\delta}(A|B)_{\rho} &\geq \tilde{H}^{\uparrow}_{\alpha}(A|B)_{\eta} - \frac{\alpha}{\alpha-1} D^\epsilon_{\max} (\rho_{AB}|| \eta_{AB}) - \frac{g_1(\delta, \epsilon)}{\alpha-1}
        \label{eq:Hmin_rho_to_Halpha_sigma_using_Dmax}
    \end{align}
    where $g_1(x, y):= - \log(1- \sqrt{1-x^2}) - \log (1-y^2)$. 
    \label{lemm:Hmin_rho_to_Halpha_sigma_using_Dmax}
\end{lemma}
\begin{proof}
    We can combine Lemmas \ref{lemm:using_smooth_Dmax_to_lower_bd_renyi_cond_ent} and \ref{lemm:lower_bd_smooth_min_ent_with_smooth_cond_ent} as follows to derive the bound in the Lemma:
    \begin{align*}
        H_{\min}^{\epsilon+\delta}(A|B)_{\rho} &\geq \tilde{H}^{\uparrow}_{\alpha, \epsilon}(A|B)_{\rho} - \frac{g_0(\delta)}{\alpha-1}\\
        &\geq \tilde{H}^{\uparrow}_{\alpha}(A|B)_{\eta} - \frac{\alpha}{\alpha-1} D^\epsilon_{\max} (\rho_{AB}|| \eta_{AB}) - \frac{1}{\alpha-1}\rndBrk{g_0(\delta) + \log \frac{1}{1-\epsilon^2}}. 
    \end{align*}
\end{proof}
We can use the asymptotic equipartition theorem for smooth min-entropy and max-relative entropy \cite{Tomamichel09,TomamichelThesis12,Tomamichel13} to derive the following novel triangle inequality for the von Neumann conditional entropy. Although we do not use this inequality in this paper, we believe it is interesting and may prove useful in the future.
\begin{corollary}
    For $\alpha \in (1,2]$ and states $\rho_{AB}$ and $\eta_{AB}$, we have that
    \begin{align}
        H(A|B)_{\rho} &\geq \tilde{H}^{\uparrow}_{\alpha}(A|B)_{\eta} - \frac{\alpha}{\alpha-1} D(\rho_{AB}|| \eta_{AB}).
        \label{eq:H_rho_to_Halpha_sigma_using_D}
    \end{align}
    \label{cor:H_rho_to_Halpha_sigma_using_D}
\end{corollary}
\begin{proof}
    Using Lemma \ref{lemm:Hmin_rho_to_Halpha_sigma_using_Dmax} with $\alpha \in (1,2]$, the states $\rho^{\otimes n}_{AB}$, and $\eta^{\otimes n}_{AB}$ and any $\epsilon > 0$ and $\delta > 0$ satisfying the conditions for the Lemma, we get
    \begin{align*}
        & H_{\min}^{\epsilon+\delta}(A_1^n|B_1^n)_{\rho^{\otimes n}} \geq \tilde{H}^{\uparrow}_{\alpha}(A_1^n|B_1^n)_{\eta^{\otimes n}} - \frac{\alpha}{\alpha-1} D^\epsilon_{\max} (\rho_{AB}^{\otimes n}|| \eta^{\otimes n}_{AB}) - \frac{g_1(\delta, \epsilon)}{\alpha-1} \\
        \Rightarrow & \frac{1}{n} H_{\min}^{\epsilon+\delta}(A_1^n|B_1^n)_{\rho^{\otimes n}} \geq \tilde{H}^{\uparrow}_{\alpha}(A|B)_{\eta} - \frac{\alpha}{\alpha-1} \frac{1}{n} D^\epsilon_{\max} (\rho_{AB}^{\otimes n}|| \eta^{\otimes n}_{AB}) - \frac{1}{n} \frac{g_1(\delta, \epsilon)}{\alpha-1}.
    \end{align*}
    Taking the limit of the above for $n \rightarrow \infty$, we get
    \begin{align*}
        & \lim_{n \rightarrow \infty} \frac{1}{n} H_{\min}^{\epsilon+\delta}(A_1^n|B_1^n)_{\rho^{\otimes n}} \geq \tilde{H}^{\uparrow}_{\alpha}(A|B)_{\eta} - \lim_{n \rightarrow \infty}\frac{\alpha}{\alpha-1} \frac{1}{n} D^\epsilon_{\max} (\rho_{AB}^{\otimes n}|| \eta^{\otimes n}_{AB}) - \lim_{n \rightarrow \infty} \frac{1}{n} \frac{g_1(\delta, \epsilon)}{\alpha-1} \\
        \Rightarrow & H(A|B)_{\rho} \geq \tilde{H}^{\uparrow}_{\alpha}(A|B)_{\eta} - \frac{\alpha}{\alpha-1} D(\rho_{AB}|| \eta_{AB})
    \end{align*}
    which proves the claim.
\end{proof}

\section{Approximately independent registers}
\label{sec:approx_ind_reg}

In this section, we introduce our technique for using the smooth min-entropy triangle inequality for considering approximations by studying a state $\rho_{A_1^n B}$ such that for every $k \in [n]$
\begin{align}
    \norm{\rho_{A_1^{k}B} - \rho_{A_k} \otimes \rho_{A_1^{k-1}B}}_1 \leq \epsilon.
    \label{eq:approx_ind_cond}
\end{align}
We assume that the registers $A_k$ all have the same dimension equal to $|A|$. One should think of the registers $A_k$ as the secret information produced during some protocol, which also provides the register $B$ to an adversary. We would like to prove that $H^{f(\epsilon)}_{\min}(A_1^n | B)$ is large (lower bounded by $\Omega(n)$) under the above \emph{approximate independence conditions} for some reasonably small function $f$ of $\epsilon$ and close to $nH(A_1)$, if we assume the states $\rho_{A_k}$ are identical. Let us first examine the case when the states above are completely classical. To show that in this case the smooth min-entropy is high, we will show that the set where the conditional probability $\rho(a_1^n |b) := \frac{\rho(a_1^n b)}{\rho(b)}$ can be large, has a small probability using the Markov inequality. We will use the following lemma for this purpose. 

\begin{lemma}
    Suppose $p, q$ are probability distributions on $\mathcal{X}$ such that $\frac{1}{2}\norm{p-q}_1 \leq \epsilon$, then $S \subseteq \mathcal{X}$ defined as $S:= \{x \in \mathcal{X}: p(x) \leq (1+\epsilon^{1/2}) q(x) \}$ is such that $q(S) \geq 1- \epsilon^{1/2}$ and $p(S)\geq 1- \epsilon^{1/2}- \epsilon$.
    \label{lemm:small_tr_norm_implies_small_Dmax}
\end{lemma}
\begin{proof}
    For $S^c := \mathcal{X} \setminus S$, where $S$ is the set defined above we have that 
    \begin{align*}
        \epsilon \geq \frac{1}{2}\norm{p-q}_1 &= \max_{H \subseteq \mathcal{X}} \vert p(H)- q(H) \vert \\
        &\geq q(S^c) \left\vert \frac{p(S^c)}{q(S^c)} - 1 \right\vert \\
        &\geq q(S^c)\rndBrk{\frac{p(S^c)}{q(S^c)} - 1} \\
        &= q(S^c)\rndBrk{\frac{\sum_{x \in S^c} p(x)}{\sum_{x \in S^c} q(x)} - 1} \\
        &\geq q(S^c)\rndBrk{\frac{\sum_{x \in S^c} (1+\epsilon^{\frac{1}{2}})q(x)}{\sum_{x \in S^c} q(x)} - 1} \\
        &\geq q(S^c) \epsilon^{\frac{1}{2}}
    \end{align*}
    which implies that $q(S^c) \leq \epsilon^{\frac{1}{2}}$. Now, the statement of the Lemma follows. 
\end{proof}
We will also assume for the sake of simplicity that $\rho_{A_k}$ are identical for all $k \in [n]$. Using the Lemma above, for every $k \in [n]$, we know that the set 
\begin{align*}
    B_k :&= \curlyBrk{(a_1^n, b): \rho(a_1^k, b) > (1+\sqrt{\epsilon})\rho(a_1^{k-1}, b)\rho(a_k)} \\
    &= \curlyBrk{(a_1^n, b): \rho(a_k | a_1^{k-1}, b) > (1+\sqrt{\epsilon})\rho(a_k)}
\end{align*}
satisfies $\Pr_\rho(B_k) \leq 2\sqrt{\epsilon}$. We can now define $L = \sum_{k=1}^n \charFn{B_k}$, which is a random variable that simply counts the number of bad sets $B_k$ an element $(a_1^n, b)$ belongs to. Using the Markov inequality, we have
\begin{align*}
    \Pr_{\rho}\sqBrk{L > n\epsilon^{\frac{1}{4}}} \leq \frac{\Expect_\rho[L]}{n\epsilon^{\frac{1}{4}}} \leq 2\epsilon^{\frac{1}{4}}.
\end{align*}
We can define the bad set $\mathcal{B} := \curlyBrk{(a_1^n, b): L(a_1^n, b)>n\epsilon^{\frac{1}{4}}}$, then we can define the subnormalised distribution $\tilde{\rho}_{A_1^n B}$ as 
\begin{align*}
    \tilde{\rho}_{A_1^n B} (a_1^n, b) = 
    \begin{cases}
        \rho_{A_1^n B} (a_1^n, b) & (a_1^n, b) \not\in \mathcal{B} \\
        0 & \text{else}
    \end{cases}.
\end{align*}
We have $P(\tilde{\rho}_{A_1^n B}, \rho_{A_1^n B}) \leq \sqrt{2}\epsilon^{1/8}$. Further, note that for every $(a_1^n, b) \not\in \mathcal{B}$, we have 
\begin{align*}
    \rho(a_1^n | b) &= \prod_{k=1}^n \rho(a_k | a_1^{k-1}, b) \\
    &= \prod_{k: (a_1^n, b) \not\in B_k} \rho(a_k | a_1^{k-1}, b) \prod_{k: (a_1^n, b) \in B_k} \rho(a_k | a_1^{k-1}, b) \\
    &\leq (1+\sqrt{\epsilon})^n \prod_{k: (a_1^n, b) \not\in B_k} \rho_{A_k}(a_k)\\
    &\leq (1+\sqrt{\epsilon})^n 2^{-n(1-\epsilon^{\frac{1}{4}})H_{\min}(A_1)}
\end{align*}
where in the third line we have used the fact that if $(a_1^n, b) \not\in B_k$, then $\rho(a_k | a_1^{k-1}b) \leq (1+\sqrt{\epsilon})\rho_{A_k}(a_k)$ and in the last line we have used the fact that for $(a_1^k, b) \not\in \mathcal{B}$, we have $|\{k \in [n]: (a_1^{n},b) \not\in B_k\}| = n- L(a_1^n, b) \geq n(1-\epsilon^{\frac{1}{4}})$, that all the states $\rho_{A_k}$ are identical and $2^{-H_{\min}(A_k)} = \max_{a_k} \rho_{A_k}(a_k)$. Note that we have essentially proven and used a $D_{\max}$ bound above. This proves the following lower bound for the smooth min-entropy of $\rho$
\begin{align}
    H^{\sqrt{2}\epsilon^{\frac{1}{8}}}_{\min} (A_1^n | B) \geq n(1-\epsilon^{\frac{1}{4}})H_{\min}(A_1) - n \log (1+\sqrt{\epsilon}).
\end{align}
The right-hand side above can be improved to get the Shannon entropy $H$ instead of the min-entropy $H_{\min}$. However, we will not pursue this here, since this bound is sufficient for the purpose of our discussion. \\

Although, we are unable to generalise the classical argument above to the quantum case, it provides a great amount of insight into the approximately independent registers problem. Two important examples of distributions, which satisfy the approximate independence conditions above were mentioned in Footnotes \ref{fn:eps_conc_bad} and \ref{fn:on_avg_eps_bd} earlier. To create the first distribution, we flip a biased coin $B$, which is $0$ with probability $\epsilon$ and $1$ otherwise. If $B=0$, then $A_1^n$ is set to the constant all zero string otherwise it is sampled randomly and independently. For this distribution, once the bad event ($B=0$) is removed, the new distribution has a high min-entropy. On the other hand, for the second distribution, $Q_{A_1^{2n}B_1^{2n}}$, we have that the random bits $B_i$ are chosen independently, with each being equal to $0$ with probability $\epsilon$ and $1$ otherwise. If the bit $B_i$ is $0$, then $A_i$ is set equal to $A_{i-1}$ otherwise it is sampled independently. In this case, there is no small probability (small as a function of $\epsilon$) event, that one can simply remove, so that the distribution becomes i.i.d. However, we expect that with high probability the number of $B_i = 0$ is close to $2n \epsilon$. Given that the distribution samples all the other $A_i$ independently, the smooth min-entropy for the distribution should be close to $2n(1-\epsilon)H(A_1)$. The above argument shows that any distribution satisfying the approximate independence conditions in Eq. \ref{eq:approx_ind_cond} can be handled by combining the methods used for these two example distributions, that is, deleting the bad part of the distribution and recognising that the probability for every element in the rest of the space behaves independently on average.\\

The above classical argument is difficult to generalise to quantum states primarily because the quantum equivalents of Lemma \ref{lemm:small_tr_norm_implies_small_Dmax} are not as nice and simple. Furthermore, quantum conditional probabilities themselves are also difficult to use. Fortunately, the substate theorem serves as the perfect tool for developing a smooth max-relative entropy bound, which we can then use with the min-entropy triangle inequality. The quantum substate theorem \cite{Jain02,Jain11} provides an upper bound on the smooth max relative entropy $D^{\epsilon}_{\max}(\rho || \sigma)$ between two states in terms of their relative entropy $D(\rho || \sigma)$. 

\begin{theorem}[Quantum substate theorem \cite{Jain11}]
    Let $\rho$ and $\sigma$ be two states on the same Hilbert space. Then for any $\epsilon \in (0,1)$, we have
    \begin{align}
        D^{\sqrt{\epsilon}}_{\max}(\rho || \sigma) \leq \frac{D(\rho || \sigma)+1}{\epsilon} + \log \frac{1}{1-\epsilon}.
        \label{eq:subset_th_eq2}
    \end{align}
\end{theorem}

In this section, we will also frequently use the multipartite mutual information \cite{Watanabe60,Horodecki94, Cerf02}. For a state $\rho_{X_1^n}$, the multipartite mutual information between the registers $(X_1, X_2, \cdots , X_n)$ is defined as 
\begin{align}
    I(X_1: X_2: \cdots : X_n)_\rho:= D(\rho_{X_1^n}|| \rho_{X_1} \tensor \rho_{X_2} \tensor \cdots \tensor \rho_{X_n}). 
\end{align}
In other words, it is the relative entropy between $\rho_{X_1^n}$ and $\rho_{X_1} \tensor \rho_{X_2} \tensor \cdots \tensor \rho_{X_n}$. It can easily be shown that the multipartite mutual information satisfies the following identities:
\begin{align}
    I(X_1: X_2: \cdots : X_n)_\rho &= \sum_{k=1}^n H(X_k)_\rho - H(X_1 \cdots\ X_n)_\rho \label{eq:MMI_sum_of_H} \\
    &= \sum_{k=2}^n I(X_k : X_1^{k-1}). \label{eq:MMI_sum_of_I}
\end{align}

Going back to proving a bound for the quantum approximately independent registers problem, note that using the Alicki-Fannes-Winter (AFW) bound \cite{Alicki04,Winter16} for mutual information \cite[Theorem 11.10.4]{Wilde13}, Eq. \ref{eq:approx_ind_cond} implies that for every $k \in [n]$
\begin{align}
    I(A_k : A_1^{k-1}B)_\rho \leq \epsilon \log |A| + g_2\rndBrk{\frac{\epsilon}{2}} \label{eq:mut_info_bd_for_approx_ind}
\end{align}
where $g_2(x):= (x+1) \log(x+1) - x \log (x)$. With this in mind, we can now focus our efforts on proving the following theorem. 
\begin{theorem}
    Let registers $A_k$ have dimension $|A|$ for all $k \in [n]$. Suppose a quantum state $\rho_{A_1^n B}$ is such that for every $k \in [n],$ we have 
    \begin{align}
        I(A_k : A_1^{k-1}B)_{\rho} \leq \epsilon
    \end{align}
    for some $0 < \epsilon< 1$. Then, we have that 
    \begin{align}
        H_{\min}^{\epsilon^{\frac{1}{4}}+\epsilon}(A_1^n|B)_{\rho} &\geq \sum_{k=1}^n H (A_k)_{\rho} - 3 n \epsilon^{\frac{1}{4}} \log(1+2|A|) \nonumber \\
        & \qquad - \frac{2 \log(1+2|A|)}{\epsilon^{3/4}} - \frac{2 \log(1+2|A|)}{\epsilon^{1/4}} \rndBrk{\log (1- \sqrt{\epsilon}) + g_1(\epsilon, \epsilon^{\frac{1}{4}})}
    \end{align}
    where $g_1(x, y):= - \log(1- \sqrt{1-x^2}) - \log (1-y^2)$. In particular, when all the states $\rho_{A_k}$ are identical, we have
    \begin{align}
        H_{\min}^{\epsilon^{\frac{1}{4}}+\epsilon}(A_1^n|B)_{\rho} &\geq n \rndBrk{H (A_1)_{\rho} - 3 \epsilon^{\frac{1}{4}} \log(1+2|A|)} \nonumber\\
        &\qquad - \frac{2 \log(1+2|A|)}{\epsilon^{3/4}} - \frac{2 \log(1+2|A|)}{\epsilon^{1/4}} \rndBrk{\log (1- \sqrt{\epsilon}) + g_1(\epsilon, \epsilon^{\frac{1}{4}})}.
    \end{align}
    \label{th:approx_ind_reg_mut_info_based}
\end{theorem}
\begin{proof}
    First note that we have, 
    \begin{align*}
        I(A_1: A_2: \cdots: A_n: B) &= D(\rho_{A_1^n B} || \bigotimes_{k=1}^n \rho_{A_k} \otimes \rho_B) \\
        &= \sum_{k=1}^n I(A_k : A_1^{k-1}B) \\
        &\leq n \epsilon.
    \end{align*}
    Using the substate theorem, we now have 
    \begin{align*}
        D_{\max}^{\epsilon^{\frac{1}{4}}}\rndBrk{\rho_{A_1^n B} \middle\| \bigotimes_{k=1}^n \rho_{A_k} \otimes \rho_B} &\leq \frac{D(\rho_{A_1^n B} || \bigotimes_{k=1}^n \rho_{A_k} \otimes \rho_B) + 1}{\sqrt{\epsilon}} - \log (1- \sqrt{\epsilon}) \\
        &\leq n \sqrt{\epsilon} + \frac{1}{\sqrt{\epsilon}} - \log (1- \sqrt{\epsilon}). \numberthis
    \end{align*}
    We now define the auxiliary state $\eta_{A_1^n B}:= \bigotimes_{k=1}^n \rho_{A_k} \otimes \rho_B$. Using Lemma \ref{lemm:Hmin_rho_to_Halpha_sigma_using_Dmax}, for $\alpha \in (1,2)$, we can transform the smooth min-entropy into an $\alpha$-R\'enyi entropy on the auxiliary product state $\eta_{A_1^n B}$ as follows:
    \begin{align*}
        H&_{\min}^{\epsilon^{\frac{1}{4}}+\epsilon}(A_1^n|B)_{\rho} \\
        &\geq \tilde{H}^{\uparrow}_{\alpha}(A_1^n|B)_{\eta} - \frac{\alpha}{\alpha-1} D^{\epsilon^{\frac{1}{4}}}_{\max} (\rho_{A_1^n B}|| \eta_{A_1^n B}) - \frac{g_1(\epsilon, \epsilon^{\frac{1}{4}})}{\alpha-1} \\
        &= \sum_{k=1}^n \tilde{H}^{\uparrow}_{\alpha}(A_k)_{\rho} - \frac{\alpha}{\alpha-1} D^{\epsilon^{\frac{1}{4}}}_{\max} (\rho_{A_1^n B}|| \eta_{A_1^n B}) - \frac{g_1(\epsilon, \epsilon^{\frac{1}{4}})}{\alpha-1} \\
        &\geq \sum_{k=1}^n H (A_k)_{\rho} - n(\alpha-1) \log^2(1+2|A|)- \frac{\alpha}{\alpha-1} D^{\epsilon^{\frac{1}{4}}}_{\max} (\rho_{A_1^n B}|| \eta_{A_1^n B}) - \frac{g_1(\epsilon, \epsilon^{\frac{1}{4}})}{\alpha-1} \\
        &\geq \sum_{k=1}^n H (A_k)_{\rho} - n(\alpha-1) \log^2(1+2|A|)- \frac{\alpha}{\alpha-1} n \sqrt{\epsilon} - \frac{\alpha}{\alpha-1} \frac{1}{\sqrt{\epsilon}} - \frac{\alpha}{\alpha-1} \log (1- \sqrt{\epsilon})- \frac{g_1(\epsilon, \epsilon^{\frac{1}{4}})}{\alpha-1}.
    \end{align*}
    In the third line above, we have used \cite[Lemma B.9]{Dupuis20} (which is an improvement of \cite[Lemma 8]{Tomamichel09}), which is valid as long as $\alpha<1+ \frac{1}{\log(1+ 2|A|)}$. We will select $\alpha = 1 + \frac{\epsilon^{1/4}}{\log(1+ 2|A|)}$ for which the above $\alpha$ bound is satisfied, this gives us 
    \begin{align*}
        H_{\min}^{\epsilon^{\frac{1}{4}}+\epsilon}(A_1^n|B)_{\rho} &\geq \sum_{k=1}^n H (A_k)_{\rho} - 3 n \epsilon^{\frac{1}{4}} \log(1+2|A|) - \frac{2 \log(1+2|A|)}{\epsilon^{3/4}} \nonumber \\
        &\qquad - \frac{2 \log(1+2|A|)}{\epsilon^{1/4}} \rndBrk{\log (1- \sqrt{\epsilon}) + g_1(\epsilon, \epsilon^{\frac{1}{4}})}. 
    \end{align*}
\end{proof}
\noindent We can now plug the bound in Eq. \ref{eq:mut_info_bd_for_approx_ind} to derive the following Corollary. 
\begin{corollary}
    Let registers $A_k$ have dimension $|A|$ for all $k \in [n]$. Suppose a quantum state $\rho_{A_1^n B}$ is such that for every $k \in [n],$ we have
    \begin{align}
        \norm{\rho_{A_1^k B} - \rho_{A_k} \otimes \rho_{A_1^{k-1} B}}_1 \leq \epsilon.
    \end{align}
    Then, we have that for $\delta = \epsilon \log |A| + g_2\rndBrk{\frac{\epsilon}{2}}$ such that $0< \delta<1$, 
    \begin{align}
        H_{\min}^{\delta^{\frac{1}{4}}+\delta}(A_1^n|B)_{\rho} &\geq \sum_{k=1}^n H (A_k)_{\rho} - 3 n \delta^{\frac{1}{4}} \log(1+2|A|) \nonumber \\
        & \qquad - \frac{2 \log(1+2|A|)}{\delta^{3/4}} - \frac{2 \log(1+2|A|)}{\delta^{1/4}} \rndBrk{\log (1- \sqrt{\delta}) + g_1(\delta, \delta^{\frac{1}{4}})}
    \end{align}
    where $g_1(x, y) = - \log(1- \sqrt{1-x^2}) - \log(1-y^2)$ and $g_2(x) = (x+1)\log(x+1) - x \log(x)$. In particular, when all the states $\rho_{A_k}$ are identical, we have
    \begin{align}
        H_{\min}^{\delta^{\frac{1}{4}}+\delta}(A_1^n|B)_{\rho} &\geq n \rndBrk{H (A_1)_{\rho} - 3 \delta^{\frac{1}{4}} \log(1+2|A|)} \nonumber \\
        &\qquad - \frac{2 \log(1+2|A|)}{\delta^{3/4}} - \frac{2 \log(1+2|A|)}{\delta^{1/4}} \rndBrk{\log (1- \sqrt{\delta}) + g_1(\delta, \delta^{\frac{1}{4}})}.
    \end{align}
    \label{cor:approx_ind_reg}
\end{corollary}

\subsection{Weak approximate asymptotic equipartition}

We can modify the proof of Theorem \ref{th:approx_ind_reg_mut_info_based} to prove a \emph{weak} approximate asymptotic equipartition property (AEP). 

\begin{theorem}
    Let registers $A_k$ have dimension $|A|$ for all $k \in [n]$ and the registers $B_k$ have dimension $|B|$ for all $k \in [n]$. Suppose a quantum state $\rho_{A_1^n B_1^n E}$ is such that for every $k \in [n],$ we have
    \begin{align}
        \norm{\rho_{A_1^k B_1^k E} - \rho_{A_k B_k} \otimes \rho_{A_1^{k-1} B_1^{k-1} E}}_1 \leq \epsilon.
    \end{align}
    Then, we have that for $\delta = \epsilon \log \rndBrk{|A||B|} + g_2\rndBrk{\frac{\epsilon}{2}}$ such that $0< \delta<1$,
    \begin{align}
        H_{\min}^{\delta^{\frac{1}{4}}+\delta}(A_1^n|B_1^n E)_{\rho} &\geq \sum_{k=1}^n H (A_k|B_k)_{\rho} - 3 n \delta^{\frac{1}{4}} \log(1+2|A|) \nonumber \\
        & \qquad - \frac{2 \log(1+2|A|)}{\delta^{3/4}} - \frac{2 \log(1+2|A|)}{\delta^{1/4}} \rndBrk{\log (1- \sqrt{\delta}) + g_1(\delta, \delta^{\frac{1}{4}})}
    \end{align}
    where $g_1(x, y) = - \log(1- \sqrt{1-x^2}) - \log(1-y^2)$ and $g_2(x) = (x+1)\log(x+1) - x \log(x)$. In particular, when all the states $\rho_{A_k B_k}$ are identical, we have
    \begin{align}
        H_{\min}^{\delta^{\frac{1}{4}}+\delta}(A_1^n|B_1^n E)_{\rho} &\geq n \rndBrk{H (A_1|B_1)_{\rho} - 3 \delta^{\frac{1}{4}} \log(1+2|A|)} \nonumber \\
        &\qquad - \frac{2 \log(1+2|A|)}{\delta^{3/4}} - \frac{2 \log(1+2|A|)}{\delta^{1/4}} \rndBrk{\log (1- \sqrt{\delta}) + g_1(\delta, \delta^{\frac{1}{4}})}.
    \end{align}
    \label{th:weak_approx_AEP}
\end{theorem}
\begin{proof}
    To prove this, we use the auxiliary state $\eta_{A_1^n B_1^n E} := \bigotimes \rho_{A_k B_k} \otimes \rho_E$. Then, we have 
    \begin{align*}
        D(\rho_{A_1^n B_1^n E} || \eta_{A_1^n B_1^n E}) &= I(A_1 B_1: A_2 B_2: \cdots: A_n B_n: E)_{\rho} \\
        &= \sum_{k=1}^n I(A_k B_k: A_1^{k-1} B_1^{k-1} E)_{\rho} \\
        &\leq n\rndBrk{\epsilon \log \rndBrk{|A||B|} + g\rndBrk{\frac{\epsilon}{2}}} = n \delta
    \end{align*}
    where we used the AFW bound for mutual information in the last line \cite[Theorem 11.10.4]{Wilde13}. The rest of the proof follows the proof of Theorem \ref{th:approx_ind_reg_mut_info_based}, only difference being that now we have $\tilde{H}^{\uparrow}_{\alpha}(A_1^n | B_1^n E)_{\eta} = \sum_{k=1}^n \tilde{H}^{\uparrow}_{\alpha}(A_k | B_k)_{\rho}$. 
\end{proof}
We call this generalisation \emph{weak} because the smoothing term ($\delta$) depends on size of the side information $|B|$. In Appendix \ref{sec:size_of_B_approx_EAT}, we show that under the assumptions of the theorem, some sort of bound on the dimension of the registers $B$ is necessary otherwise one cannot have a non-trivial bound on the smooth min-entropy.

\subsection{Simple security proof for sequential device independent quantum key distribution}

The approximately independent register scenario and the associated min-entropy lower bound can be used to provide simple ``proof of concept'' security proofs for cryptographic protocols. In this section, we will briefly sketch a proof for sequential device independent quantum key distribution (DIQKD) to demonstrate this idea. The protocol for sequential DIQKD used in \cite{Friedman20} is presented as Protocol \ref{frame:Seq_DI-QKD_prot}. \\

\begin{figure}
    \begin{mdframed}
    \textbf{Sequential DIQKD protocol}\\

    \textbf{Parameters:}
    \begin{itemize}
        \item $\omega_{\text{exp}}$ is the expected winning probability for the honest implementation of the device
        \item $n \geq 1$ is the number of rounds in the protocol
        \item $\gamma \in (0,1]$ is the fraction of test rounds
    \end{itemize}

    \textbf{Protocol:}
    \begin{enumerate}
        \item For every $0 \leq i \leq n$ perform the following steps:
        \begin{enumerate}
            \item Alice chooses a random $T_i \in \{ 0,1\}$ with $\Pr [T_i =1] =\gamma$. 
            \item Alice sends $T_i$ to Bob. 
            \item If $T_i = 0$, Alice and Bob set the questions $(X_i,Y_i) = (0,2)$, otherwise they sample $(X_i,Y_i)$ uniformly at random from $\{ 0,1\}$.
            \item Alice and Bob use their device with the questions $(X_i,Y_i)$ and obtain the outputs $A_i, B_i$.
        \end{enumerate}
        \item Alice announces her questions $X_1^n$ to Bob. 
        \item \textbf{Error correction:} Alice and Bob use an error correction procedure, which lets Bob obtain the raw key $\tilde{A}_1^n$ (if the error correction protocol succeeds, then $A_1^n = \tilde{A}_1^n$). In case the error correction protocol aborts, they abort the QKD protocol too. 
        \item \textbf{Parameter Estimation:} Bob uses $B_1^n$ and $\tilde{A}_1^n$ to compute the average winning probability $\omega_{\text{avg}}$ on the test rounds. He aborts if $\omega_{\text{avg}} < \omega_{\text{exp}}$
        \item \textbf{Privacy Amplification:} Alice and Bob use a privacy amplification protocol to create a secret key $K$ from $A_1^n$ (using $\tilde{A}_1^n$ for Bob).
    \end{enumerate}
    \end{mdframed}
    {\captionof{Protocol}{}
    \label{frame:Seq_DI-QKD_prot}}
\end{figure}

We consider a simple model for DIQKD, where Eve (the adversary) distributes a state $\rho^{(0)}_{E_A E_B E}$ between Alice and Bob at the beginning of the protocol. Alice and Bob then use their states sequentially as given in Protocol \ref{frame:Seq_DI-QKD_prot}. The $k$th round of the protocol produces the questions $X_k, Y_k$ and $T_k$, the answers $A_k$ and $ B_k$ and transforms the shared state from $\rho^{(k-1)}_{E_A E_B E}$ to $\rho^{(k)}_{E_A E_B E}$. \\

Given the questions and answers of the previous rounds, the state shared between Alice and Bob and their devices in each round can be viewed as a device for playing the CHSH game. Suppose in the $k$th round, the random variables produced in the previous $k-1$ rounds are $r_{k-1}:= x_1^{k-1}, y_1^{k-1}, t_1^{k-1}, a_1^{k-1}, b_1^{k-1}$ and that the state shared between Alice and Bob is $\rho^{(k-1)}_{E_A E_B E| r_{k-1}}$. We can then define $\Pr[W_k |r_{k-1}]$ to be the winning probability of the CHSH game played by Alice and Bob using the state and their devices in the $k$th round. Note that Alice's device cannot distinguish whether the CHSH game is played in a round or is used for key generation. We can further take an average over all the previous round's random variables to derive the probability of winning the $k$th game
\begin{align}
    \Pr[W_k] = \Expect_{r_{k-1}}\sqBrk{\Pr[W_k |r_{k-1}]}.
\end{align}
Alice and Bob randomly sample a subset of the rounds (using the random variable $T_k$) and play the CHSH game on this subset. If the average winning probability of CHSH game on this subset is small, they abort the protocol. For simplicity and brevity, we will assume here that the state $\rho^{(0)}_{E_A E_B E}$ distributed between Alice and Bob at the start of the protocol by Eve has an average winning probability at least $\omega_{{\text{exp}}}$, that is, 
\begin{align}
    \frac{1}{n}\sum_{k=1}^n \Pr[W_k] \geq \omega_{{\text{exp}}} - \delta
\end{align}
for some small $\delta>0$. Using standard sampling arguments it can be argued that either this is true or the protocol aborts with high probability. \\

For any shared state $\sigma_{E_A E_B E}$ (where $E_A$ is held by Alice, $E_B$ is held by Bob and $E$ is held by the adversary) and local measurement devices, if Alice and Bob win the CHSH game with a probability $\omega \in (\frac{3}{4}, \frac{2+\sqrt{2}}{4}]$, then Alice's answer $A$ to the game is random given the questions $X,Y$ and the register $E$ held by adversary. This is quantified by the following entropic bound \cite{Pironio09} (see \cite[Lemma 5.3]{Friedman20} for the following form)
\begin{align}
    H(A| XY E) \geq f(\omega) = \begin{cases}
        1 - h\rndBrk{\frac{1}{2} + \frac{1}{2}\sqrt{16 \omega(\omega-1) +3}} & \text{if } \omega \in [\frac{3}{4}, \frac{2+\sqrt{2}}{4}] \\
        0 & \text{if } \omega \in [0, \frac{3}{4})
    \end{cases}
    \label{eq:single_rnd_bnd}
\end{align}
where $h(\cdot)$ is the binary entropy. The function $f$ is convex over the interval $\sqBrk{0, \frac{2+\sqrt{2}}{4}}$. We plot it in the interval $[\frac{3}{4}, \frac{2+\sqrt{2}}{4}]$ in Figure \ref{fig:sing_rnd_bnd}.\\

\begin{figure}
    \centering
    \includegraphics*[scale=0.6]{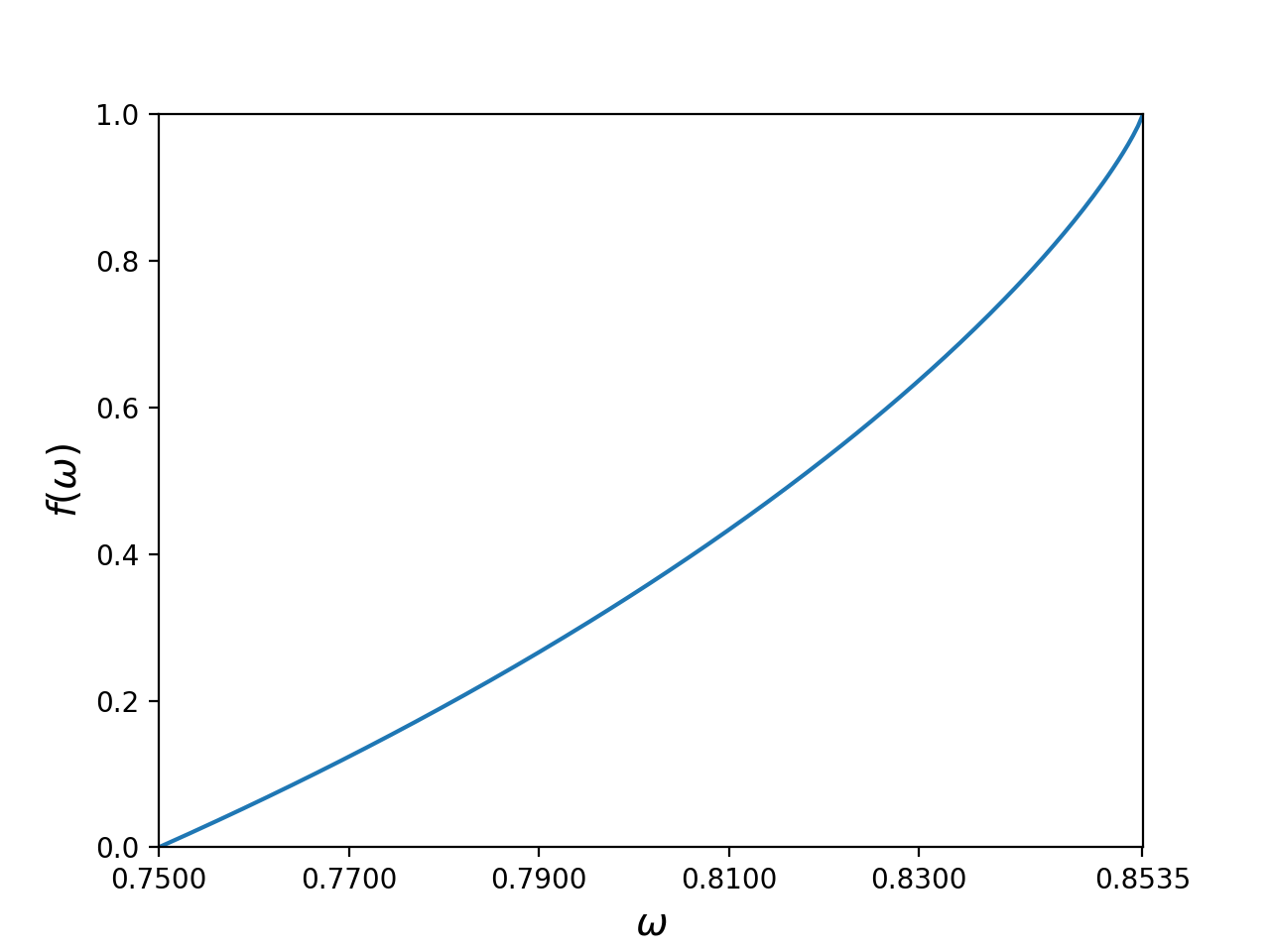}
    \caption{The lower bound in Eq. \ref{eq:single_rnd_bnd} for the interval $[\frac{3}{4}, \frac{2+\sqrt{2}}{4}]$}
    \label{fig:sing_rnd_bnd}
\end{figure}
For $\epsilon >0$, we choose the parameter $\omega_{\text{exp}} \in [\frac{3}{4} + \delta, \frac{2+\sqrt{2}}{4}]$ to be large enough so that
\begin{align}
    1- f(\omega_{\text{exp}} - \delta) = h\rndBrk{\frac{1}{2} + \frac{1}{2}\sqrt{16 (\omega_{\text{exp}}-\delta) (\omega_{\text{exp}} - \delta-1) +3}} \leq \epsilon^4. 
\end{align}
We will now use Eq. \ref{eq:single_rnd_bnd} to bound the von Neumann entropy of the answers given Eve's information for the sequential DIQKD protocol. We have 
\begin{align*}
    H(A_1^n | X_1^n Y_1^n T_1^n E) & = \sum_{k=1}^n H(A_k | A_1^{k-1} X_1^n Y_1^n T_1^n E) \\
    &\overset{(1)}{=} \sum_{k=1}^n H(A_k | A_1^{k-1} X_1^k Y_1^k T_1^k E) \\
    &\overset{(2)}{=} \sum_{k=1}^n H(A_k | X_k Y_k R_{k-1} E) \\
    &= \sum_{k=1}^n \Expect_{r_{k-1} \sim R_{k-1}} \sqBrk{H(A_k | X_k Y_k E)_{\rho^{(k)}_{|r_{k-1}}}} \\
    &\overset{(3)}{\geq} \sum_{k=1}^n \Expect_{r_{k-1} \sim R_{k-1}} \sqBrk{f\rndBrk{\Pr[W_k | r_{k-1}]}} \\
    &\geq \sum_{k=1}^n f\rndBrk{\Pr[W_k]} \\
    &\geq n f\rndBrk{\frac{1}{n} \sum_{k=1}^n \Pr[W_k]} \\
    &\geq nf(\omega_{\text{exp}}-\delta) \geq n(1- \epsilon^4)
\end{align*}
where in (1) we have used the fact that the questions sampled in the rounds after the $k$th round are independent of the random variables in the previous rounds, in (2) we use the fact that Alice's answers are independent of the random variable $T_k$ given the question $X_k$ and we also grouped the random variables generated in the previous round into the random variable $R_{k-1} := A_1^{k-1} B_1^{k-1} X_1^{k-1} Y_1^{k-1} T_1^{k-1}$, in (3) we use the bound in Eq. \ref{eq:single_rnd_bnd}, and in the next two steps we use convexity of $f$. If instead of the von Neumann entropy on the left-hand side above we had the smooth min-entropy, then the bound above would be sufficient to prove the security of DIQKD. However, this argument cannot be easily generalised to the smooth min-entropy because a chain rule like the one used in the first step does not exist for the smooth min-entropy (entropy accumulation \cite{Dupuis20,Metger22} generalises exactly such an argument). We can use the argument used for the approximately independent register case to transform this von Neumann entropy bound to a smooth min-entropy bound. \\

This bound results in the following bound on the multipartite mutual information
\begin{align*}
    I(A_1: \cdots : A_n: X_1^n Y_1^n T_1^n E) &= \sum_{k=1}^n H(A_k) + H(X_1^n Y_1^n T_1^n E) - H(A_1^n X_1^n Y_1^n T_1^n E) \\
    &= \sum_{k=1}^n H(A_k) - H(A_1^n | X_1^n Y_1^n T_1^n E) \\
    &\leq n - n(1-\epsilon^4) = n \epsilon^4
\end{align*}
where we have used the dimension bound $H(A_k) \leq 1$ for every $k \in [n]$. This is the same as the multipartite mutual information bound we derived while analysing approximately independent registers in Theorem \ref{th:approx_ind_reg_mut_info_based}. We can simply use the smooth min-entropy bound derived there here as well. This gives us the bound
\begin{align*}
    H^{2\epsilon}_{\min}(A_1^n | X_1^n Y_1^n T_1^n E) &\geq \sum_{k=1}^n H(A_k) - 3n \epsilon \log 5 - O\rndBrk{\frac{1}{\epsilon^3}}\\
    &=n(1- 3\epsilon \log 5) - O\rndBrk{\frac{1}{\epsilon^3}} \numberthis
\end{align*}
where we have used the fact that the answers $A_k$ can always be assumed to be uniformly distributed \cite{Pironio09,Friedman20}. For every $\epsilon >0$, we can choose a sufficiently large $n$ so that this bound is large and positive. \\

We note that this method is only able to provide ``proof of concept'' or existence type security proofs. This proof method couples the value of the security parameter for privacy amplification $\epsilon$ with the average winning probability, which is not desirable. The parameter $\epsilon$ is chosen according to the security requirements of the protocol and is typically very small. For such values of $\epsilon$, the average winning probability of the protocol will have to be extremely close to the maximum and we cannot realistically expect practical implementations to achieve such high winning probabilities. However, we do expect that this method will make it easier to create ``proof of concept'' type proofs for new cryptographic protocols in the future.  

% source correlations
% \input{source_correlations.tex}

% approximate EAT
\input{Divergence_bound.tex}

\section*{Acknowledgments}
We would like to thank Omar Fawzi for interesting discussions and for pointing out Lemma \ref{lemm:hatD_pure_st}. We also thank Ernest Tan and Amir Arqand, whose observations helped improve Theorem \ref{th:approx_EAT}. AM was supported by the J.A. DeS\`eve Foundation and by bourse d'excellence Google. This work was also supported by the Natural Sciences and Engineering Research Council of Canada.

\input{Appendix.tex}

\bibliographystyle{halpha}
\bibliography{bib}

\end{document}

%% file: Divergence_bound.tex
\section{Approximate entropy accumulation}
\label{sec:approx_EAT}

\begin{figure}
    \centering
    \includegraphics[scale=0.1]{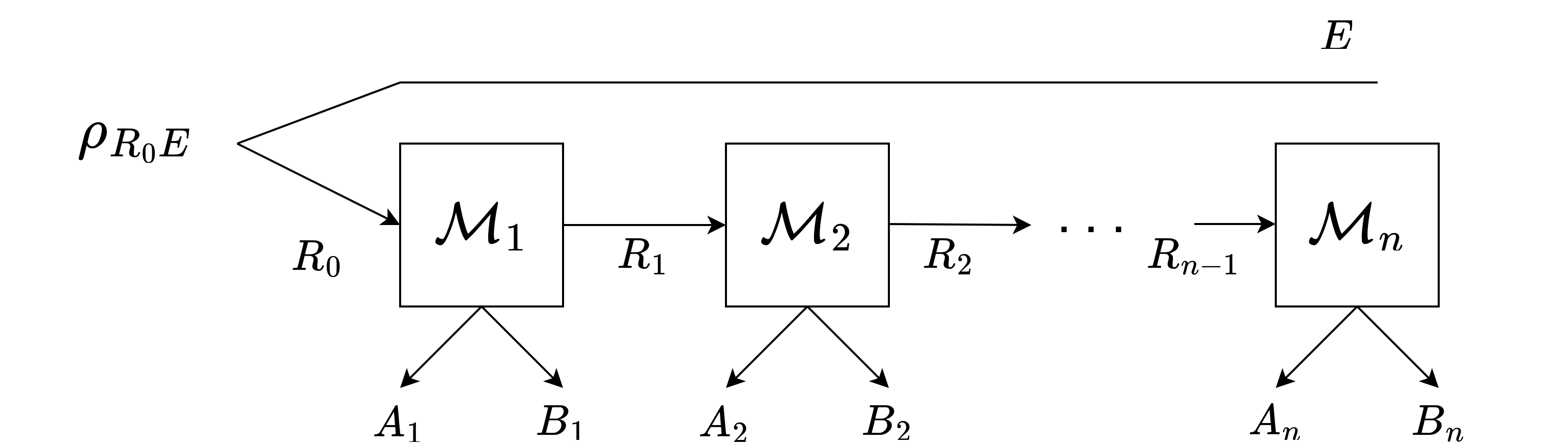}
    \caption{The setting for entropy accumulation and Theorem \ref{th:approx_EAT}. For $k \in [n]$, the channels $\cM_k$ are repeatedly applied to the registers $R_{k-1}$ to produce the ``secret'' information $A_k$ and the side information $B_k$.}
    \label{fig:EAT_setup}
\end{figure}

In general, it is very difficult to estimate the smooth min-entropy of states produced during cryptographic protocols. The entropy accumulation theorem (EAT) \cite{Dupuis20} provides a tight and simple lower bound for the smooth min-entropy $H^{\epsilon}_{\min}(A_1^n | B_1^n E)_{\rho}$ of sequential processes, under certain Markov chain conditions. The state $\rho_{A_1^n  B_1^n E}$ in the setting for EAT is produced by a sequential process of the form shown in Figure \ref{fig:EAT_setup}. The parties implementing the protocol begin with the registers $R_0$ and $E$. In the context of a cryptographic protocol, the register $R_0$ is usually held by the honest parties, whereas the register $E$ is held by the adversary. Then, in each round $k \in [n]$ of the process, a channel $\cM_{k} : R_{k-1} \rightarrow A_k B_k R_k$ is applied on the register $R_{k-1}$ to produce the registers $A_k, B_k$ and $R_k$. The registers $A_1^n$ usually contain a partially secret raw key and the registers $B_1^n$ contain the side information about $A_1^n$ revealed to the adversary during the protocol. EAT requires that for every $k \in [n]$, the side information $B_k$ satisfies the Markov chain $A_1^{k-1} \leftrightarrow B_1^{k-1}E \leftrightarrow B_k$, that is, the side information revealed in the $k$th round does not reveal anything more about the secret registers of the previous rounds than was already known to the adversary through $B_1^{k-1}E$. Under this assumption, EAT provides the following lower bound for the smooth min-entropy
\begin{align}
    H^{\epsilon}_{\min}(A_1^n | B_1^n E)_{\rho} \geq \sum_{k=1}^n \inf_{\omega_{R_{k-1} \tilde{R}}} H(A_k | B_k \tilde{R})_{\cM_k (\omega)} - c\sqrt{n}
\end{align}
where the infimum is taken over all input states to the channels $\cM_k$ and $c>0$ is a constant depending only on $|A|$ (size of registers $A_k$) and $\epsilon$. We will state and prove an approximate version of EAT. Consider the sequential process in Figure \ref{fig:EAT_setup} again. Now, suppose that the channels $\cM_k$ do not necessarily satisfy the Markov chain conditions mentioned above, but each of the channels $\cM_k$ can be $\epsilon$-approximated by $\cM'_k$ which satisfy the Markov chain $A_1^{k-1} \leftrightarrow B_1^{k-1}E \leftrightarrow B_k$ for a certain collection of inputs. The approximate entropy accumulation theorem below provides a lower bound on the smooth min-entropy in such a setting. The proof of this theorem again uses the technique based on the smooth min-entropy triangle inequality developed in the previous section. In this setting too, we have a chain of approximations. For each $k \in [n]$, we have 
\begin{align*}
    \rho_{A_1^k B_1^k E} &= \tr_{R_k} \circ \cM_k \rndBrk{ \cM_{k-1} \circ \cdots \circ \cM_1(\rho^{(0)}_{R_0 E})} \\
    &\approx_{\epsilon} \tr_{R_k} \circ \cM'_k \rndBrk{ \cM_{k-1} \circ \cdots \circ \cM_1(\rho^{(0)}_{R_0 E})} := \sigma^{(k)}_{A_1^k B_1^k E}.
\end{align*}
According to the Markov chain assumption for the channels $\cM'_k$, the state $\sigma^{(k)}_{A_1^k B_1^k E}$, satisfies the Markov chain $A_1^{k-1} \leftrightarrow B_1^{k-1}E \leftrightarrow B_k$. Therefore, we expect that the register $A_k$ adds some entropy to the smooth min-entropy $H^{\epsilon'}_{\min}(A_1^n | B_1^n E)_{\rho}$ and that the information leaked through $B_1^n$ is not too large. We show that this is indeed the case in the approximate entropy accumulation theorem.\\

The approximate entropy accumulation theorem can be used to analyse and prove the security of cryptographic protocols under certain imperfections. For example, the entropy accumulation theorem can be used to prove the security of sequential device independent quantum key distribution (DIQKD) protocols \cite{Friedman18}. In these protocols, the side information $B_k$ produced during each of the rounds are the questions used during the round to play a non-local game, like the CHSH game. In the ideal case, these questions are sampled independently of everything which came before. As an example of an imperfection, we can imagine that some physical effect between the memory storing the secret bits $A_1^{k-1}$ and the device producing the questions may lead to a small correlation between the side information produced during the $k$th round and the secret bits $A_1^{k-1}$ (also see \cite{Jain23,Tan23}). The approximate entropy accumulation theorem below can be used to prove security of DIQKD under such imperfections. We do not, however, pursue this example here and leave the applications of this theorem for future work. In Sec. \ref{sec:testing}, we modify this Theorem to incorporate testing for EAT. 

\begin{theorem}
    For $k \in [n]$, let the registers $A_k$ and $B_k$ be such that $|A_k| = |A|$ and $|B_k| = |B|$. For $k \in [n]$, let $\cM_k$ be channels from $R_{k-1} \rightarrow R_{k} A_k B_k$ and
    \begin{align}
        \rho_{A_1^n B_1^n E} = \tr_{R_n} \circ \cM_n \circ \cdots \circ \cM_1 (\rho^{(0)}_{R_0 E})
        \label{eq:real_st_eq}
    \end{align}
    be the state produced by applying these maps sequentially. Suppose the channels $\cM_k$ are such that for every $k \in [n]$, there exists a channel $\cM'_{k}$ from $R_{k-1} \rightarrow R_{k} A_k B_k$ such that
    \begin{enumerate}
        \item $\cM'_{k}$ $\epsilon$-approximates $\cM_{k}$ in the diamond norm: 
        \begin{align}
            \frac{1}{2}\norm{\cM_k- \cM'_k}_{\diamond} \leq \epsilon 
            \label{eq:map_approx}
        \end{align}
        \item For every choice of a sequence of channels $\cN_i \in \{ \cM_i, \cM'_i \}$ for $i \in [k-1]$, the state $\cM'_k \circ \cN_{k-1} \circ \cdots \circ \cN_1 (\rho^{(0)}_{R_0 E})$ satisfies the Markov chain
        \begin{align}
            A_1^{k-1} \leftrightarrow B_1^{k-1}E \leftrightarrow B_k.
            \label{eq:approx_map_Markov_ch}
        \end{align}
    \end{enumerate} 
        Then, for $0<\delta, \epsilon_1, \epsilon_2<1$ such that $\epsilon_1 + \epsilon_2 <1$, $\alpha \in \rndBrk{1, 1+ \frac{1}{\log(1+2|A|)}}$ and $\beta >1$, we have
        \begin{align*}
            H_{\min}^{\epsilon_1+\epsilon_2}(A_1^n|B_1^n E)_{\rho} \geq \sum_{k=1}^n \inf_{\omega_{R_{k-1} \tilde{R}}} & H(A_k | B_k \tilde{R})_{\cM'_k(\omega)} - n(\alpha-1) \log^2(1+ 2 |A|) \\
            &- \frac{\alpha}{\alpha-1} n \log\rndBrk{1 + \delta \rndBrk{4^{\frac{\alpha-1}{\alpha}\log(|A||B|)}-1}} \\
            & - \frac{\alpha}{\alpha-1} n z_{\beta}(\epsilon, \delta) - \frac{1}{\alpha-1}\rndBrk{g_1(\epsilon_2, \epsilon_1)+ \frac{\alpha g_0(\epsilon_1)}{\beta-1}}.
            \numberthis
            \label{eq:approx_EAT}
        \end{align*}
        where 
        \begin{align}
            z_\beta(\epsilon, \delta) := \frac{\beta+1}{\beta-1}\log\rndBrk{\rndBrk{1+ \sqrt{(1-\delta)\epsilon}}^{\frac{\beta}{\beta+1}} + \rndBrk{\frac{\sqrt{(1-\delta)\epsilon}}{\delta^\beta}}^{{\frac{1}{\beta+1}}}}
            \label{eq:def_z_eps_delta}
        \end{align}
        and $g_1(x,y) = - \log(1- \sqrt{1-x^2}) - \log(1-y^2)$ and the infimum in Eq. \ref{eq:approx_EAT} is taken over all input states $\omega_{R_{k-1} \tilde{R}}$ to the channels $\cM'_k$ where $\tilde{R}$ is a reference register ($\tilde{R}$ can be chosen isomorphic to $R_{k-1}$).
        \label{th:approx_EAT}
\end{theorem}
For the choice of $\beta = 2$, $\delta = \epsilon^{\frac{1}{8}}$, we have
\begin{align*}
    z_2(\epsilon, \delta) \leq 3 \log\rndBrk{\rndBrk{1+\epsilon^{\frac{1}{2}}}^{\frac{2}{3}} + \epsilon^{\frac{1}{12}}}.
\end{align*}
We also have that
\begin{align*}
    \log\rndBrk{1 + \delta 2^{\frac{\alpha-1}{\alpha} 2 \log(|A||B|)}} \leq (|A||B|)^2\epsilon^{\frac{1}{8}}.
\end{align*}
Finally, if we define $\epsilon_r := (|A||B|)^2 \epsilon^{\frac{1}{8}} + 3 \log\rndBrk{\rndBrk{1+\epsilon^{\frac{1}{2}}}^{\frac{2}{3}} + \epsilon^{\frac{1}{12}}}$, and choose $\alpha = \sqrt{\epsilon_r}$, we get the bound 
\begin{align*}
    H_{\min}^{\epsilon_1+\epsilon_2}(A_1^n|B_1^n E)_{\rho} \geq \sum_{k=1}^n \inf_{\omega_{R_k \tilde{R}_k}} & H(A_k | B_k \tilde{R}_k)_{\cM'_k(\omega_{R_k \tilde{R}_k})} \\
    &- n\sqrt{\epsilon_r} (\log^2(1+ 2 |A|) + 2) - \frac{1}{\sqrt{\epsilon_r}}\rndBrk{g_1(\epsilon_2, \epsilon_1)+ 2 g_0(\epsilon_1)} \numberthis
    \label{eq:approx_EAT_eps_only_bd}
\end{align*}
The entropy loss per round in the above bound behaves as $\sim \epsilon^{\frac{1}{24}}$. This dependence on $\epsilon$ is indeed very poor. In comparison, we can carry out a similar proof argument for classical probability distributions to get a dependence of $O(\sqrt{\epsilon})$ (Theorem \ref{th:approx_cl_EAT}). The exponent of $\epsilon$ in our bound seems to be almost a factor of $12$ off from the best possible bound. Roughly speaking, while carrying out the proof classically, we can bound the relevant channel divergences in the proof by $O\rndBrk{\epsilon}$, whereas in Eq. \ref{eq:approx_EAT_eps_only_bd}, we were only able to bound the channel divergence by $\sim \epsilon^{1/12}$. This leads to the deterioration of performance we see here as compared to the classical case. We will discuss this further in Sec. \ref{sec:approx_EAT_lim_and_impr}. \\

% \ifcomments
% \textcolor{blue}{Potential refinement:
%     Also, the maps $\cM_k$ can be from $A_1^{k-1}B_1^{k-1}R_{k-1} E$ to $A_1^{k}B_1^{k}R_{k} E$ and be close to $I_{A_1^{k-1}B_1^{k-1}E} \otimes \cM'_k$, where $\cM'_k$ is from $R_{k-1}$ to $A_k B_k R_k$.}\\
% \fi

In order to prove this theorem, we will use a channel divergence based chain rule. Recently proven chain rules for $\alpha$-R\'enyi relative entropy \cite[Corollary 5.2]{Fawzi21} state that for $\alpha>1$ and states $\rho_A$ and $\sigma_A$, and channels $\mathcal{E}_{A \rightarrow B}$ and $\mathcal{F}_{A \rightarrow B}$, we have
\begin{align}
    \tilde{D}_{\alpha}(\mathcal{E}_{A \rightarrow B}(\rho_{A})|| \mathcal{F}_{A \rightarrow B}(\sigma_{A})) \leq \tilde{D}_{\alpha}(\rho_{A}|| \sigma_{A}) + \tilde{D}_{\alpha}^{\text{reg}}(\mathcal{E}_{A \rightarrow B} || \mathcal{F}_{A \rightarrow B})
    \label{eq:rel_ent_ch_rule}
\end{align}
where $\tilde{D}_{\alpha}^{\text{reg}}(\mathcal{E}_{A \rightarrow B} || \mathcal{F}_{A \rightarrow B}) := \lim_{n \rightarrow \infty} \frac{1}{n} \tilde{D}_{\alpha}(\mathcal{E}_{A \rightarrow B}^{\otimes n} || \mathcal{F}_{A \rightarrow B}^{\otimes n})$ and $\tilde{D}_{\alpha}(\cdot || \cdot)$ is the channel divergence. \\
% no use in using unstabilised divergence because you also have Id_E

Now observe that if we were guaranteed that for the maps in Theorem \ref{th:approx_EAT} above, $\tilde{D}_{\alpha}^{\text{reg}}(\cM_k || \cM'_k) \leq \epsilon$ for every $k$ for some $\alpha >1$. Then, we could use the chain rule in Eq. \ref{eq:rel_ent_ch_rule} as follows
\begin{align*}
    \tilde{D}_{\alpha}(\cM_n & \circ \cdots \circ \cM_1(\rho^{(0)}_{R_0 E}) || \cM'_n \circ \cdots \circ \cM'_1(\rho^{(0)}_{R_0 E})) \\
    &\leq \tilde{D}_{\alpha}(\cM_{n-1} \circ \cdots \circ \cM_1(\rho^{(0)}_{R_0 E}) || \cM'_{n-1} \circ \cdots \circ \cM'_1(\rho^{(0)}_{R_0 E})) + \tilde{D}_{\alpha}^{\text{reg}}(\cM_n || \cM'_n) \\
    &\leq \cdots \\
    &\leq \tilde{D}_{\alpha}(\rho^{(0)}_{R_0 E}||\rho^{(0)}_{R_0 E}) + \sum_{k=1}^n \tilde{D}_{\alpha}^{\text{reg}}(\cM_k || \cM'_k) \\
    &\leq n \epsilon.
\end{align*}
Once we have the above result we can simply use the well known relation between smooth max-relative entropy and $\alpha$-R\'enyi relative entropy \cite[Proposition 6.5]{TomamichelBook16} to get the bound
\begin{align*}
    {D}_{\max}^{\epsilon'}(\cM_n & \circ \cdots \circ \cM_1(\rho^{(0)}_{R_0 E}) || \cM'_n \circ \cdots \circ \cM'_1(\rho^{(0)}_{R_0 E})) \\
    & \leq \tilde{D}_{\alpha}(\cM_n \circ \cdots \circ \cM_1(\rho^{(0)}_{R_0 E}) || \cM'_n \circ \cdots \circ \cM'_1(\rho^{(0)}_{R_0 E})) + \frac{g_0(\epsilon')}{\alpha-1}\\
    &\leq n\epsilon + O(1).
\end{align*}
This bound can subsequently be used in Lemma \ref{lemm:Hmin_rho_to_Halpha_sigma_using_Dmax} to relate the smooth min-entropy of the real state $\cM_n \circ \cdots \circ \cM_1(\rho^{(0)}_{R_0 E})$ with the $\alpha-$R\'enyi conditional entropy of the auxiliary state $\cM'_n \circ \cdots \circ \cM'_1(\rho^{(0)}_{R_0 E})$, for which we can use the original entropy accumulation theorem. \\

In order to prove Theorem \ref{th:approx_EAT}, we broadly follow this idea. However, the condition $\norm{\cM_k- \cM'_k}_{\diamond} \leq \epsilon$ does not lead to any kind of bound on $\tilde{D}_{\alpha}^{\text{reg}}$ or any other channel divergence. We will get around this issue by instead using mixed channels $\cM^{\delta}_{k} := (1-\delta)\cM'_{k} + \delta \cM_{k}$. Also, instead of trying to bound channel divergence in terms of $\tilde{D}_{\alpha}^{\text{reg}}$, we will bound the $D^\#_{\alpha}$ (defined in the next section) channel divergence and use its chain rule. We develop the relevant $\alpha$-R\'enyi divergence bounds for this divergence in the next two subsections and then prove the theorem above in Sec \ref{sec:proof_approx_EAT}. 

\subsection{Divergence bound for approximately equal states}
\label{sec:div_bd_approx_eq_st}

We will use the sharp R\'enyi divergence $D^{\#}_{\alpha}$ defined in Ref. \cite{Fawzi21} (see \cite{Bergh21} for the following equivalent definition) in this section. For $\alpha > 1$ and two positive operators $P$ and $Q$, it is defined 
\begin{align}
    D^{\#}_{\alpha}(P || Q) := \min_{A \geq P} \hat{D}_{\alpha} (A || Q)
\end{align}
where $\hat{D}_{\alpha} (A || Q)$ is the $\alpha$-R\'enyi geometric divergence \cite{Matsumoto18}. For $\alpha>1$, it is defined as 
\begin{align}
    \hat{D}_{\alpha} (A || Q) = \begin{cases}
        \frac{1}{\alpha-1}\log \tr \rndBrk{Q \rndBrk{Q^{-\frac{1}{2}} A Q^{-\frac{1}{2}}}^\alpha} & \text{if } A \ll Q\\
        \infty &\text{otherwise.}
    \end{cases}
\end{align} 
$A$ in the optimisation above is any operator $A \geq P$. In general, such an operator $A$ is unnormalised. We will prove a bound on $D^{\#}_{\alpha}$ between two states in terms of the distance between them and their max-relative entropy. In order to prove this bound, we require the following simple generalisation of the pinching inequality (see for example \cite[Sec. 2.6.3]{TomamichelBook16}). 

\begin{lemma}[Asymmetric pinching]
    For $t>0$, a positive semidefinite operator $X \geq 0 $ and orthogonal projections $\Pi$ and $\Pi_\perp = \Id - \Pi$, we have that 
    \begin{align}
        X \leq (1+t) \Pi X \Pi + \rndBrk{1+ \frac{1}{t}}\Pi_\perp X \Pi_\perp.
    \end{align}
    \label{lemm:asymm_pinching}
\end{lemma}
\begin{proof}
    We will write the positive matrix $X$ as the block matrix
    \begin{align*}
        X= \begin{pmatrix}
            X_1 \ X_2 \\
            X_2^\ast \ X_3
        \end{pmatrix}
    \end{align*}
    where the blocks are partitioned according to the direct sum $\im(\Pi) \oplus \im(\Pi_\perp)$. Then, the statement in the Lemma is equivalent to proving that 
    \begin{align*}
        \begin{pmatrix}
            X_1 \ & X_2 \\
            X_2^\ast \ & X_3
        \end{pmatrix} \leq 
        \begin{pmatrix}
            (1+t) X_1 \ & 0 \\
            0 \ & 0
        \end{pmatrix} + 
        \begin{pmatrix}
            0 \ & 0 \\
            0 \ & \rndBrk{1+ \frac{1}{t}} X_3
        \end{pmatrix}
    \end{align*}
    which is equivalent to proving that
    \begin{align*}
        0 \leq \begin{pmatrix}
            t X_1 \ & - X_2 \\
            - X_2^\ast \ & \frac{1}{t} X_3
        \end{pmatrix}.
    \end{align*}
    This is true because
    \begin{align*}
        \begin{pmatrix}
            t X_1 \ & - X_2 \\
            - X_2^\ast \ & \frac{1}{t} X_3
        \end{pmatrix} = 
        \begin{pmatrix}
            - t^{1/2} \ & 0 \\
            0 \ &  t^{-1/2}
        \end{pmatrix}
        \begin{pmatrix}
            X_1 \ & X_2 \\
            X_2^\ast \ & X_3
        \end{pmatrix}
        \begin{pmatrix}
            - t^{1/2} \ & 0 \\
            0 \ &  t^{-1/2}
        \end{pmatrix} \geq 0 
    \end{align*}
    since $X \geq 0$.
\end{proof}

\begin{lemma}
    Let $\epsilon>0$ and $\alpha\in (1, \infty)$, $\rho$ and $\sigma$ be two normalised quantum states on the Hilbert space $\mathbb{C}^n$ such that $\frac{1}{2}\norm{\rho- \sigma}_1 \leq \epsilon$ and also $D_{\max}(\rho||\sigma) \leq d < \infty$, then we have the bound
    \begin{align}
        D^{\#}_{\alpha}(\rho || \sigma) \leq \frac{\alpha+1}{\alpha-1}\log\rndBrk{(1+ \sqrt{\epsilon})^{\frac{\alpha}{\alpha+1}} + \rndBrk{2^{\alpha d} \sqrt{\epsilon}}^{\frac{1}{\alpha+1}}}.
        \label{eq:Dsharp_close_st_bd}
    \end{align}
    \label{lemm:Dsharp_close_st_bd}
\end{lemma}
\noindent\textbf{Note:} For a fixed $\alpha \in (1, \infty)$, this upper bound tends to zero as $\epsilon \rightarrow 0$. On the other hand, for a fixed $\epsilon \in (0,1)$, the upper bound tends to infinity as $\alpha \rightarrow 1$ (that is, the bound becomes trivial). In Appendix \ref{sec:better_Dsharp_bd_disc}, we show that a bound of this form for $D^{\#}_{\alpha}$ necessarily diverges for $\epsilon>0$ as $\alpha \rightarrow 1$.
\begin{proof}
    Since, $D_{\max}(\rho||\sigma) < \infty$, we have that $\rho \ll \sigma$. We can assume that $\sigma$ is invertible. If it was not, then we could always restrict our vector space to the subspace $\text{supp}(\sigma)$.\\
    
    Let $\rho- \sigma = P- Q$, where $P \geq 0$ is the positive part of the matrix $\rho- \sigma$ and $Q \geq 0$ is its negative part. We then have that $\tr(P)= \tr(Q) \leq \epsilon$. \\
    
    \noindent Further, let 
    \begin{align}
        \sigma^{-\frac{1}{2}} P \sigma^{-\frac{1}{2}} = \sum_{i=1}^n \lambda_i \ket{x_i} \bra{x_i} \label{eq:eigval_decomp}
    \end{align}
    be the eigenvalue decomposition of $\sigma^{-\frac{1}{2}} P \sigma^{-\frac{1}{2}}$. Define the real vector $q \in \mathbb{R}^n$ as 
    \begin{align*}
        q(i) := \bra{x_i} \sigma \ket{x_i}.
    \end{align*}
    Note that $q$ is a probability distribution. Observe that 
    \begin{align*}
        \mathbb{E}_{I \sim q} \sqBrk{\lambda_I} &= \sum_{i=1}^n \lambda_i \bra{x_i} \sigma \ket{x_i} \\
        &= \tr\rndBrk{\sigma \sum_{i=1}^n \lambda_i \ket{x_i} \bra{x_i}}\\
        &= \tr\rndBrk{\sigma \sigma^{-\frac{1}{2}} P \sigma^{-\frac{1}{2}}} \\
        &= \tr(P)\\
        &\leq \epsilon.
    \end{align*}
    Also, observe that $\lambda_i \geq 0$ for all $i \in [n]$ because $\sigma^{-\frac{1}{2}} P \sigma^{-\frac{1}{2}} \geq 0$. Let's define
    \begin{align}
        S := \{i \in [n]: \lambda_i \leq \sqrt{\epsilon}\}.
        \label{eq:defn_setS}
    \end{align}
    Since, $\lambda_i \geq 0$ for all $i \in [n]$, we can use the Markov inequality to show:
    \begin{align*}
        \Pr_q (I \in S^c) &= \Pr_q (\lambda_I > \sqrt{\epsilon}) \\
        &\leq \frac{\Expect_{I \sim q}\sqBrk{\lambda_I}}{\sqrt{\epsilon}} \\
        &\leq \sqrt{\epsilon}.
    \end{align*}
    Thus, if we define the projectors $\Pi := \sum_{i \in S} \ket{x_i} \bra{x_i}$ and $\Pi_{\perp} := \sum_{i \in S^c} \ket{x_i} \bra{x_i} = \Id - \Pi$, we have 
    \begin{align*}
        \tr(\sigma \Pi_\perp) &= \sum_{i \in S^c} \bra{x_i} \sigma \ket{x_i} \\
        &= \Pr_q (I \in S^c) \\
        &\leq \sqrt{\epsilon}. \numberthis \label{eq:overlap_bd}
    \end{align*}
    Moreover, by the definition of set $S$ (Eq. \ref{eq:defn_setS}) we have
    \begin{align*}
        \Pi \sigma^{-\frac{1}{2}} P \sigma^{-\frac{1}{2}} \Pi &= \sum_{i \in S} \lambda_i \ket{x_i} \bra{x_i} \\
        &\leq \sqrt{\epsilon} \Pi \numberthis 
        \label{eq:goodProj_small}
    \end{align*}
    and using $D_{\max}(\rho || \sigma) \leq d$, we have that 
    \begin{align}
        \sigma^{-\frac{1}{2}} \rho \sigma^{-\frac{1}{2}} \leq 2^d \Id.
        \label{eq:badProj_bd}
    \end{align}
    Now, observe that since $\sigma^{-\frac{1}{2}} \rho \sigma^{-\frac{1}{2}} \geq 0$, for an arbitrary $t >0$, using Lemma \ref{lemm:asymm_pinching} we have 
    \begin{align*}
        \sigma^{-\frac{1}{2}} \rho \sigma^{-\frac{1}{2}} &\leq (1+t) \Pi \sigma^{-\frac{1}{2}} \rho \sigma^{-\frac{1}{2}} \Pi + \rndBrk{1+ \frac{1}{t}} \Pi_\perp \sigma^{-\frac{1}{2}} \rho \sigma^{-\frac{1}{2}} \Pi_\perp \\
        &\leq (1+t) \Pi \rndBrk{\Id + \sigma^{-\frac{1}{2}} P \sigma^{-\frac{1}{2}} }\Pi + \rndBrk{1+ \frac{1}{t}} 2^d \Pi_\perp \\
        &\leq (1+t)(1+ \sqrt{\epsilon}) \Pi + \rndBrk{1+ \frac{1}{t}} 2^d \Pi_\perp
    \end{align*} 
    where we have used $\rho \leq \sigma + P$ to bound the first term and Eq. \ref{eq:badProj_bd} to bound the second term in the second line, and Eq. \ref{eq:goodProj_small} to bound $\Pi \sigma^{-\frac{1}{2}} P \sigma^{-\frac{1}{2}} \Pi$ in the last step. \\
    
    We will define $A_t := (1+t)(1+ \sqrt{\epsilon}) \sigma^{\frac{1}{2}}\Pi\sigma^{\frac{1}{2}} + \rndBrk{1+ \frac{1}{t}} 2^d \sigma^{\frac{1}{2}}\Pi_\perp \sigma^{\frac{1}{2}}$. Above, we have shown that $A_t \geq \rho$ for every $t>0$. Therefore, for each $t>0$, $D^\#_\alpha (\rho|| \sigma) \leq \hat{D}_\alpha(A_t || \sigma)$. We will now bound $\hat{D}_\alpha(A_t || \sigma)$ for $\alpha \in (1, \infty)$ as:
    \begin{align*}
        \hat{D}_\alpha(A_t || \sigma) &= \frac{1}{\alpha-1} \log \tr \rndBrk{\sigma \rndBrk{\sigma^{-\frac{1}{2}} A_t \sigma^{-\frac{1}{2}}}^\alpha}\\
        &= \frac{1}{\alpha-1} \log \tr \rndBrk{\sigma \rndBrk{(1+t)(1+ \sqrt{\epsilon}) \Pi + \rndBrk{1+ \frac{1}{t}} 2^d \Pi_\perp}^\alpha}\\
        &= \frac{1}{\alpha-1} \log \tr \rndBrk{\sigma \rndBrk{(1+t)^\alpha (1+ \sqrt{\epsilon})^\alpha \Pi + \rndBrk{1+ \frac{1}{t}}^\alpha 2^{d\alpha} \Pi_\perp}} \\
        &= \frac{1}{\alpha-1} \log \rndBrk{(1+t)^\alpha (1+ \sqrt{\epsilon})^\alpha \tr \rndBrk{\sigma \Pi} + \rndBrk{1+ \frac{1}{t}}^\alpha 2^{d\alpha} \tr \rndBrk{\sigma\Pi_\perp}}\\
        &\leq \frac{1}{\alpha-1} \log \rndBrk{(1+t)^\alpha (1+ \sqrt{\epsilon})^\alpha + \rndBrk{1+ \frac{1}{t}}^\alpha 2^{d\alpha} \sqrt{\epsilon}}
    \end{align*}
    where in the last line we use $\tr(\sigma \Pi)\leq 1$ and $\tr(\sigma \Pi_\perp)\leq \sqrt{\epsilon}$ (Eq. \ref{eq:overlap_bd}). Finally, since $t>0$ was arbitrary, we can choose the $t>0$ which minimizes the right-hand side. For this choice of $t_{\min}= \rndBrk{\frac{2^{\alpha d} \sqrt{\epsilon}}{\rndBrk{1+ \sqrt{\epsilon}}^{\alpha}}}^{\frac{1}{\alpha+1}}$, we get
    \begin{align*}
        \hat{D}_\alpha(A_{t_{\min}} || \sigma) &\leq \frac{\alpha+1}{\alpha-1}\log\rndBrk{(1+ \sqrt{\epsilon})^{\frac{\alpha}{\alpha+1}} + 2^{\frac{\alpha}{\alpha+1} d} \epsilon^{\frac{1}{2(\alpha+1)}}}
    \end{align*}
    which proves the required bound.
\end{proof}

\subsection{Bounding the channel divergence for two channels close to each other}
\label{sec:ch_div_bd}

Suppose there are two channels $\cN$ and $\cM$ mapping registers from the space $A$ to $B$ such that $\frac{1}{2} \norm{\cN- \cM}_\diamond \leq \epsilon$. In general, the channel divergence between two such channels can be infinite because there may be states $\rho$ such that $\cN(\rho) \not\ll \cM(\rho)$. In order to get around this issue, we will use the $\delta-$mixed channel, $\cM_\delta$. For $\delta \in (0,1)$, we define $\cM_\delta$ as 
\begin{align*}
    \mathcal{M}_\delta := (1- \delta) \mathcal{M} + \delta \mathcal{N}.
\end{align*}
This guarantees that $D_{\max}(\cN || \cM_\delta) \leq \log\frac{1}{\delta}$, which is enough to ensure that the divergences we are interested in are finite. Moreover, by mixing $\cM$ with $\cN$, we only decrease the distance: 
\begin{align*}
    \frac{1}{2}\norm{\mathcal{M}_\delta - \cN}_{\diamond} &= \frac{1}{2}\norm{(1-\delta)\cM + \delta\cN - \cN}_{\diamond} \\
    &= (1-\delta)\frac{1}{2}\norm{\cM - \cN}_{\diamond}\\
    &\leq (1-\delta) \epsilon. \numberthis
    \label{eq:N_mixedM_dist}
\end{align*}
We will now show that $D^\#_{\alpha}(\cN || \cM_{\delta})$ is small for an appropriately chosen $\delta$. By the definition of channel divergence, we have that 
\begin{align*}
    D^\#_{\alpha}(\cN || \cM_{\delta}) = \sup_{\rho_{AR}} D^\#_{\alpha}(\cN(\rho_{AR}) || \cM_{\delta}(\rho_{AR}))
\end{align*}
where $R$ is an arbitrary reference system ($\cN, \cM_{\delta}$ map register $A$ to register $B$). We will show that for every $\rho_{AR}$, $D^\#_{\alpha}(\cN(\rho_{AR}) || \cM_{\delta}(\rho_{AR}))$ is small. Note that 
\begin{align*}
    \cM_{\delta}(\rho_{AR}) &= (1-\delta)\cM(\rho_{AR}) + \delta\cN(\rho_{AR})\\
    &\geq \delta\cN(\rho_{AR})
\end{align*}
which implies that $D_{\max}(\cN(\rho_{AR}) || \cM_{\delta}(\rho_{AR})) \leq \log\frac{1}{\delta}$. Also, using Eq. \ref{eq:N_mixedM_dist} have that 
\begin{align*}
    \frac{1}{2}\norm{\mathcal{M}_\delta (\rho_{AR})- \cN (\rho_{AR})}_1 \leq (1-\delta)\epsilon.
\end{align*}
Using Lemma \ref{lemm:Dsharp_close_st_bd}, we have for every $\alpha \in (1, \infty)$
\begin{align*}
    D^\#_{\alpha}(\cN(\rho_{AR}) || \cM_{\delta}(\rho_{AR})) \leq \frac{\alpha+1}{\alpha-1}\log\rndBrk{\rndBrk{1+ \sqrt{(1-\delta)\epsilon}}^{\frac{\alpha}{\alpha+1}} + \rndBrk{\frac{\sqrt{(1-\delta)\epsilon}}{\delta^\alpha}}^{{\frac{1}{\alpha+1}}}}.
\end{align*}
Since, this is true for all $\rho_{AR}$, for every $\alpha \in (1, \infty)$ we have
\begin{align*}
    D^\#_{\alpha}(\cN || \cM_\delta) \leq \frac{\alpha+1}{\alpha-1}\log\rndBrk{\rndBrk{1+ \sqrt{(1-\delta)\epsilon}}^{\frac{\alpha}{\alpha+1}} + \rndBrk{\frac{\sqrt{(1-\delta)\epsilon}}{\delta^\alpha}}^{{\frac{1}{\alpha+1}}}}.
\end{align*}
Note that since $\delta$ was arbitrary, we can choose it appropriately to make sure that the above bound is small, for example by choosing $\delta = \epsilon^\frac{1}{4 \alpha}$, we get the bound
\begin{align*}
    D^\#_{\alpha}(\cN || \cM_\delta) \leq \frac{\alpha+1}{\alpha-1}\log\rndBrk{(1+ \sqrt{\epsilon})^{\frac{\alpha}{\alpha+1}} + \epsilon^{\frac{1}{4(\alpha+1)}}}
\end{align*}
which is a small function of $\epsilon$ in the sense that it tends to $0$ as $\epsilon \rightarrow 0$. We summarise the bound derived above in the following lemma. 
\begin{lemma}
    Let $\epsilon>0$. Suppose channels $\cN$ and $\cM$ from register $A$ to $B$ are such that $\frac{1}{2}\norm{\cN- \cM}_{\diamond} \leq \epsilon$. For $\delta \in (0,1)$, we can define the mixed channel $\cM_\delta := (1- \delta) \mathcal{M} + \delta \mathcal{N}$. Then, for every $\alpha \in (1, \infty)$, we have the following bound on the channel divergence
    \begin{align}
        D^\#_{\alpha}(\cN || \cM_\delta) \leq \frac{\alpha+1}{\alpha-1}\log\rndBrk{\rndBrk{1+ \sqrt{(1-\delta)\epsilon}}^{\frac{\alpha}{\alpha+1}} + \rndBrk{\frac{\sqrt{(1-\delta)\epsilon}}{\delta^\alpha}}^{{\frac{1}{\alpha+1}}}}.
        \label{eq:ch_div_bd}
    \end{align}
    \label{lemm:ch_div_bd}
\end{lemma}

\subsection{Proof of the approximate entropy accumulation theorem}
\label{sec:proof_approx_EAT}

We use the mixed channels defined in the previous section to define the auxiliary state $\cM^\delta_n \circ \cdots \circ \cM^\delta_1(\rho^{(0)}_{R_0 E})$ for our proof. It is easy to show using the divergence bounds in Sec. \ref{sec:ch_div_bd} and the chain rule for $D^\#_{\alpha}$ entropies that the relative entropy distance between the real state and this choice of the auxiliary state is small. However, the state $\cM^\delta_n \circ \cdots \circ \cM^\delta_1(\rho^{(0)}_{R_0 E})$ does not necessarily satisfy the Markov chain conditions required for entropy accumulation. Thus, we also need to reprove the entropy lower bound on this state by modifying the approach used in the proof of the original entropy accumulation theorem. \\

\begin{proof}[Proof of Theorem \ref{th:approx_EAT}]
    Using Lemma \ref{lemm:ch_div_bd}, for every $\delta \in (0,1)$ and for each $k \in [n]$ we have that for every $\beta > 1$, the mixed maps $\cM^\delta_k := (1-\delta)\cM'_k + \delta \cM_k$ satisfy
    \begin{align}
        D^\#_\beta(\cM_k || \cM^\delta_k) &\leq \frac{\beta+1}{\beta-1}\log\rndBrk{\rndBrk{1+ \sqrt{(1-\delta)\epsilon}}^{\frac{\beta}{\beta+1}} + \rndBrk{\frac{\sqrt{(1-\delta)\epsilon}}{\delta^\beta}}^{{\frac{1}{\beta+1}}}} \nonumber \\
        &:= z_\beta(\epsilon, \delta) \label{eq:mixed_ch_div_bd}
    \end{align} 
    where we defined the right-hand side above as $z_\beta(\epsilon, \delta)$. This can be made ``small'' by choosing $\delta = \epsilon^{\frac{1}{4\beta}}$ as was shown in the previous section. We use these maps to define the auxiliary state as
    \begin{align}
        \sigma_{A_1^n B_1^n E} := \cM^\delta_n \circ \cdots \circ \cM^\delta_1(\rho^{(0)}_{R_0 E}). \label{eq:aux_st_defn}
    \end{align}
    Now, we have that for $\beta>1$ and $\epsilon_1 >0$
    \begin{align*}
        D^{\epsilon_1}_{\max}&(\rho_{A_1^n B_1^n E}||\sigma_{A_1^n B_1^n E}) \\
        &\leq \tilde{D}_\beta(\rho_{A_1^n B_1^n E}||\sigma_{A_1^n B_1^n E}) + \frac{g_0(\epsilon_1)}{\beta-1} \\
        &\leq D^\#_\beta (\rho_{A_1^n B_1^n E}||\sigma_{A_1^n B_1^n E}) + \frac{g_0(\epsilon_1)}{\beta-1}\\
        &= D^\#_\beta (\cM_n \circ \cdots \circ \cM_1(\rho^{(0)}_{R_0 E}) || \cM^\delta_n \circ \cdots \circ \cM^\delta_1(\rho^{(0)}_{R_0 E})) + \frac{g_0(\epsilon_1)}{\beta-1}\\
        &\leq D^\#_\beta (\cM_{n-1} \circ \cdots \circ \cM_1(\rho^{(0)}_{R_0 E}) || \cM^\delta_{n-1} \circ \cdots \circ \cM^\delta_1(\rho^{(0)}_{R_0 E})) + D^\#_\beta (\cM_n || \cM^\delta_n) + \frac{g_0(\epsilon_1)}{\beta-1}\\
        &\leq \cdots \\
        &\leq \sum_{k=1}^n D^\#_\beta (\cM_k || \cM^\delta_k) + \frac{g_0(\epsilon_1)}{\beta-1}\\
        &\leq n z_\beta(\epsilon, \delta) + \frac{g_0(\epsilon_1)}{\beta-1} \numberthis
    \end{align*}
    % \ifcomments
    % \textcolor{blue}{(the definition of $D^\#_\beta$ and $\tilde{D}_\beta$ in \cite{Fawzi21} is non-conventional when $\tr(\rho) \neq 1$, but this is not a problem for us since in our case everything is normalised)}
    % \fi
    where the first line follows from \cite[Proposition 6.5]{TomamichelBook16}, the second line follows from \cite[Proposition 3.4]{Fawzi21}, fourth line follows from the chain rule for $D^\#_\beta$ \cite[Proposition 4.5]{Fawzi21}, and the last line follows from Eq. \ref{eq:mixed_ch_div_bd}.\\

    For $\epsilon_2 >0$ and $\alpha \in (1, 1+ \frac{1}{\log(1+2|A|)})$, we can plug the above in the bound provided by Lemma \ref{lemm:Hmin_rho_to_Halpha_sigma_using_Dmax} to get
    \begin{align}
        H_{\min}^{\epsilon_1+\epsilon_2}(A_1^n|B_1^n E)_{\rho} &\geq \tilde{H}^{\uparrow}_{\alpha}(A_1^n|B_1^n E)_{\sigma} - \frac{\alpha}{\alpha-1} n z_{\beta}(\epsilon, \delta) \nonumber \\
        & \qquad - \frac{1}{\alpha-1}\rndBrk{g_1(\epsilon_2, \epsilon_1)+ \frac{\alpha g_0(\epsilon_1)}{\beta-1}}.
        \label{eq:Hmin_rho_to_Halpha_sigma}
    \end{align}
    We have now reduced our problem to lower bounding $\tilde{H}^{\uparrow}_{\alpha}(A_1^n|B_1^n E)_{\sigma}$. Note that we cannot directly use the entropy accumulation here, since the mixed maps $\cM^\delta_k = (1-\delta)\cM'_k + \delta \cM_k$, which means that with $\delta$ probability the $B_k$ register may be correlated with $A_1^{k-1}$ even given $B_1^{k-1}E$, and it may not satisfy the Markov chain required for entropy accumulation. \\

    The application of the maps $\cM^\delta_k$ can be viewed as applying the channel $\cM'_k$ with probability $1-\delta$ and the channel $\cM_k$ with probability $\delta$. We can define the channels $\cN_k$ which map the registers $R_{k-1}$ to $R_k A_k B_k C_k$, where $C_k$ is a binary register. The action of $\cN_k$ can be defined as:
    \begin{enumerate}
        \item Sample the classical random variable $C_k \in \{0,1\}$ independently. $C_k= 1$ with probability $1-\delta$ and $0$ otherwise.
        \item If $C_k =1$ apply the map $\cM'_k$ on $R_{k-1}$, else apply $\cM_k$ on $R_{k-1}$.
    \end{enumerate}
    Let us call $\theta_{A_1^n B_1^n C_1^n E} =\cN_n \circ \cdots \circ \cN_1 (\rho^{(0)}_{R_0 E})$. Clearly $\tr_{C_1^n}\rndBrk{\theta_{A_1^n B_1^n C_1^n E}}= \sigma_{A_1^n B_1^n E}$. Thus, we have
    \begin{align*}
        \tilde{H}^{\uparrow}_{\alpha}(A_1^n|B_1^n E)_{\sigma} &= \tilde{H}^{\uparrow}_{\alpha}(A_1^n|B_1^n E)_{\theta}\\
        &\geq \tilde{H}^{\uparrow}_{\alpha}(A_1^n|B_1^n C_1^n E)_{\theta}.
        \numberthis
        \label{eq:Halpha_sigma_to_theta}
    \end{align*}
    We will now focus on lower bounding $\tilde{H}^{\uparrow}_{\alpha}(A_1^n|B_1^n C_1^n E)_{\theta}$. Using \cite[Proposition 5.1]{TomamichelBook16}, we have that 
    \begin{align*}
        \tilde{H}^{\uparrow}_{\alpha}(A_1^n|B_1^n C_1^n E)_{\theta} = \frac{\alpha}{1-\alpha} \log \sum_{c_1^n} \theta(c_1^n) \exp \rndBrk{\frac{1-\alpha}{\alpha}\tilde{H}^{\uparrow}_{\alpha}(A_1^n|B_1^n E)_{\theta_{|c_1^n}}}. 
    \end{align*}
    We will show that for a given $c_1^n$, the conditional entropy $\tilde{H}^{\uparrow}_{\alpha}(A_1^n|B_1^n E)_{\theta_{|c_1^n}}$ accumulates whenever the ``good'' map $\cM'_k$ is used and loses some entropy for the rounds where the ``bad'' map $\cM_k$ is used. The fact that $c_1^n$ contains far more $1$s than $0$s with a large probability then allows us to prove a lower bound on $\tilde{H}^{\uparrow}_{\alpha}(A_1^n|B_1^n C_1^n E)_{\theta}$.

    \begin{claim}
        Define $h_k := \inf_{\omega} \tilde{H}^{\downarrow}_{\alpha}(A_k | B_k \tilde{R}_{k-1})_{\cM'_k(\omega)}$ where the infimum is over all states $\omega_{R_{k-1} \tilde{R}_{k-1}}$ for a register $\tilde{R}_{k-1}$, which is isomorphic to $R_{k-1}$, and $s:= \log(|A||B|^2)$. Then, we have
        \begin{align}
            \tilde{H}^{\uparrow}_{\alpha}(A_1^n|B_1^n E)_{\theta_{|c_1^n}} \geq \sum_{k=1}^n \rndBrk{\delta({c_k,1}) h_k - \delta({c_k,0})s}
        \end{align}
        where $\delta({x,y})$ is the Kronecker delta function ($\delta({x,y})=1$ if $x=y$ and $0$ otherwise).
    \end{claim}
    \begin{proof}
        We will prove the statement 
        \begin{align*}
            \tilde{H}^{\uparrow}_{\alpha}(A_1^k|B_1^k E)_{\theta_{|c_1^k}} \geq \tilde{H}^{\uparrow}_{\alpha}(A_1^{k-1}|B_1^{k-1} E)_{\theta_{|c_1^{k-1}}} + \rndBrk{\delta({c_k,1}) h_k - \delta({c_k,0})s}
        \end{align*}
        then the claim will follow inductively. We will consider two cases: when $c_k=0$ and when $c_k =1$. First suppose, $c_k=0$ then $\theta_{A_1^k B_1^k E | c_1^k} = \tr_{R_k} \circ \cM_k^{R_{k-1} \rightarrow R_{k} A_k B_k} \rndBrk{\theta_{R_{k-1} A_1^{k-1} B_1^{k-1} E|  c_1^{k}}}$. In this case, we have
        \begin{align*}
            \tilde{H}^{\uparrow}_{\alpha}(A_1^k|B_1^k E)_{\theta_{|c_1^k}} &\geq \tilde{H}^{\uparrow}_{\alpha}(A_1^{k-1}|B_1^k E)_{\theta_{|c_1^k}} - \log|A| \\
            &\geq \tilde{H}^{\uparrow}_{\alpha}(A_1^{k-1}|B_1^{k-1} E)_{\theta_{|c_1^k}} - \log\rndBrk{|A||B|^2}\\
            &= \tilde{H}^{\uparrow}_{\alpha}(A_1^{k-1}|B_1^{k-1} E)_{\theta_{|c_1^{k-1}}} - s
        \end{align*}
        where in the first line we have used the dimension bound in Lemma \ref{lemm:dim_bd}, in the second line we have used the dimension bound in Lemma \ref{lemm:cond_reg_dim_bd} and in the last line we have used $\theta_{A_1^{k-1}B_1^{k-1} E|c_1^k} = \theta_{A_1^{k-1}B_1^{k-1} E|c_1^{k-1}}$. \\

        Now, suppose that $c_k =1$. In this case, we have that $\theta_{A_1^k B_1^k E | c_1^k} = \tr_{R_k} \circ \cM_k' \rndBrk{\theta_{R_{k-1} A_1^{k-1} B_1^{k-1} E|  c_1^{k}}}$ and since $\theta_{R_{k-1} A_1^{k-1} B_1^{k-1} E|  c_1^{k}} = \Phi_{k-1} \circ \Phi_{k-2} \cdots \circ \Phi_{1}(\rho^{(0)}_{R_0 E})$ where each of the $\Phi_i \in \{\cM_i, \cM'_i\}$, using the hypothesis of the theorem we have that the state $\theta_{A_1^k B_1^k E | c_1^k} = \cM_k' \rndBrk{\theta_{R_{k-1} A_1^{k-1} B_1^{k-1} E|  c_1^{k}}}$ satisfies the Markov chain 
        \begin{align*}
            A_1^{k-1} \leftrightarrow B_1^{k-1} E \leftrightarrow B_k.
        \end{align*}
        Now, using Corollary \ref{cor:Markoc_ch_opt_Halpha_bd} (the $ \tilde{H}^{\uparrow}_{\alpha}$ counterpart for \cite[Corollary 3.5]{Dupuis20}, which is the main chain rule used for proving entropy accumulation), we have 
        \begin{align*}
            \tilde{H}^{\uparrow}_{\alpha}(A_1^k|B_1^k E)_{\theta_{|c_1^k}} &\geq \tilde{H}^{\uparrow}_{\alpha}(A_1^{k-1}|B_1^{k-1} E)_{\theta_{|c_1^k}} + \inf_{\omega} \tilde{H}^{\downarrow}_{\alpha}(A_k | B_k \tilde{R}_{k-1})_{\cM'_k(\omega)} \\
            &= \tilde{H}^{\uparrow}_{\alpha}(A_1^{k-1}|B_1^{k-1} E)_{\theta_{|c_1^{k-1}}} + h_k
        \end{align*}
        where in the last line we have again used $\theta_{A_1^{k-1}B_1^{k-1} E|c_1^k} = \theta_{A_1^{k-1}B_1^{k-1} E|c_1^{k-1}}$. Combining these two cases, we have 
        \begin{align}
            \tilde{H}^{\uparrow}_{\alpha}(A_1^k|B_1^k E)_{\theta_{|c_1^k}} \geq \tilde{H}^{\uparrow}_{\alpha}(A_1^{k-1}|B_1^{k-1} E)_{\theta_{|c_1^{k-1}}} + \rndBrk{\delta({c_k,1}) h_k - \delta({c_k,0})s}.
        \end{align}
        Using this bound $n$ times starting with $\tilde{H}^{\uparrow}_{\alpha}(A_1^n|B_1^n E)_{\theta_{|c_1^n}}$ gives us the bound required in the claim. 
        % \ifcomments
        % \textcolor{blue}{
        % Once, we have this equation we can apply it inductively to derive the bound in the claim. 
        % \begin{align*}
        %     \tilde{H}^{\uparrow}_{\alpha}(A_1^n|B_1^n E)_{\theta_{|c_1^n}} &\geq \tilde{H}^{\uparrow}_{\alpha}(A_1^{n-1}|B_1^{n-1} E)_{\theta_{|c_1^{n-1}}} + \rndBrk{\delta({c_n,1}) h_k - \delta({c_n,0})s}\\
        %     &\geq \tilde{H}^{\uparrow}_{\alpha}(A_1^{n-2}|B_1^{n-2} E)_{\theta_{|c_1^{n-2}}} + \sum_{k=n-1}^n \rndBrk{\delta({c_k,1}) h_k - \delta({c_k,0})s}\\
        %     &\geq \cdots \\
        %     &\geq \sum_{k=1}^n \rndBrk{\delta({c_k,1}) h_k - \delta({c_k,0})s}
        % \end{align*}}
        % \fi
    \end{proof}
    \noindent For the sake of clarity let $l_k(c_k) := \rndBrk{\delta({c_k,1}) h_k - \delta({c_k,0})s}$. We will now evaluate 
    \begin{align*}
        \sum_{c_1^n} \theta(c_1^n) \exp \rndBrk{\frac{1-\alpha}{\alpha}\tilde{H}^{\uparrow}_{\alpha}(A_1^n|B_1^n E)_{\theta_{|c_1^n}}} &\leq \sum_{c_1^n} \theta(c_1^n) \exp \rndBrk{\frac{1-\alpha}{\alpha}\sum_{k=1}^n l_k(c_k)}\\
        &= \sum_{c_1^n} \prod_{k=1}^n \theta(c_k) 2^{\frac{1-\alpha}{\alpha}l_k(c_k)}\\
        &= \prod_{k=1}^n \sum_{c_k} \theta(c_k) 2^{\frac{1-\alpha}{\alpha}l_k(c_k)}.
        \numberthis
        \label{eq:Halpha_sum_bd}
    \end{align*}
    Then, we have
    \begin{align*}
        \tilde{H}^{\uparrow}_{\alpha}(A_1^n|B_1^n C_1^n E)_{\theta} &= \frac{\alpha}{1-\alpha} \log \sum_{c_1^n} \theta(c_1^n) \exp_2 \rndBrk{\frac{1-\alpha}{\alpha}\tilde{H}^{\uparrow}_{\alpha}(A_1^n|B_1^n E)_{\theta_{|c_1^n}}}.\\  
        &\geq \frac{\alpha}{1-\alpha} \sum_{k=1}^n \log \sum_{c_k} \theta(c_k) 2^{\frac{1-\alpha}{\alpha}l_k(c_k)} \\
        &= \frac{\alpha}{1-\alpha} \sum_{k=1}^n \log \rndBrk{(1-\delta) 2^{\frac{1-\alpha}{\alpha}h_k} + \delta 2^{-\frac{1-\alpha}{\alpha}s}} \\
        &= \sum_{k=1}^n h_k - \frac{\alpha}{\alpha-1} \sum_{k=1}^n \log\rndBrk{1-\delta + \delta 2^{\frac{\alpha-1}{\alpha}(s+ h_k)}}\\
        &\geq \sum_{k=1}^n h_k - \frac{\alpha}{\alpha-1} n \log\rndBrk{1+ \delta \rndBrk{2^{\frac{\alpha-1}{\alpha}(s+ \log|A|)}-1}}
        \numberthis
        \label{eq:Halpha_theta_EAT_bd_in_Halpha}
    \end{align*}
    where in the second line we have used Eq. \ref{eq:Halpha_sum_bd} and in the last line we have used the fact that $h_k \leq \log|A|$ for all $k \in [n]$.\\

    We restricted the choice of $\alpha$ to the region $\rndBrk{1, 1+ \frac{1}{\log(1+2 |A|)}}$ in the theorem, so that we can now use \cite[Lemma B.9]{Dupuis20} to transform the above to  
    \begin{align*}
        \tilde{H}^{\uparrow}_{\alpha}(A_1^n|B_1^n C_1^n E)_{\theta} \geq \sum_{k=1}^n \inf_{\omega_{R_{k-1} \tilde{R}_{k-1}}} & H(A_k | B_k \tilde{R}_{k-1})_{\cM'_k(\omega)} - n(\alpha-1) \log^2(1+ 2 |A|) \\
        &- \frac{\alpha}{\alpha-1} n \log\rndBrk{1 + \delta \rndBrk{2^{\frac{\alpha-1}{\alpha}2\log(|A||B|)}-1}}.
        \numberthis
        \label{eq:Halpha_theta_EAT_bd}
    \end{align*}
    Putting Eq. \ref{eq:Hmin_rho_to_Halpha_sigma}, Eq. \ref{eq:Halpha_sigma_to_theta}, and Eq. \ref{eq:Halpha_theta_EAT_bd} together, we have 
    \begin{align*}
        H_{\min}^{\epsilon_1+\epsilon_2}(A_1^n|B_1^n E)_{\rho} \geq \sum_{k=1}^n \inf_{\omega_{R_{k-1} \tilde{R}_{k-1}}} & H(A_k | B_k \tilde{R}_{k-1})_{\cM'_k(\omega)} - n(\alpha-1) \log^2(1+ 2 |A|) \\
        &- \frac{\alpha}{\alpha-1} n \log\rndBrk{1 + \delta \rndBrk{2^{\frac{\alpha-1}{\alpha}2\log(|A||B|)}-1}} \\
        & - \frac{\alpha}{\alpha-1} n z_{\beta}(\epsilon, \delta)\frac{1}{\alpha-1}\rndBrk{g_1(\epsilon_2, \epsilon_1)+ \frac{\alpha g_0(\epsilon_1)}{\beta-1}}.
    \end{align*}
\end{proof}

\subsection{Testing}
\label{sec:testing}

We follow \cite{Metger22}, which is itself based on \cite{Dupuis19}, to incorporate testing in Theorem \ref{th:approx_EAT}. \\

In this section, we will consider the channels $\cM_k$ and $\cM'_k$ which map registers $R_{k-1}$ to $A_k B_k X_k R_k$ such that $X_k$ is a classical value which is determined using the registers $A_k$ and $B_k$. Concretely, suppose that for every $k$, there exist a channel $\mathcal{T}_k : A_k B_k \rightarrow A_k B_k X_k$ of the form
\begin{align}
    \mathcal{T}_k (\omega_{A_k B_k}) = \sum_{y, z} \Pi_{A_k}^y \otimes \Pi_{B_k}^z \omega_{A_k B_k} \Pi_{A_k}^y \otimes \Pi_{B_k}^z \otimes \ket{f(y,z)}\bra{f(y,z)}_{X_k}
    \label{eq:test_maps}
\end{align}
where $\{\Pi_{A_k}^y\}_y$ and $\{\Pi_{B_k}^z\}_z$ are orthogonal projectors and $f$ is some deterministic function which uses the measurements $y$ and $z$ to create the output register $X_k$.\\

In order to define the min-tradeoff functions, we let $\mathbb{P}$ be the set of probability distributions over the alphabet of $X$ registers. Let $R$ be any register isomorphic to $R_{k-1}$. For a probability $q \in \mathbb{P}$ and a channel $\cN_{k} : R_{k-1} \rightarrow A_k B_k X_k R_k$, we also define the set 
\begin{align}
    \Sigma_k (q | \cN_{k}) := \curlyBrk{\nu_{A_k B_k X_k R_k R} = \cN_{k}(\omega_{R_{k-1}R}): \text{ for a state } \omega_{R_{k-1}R} \text{ such that }\nu_{X_k}= q}.
\end{align}

\begin{definition}
    A function $f: \mathbb{P} \rightarrow \mathbb{R}$ is called a min-tradeoff function for the channels $\{\cN_k \}_{k=1}^n$ if for every $k \in [n]$, it satisfies
    \begin{align}
        f(q) \leq \inf_{\nu \in \Sigma_k (q| \cN_{k})} H(A_k | B_k R)_{\nu}.
    \end{align}
\end{definition}
We will also need the definitions of the following simple properties of the min-tradeoff functions for our entropy accumulation theorem:
\begin{align}
    &\text{Max}(f) := \max_{q \in \mathbb{P}} f(q)\\
    &\text{Min}(f) := \min_{q \in \mathbb{P}} f(q)\\
    &\text{Min}_{\Sigma}(f) := \min_{q: \Sigma (q) \neq \phi} f(q)\\
    &\text{Var}(f) := \max_{q: \Sigma (q) \neq \phi} \sum_{x} q(x)f(\delta_x)^2 - \rndBrk{\sum_{x} q(x)f(\delta_x)}^2
\end{align}
where $\Sigma(q) := \bigcup_k \Sigma_k(q)$ and $\delta_x$ is the distribution with unit weight on the alphabet $x$. 

\begin{theorem}
    For $k \in [n]$, let the registers $A_k$ and $B_k$ be such that $|A_k| = |A|$ and $|B_k| = |B|$. For $k \in [n]$, let $\cM_k$ be channels from $R_{k-1} \rightarrow R_{k} A_k B_k X_k$ and
    \begin{align}
        \rho_{A_1^n B_1^n X_1^n E} = \tr_{R_n} \circ \cM_n \circ \cdots \circ \cM_1 (\rho^{(0)}_{R_0 E})
        \label{eq:real_st_wtest_eq}
    \end{align}
    be the state produced by applying these maps sequentially. Further, let $\cM_k$ be such that $\cM_k = \mathcal{T}_k \circ \cM^{(0)}_k$ for $\mathcal{T}_k$ defined in Eq. \ref{eq:test_maps} and some channel $\cM^{(0)}_k: R_{k-1} \rightarrow R_{k} A_k B_k$. Suppose the channels $\cM_k$ are such that for every $k \in [n]$, there exists a channel $\cM'_{k}$ from $R_{k-1} \rightarrow R_{k} A_k B_k X_k$ such that
    \begin{enumerate}
        \item $\cM'_k = \mathcal{T}_k \circ \cM^{\prime (0)}_k$ for some channel $\cM^{\prime (0)}_k: R_{k-1} \rightarrow R_{k} A_k B_k$.
        \item $\cM'_{k}$ $\epsilon$-approximates $\cM_{k}$ in the diamond norm: 
        \begin{align}
            \frac{1}{2}\norm{\cM_k- \cM'_k}_{\diamond} \leq \epsilon 
            \label{eq:map_approx_wtest}
        \end{align}
        \item For every choice of a sequence of channels $\cN_i \in \{ \cM_i, \cM'_i \}$ for $i \in [k-1]$, the state $\cM'_k \circ \cN_{k-1} \circ \cdots \circ \cN_1 (\rho^{(0)}_{R_0 E})$ satisfies the Markov chain
        \begin{align}
            A_1^{k-1} \leftrightarrow B_1^{k-1}E \leftrightarrow B_k.
            \label{eq:approx_map_Markov_ch_wtest}
        \end{align}
    \end{enumerate} 
    Then, for an event $\Omega$ defined using $X_1^n$, an affine min-tradeoff function $f$ for $\{\cM'_k\}_{k=1}^n$ such that for every $x_1^n \in \Omega$, $f(\text{freq}(x_1^n)) \geq h$, for parameters $0<\delta, \epsilon_1, \epsilon_3<1$ and $\epsilon_2 := 2\sqrt{\frac{\epsilon_1}{P_{\rho}(\Omega)}}$ such that $\epsilon_2 + \epsilon_3 <1$, $\alpha \in \rndBrk{1,2}$, and $\beta >1$, we have
    \begin{align*}
        H_{\min}^{\epsilon_2+\epsilon_3}(A_1^n|B_1^n E)_{\rho_{|\Omega}} &\geq nh - \frac{(\alpha-1)\ln(2)}{2}\rndBrk{\log(2 |A|^2 + 1) + \sqrt{2+ \text{Var}(f)}}^2 - n(\alpha-1)^2 K_\alpha \\
        &\qquad - \frac{\alpha}{\alpha-1} n \log\rndBrk{1 + \delta \rndBrk{4^{\frac{\alpha-1}{\alpha}(\log(|A||B|) + \text{max}(f) - \text{min}(f)+1)}-1}} \\
        & \qquad - \frac{\alpha}{\alpha-1} n z_{\beta}(\epsilon, \delta) - \frac{1}{\alpha-1}\rndBrk{\alpha\log\frac{1}{P_{{\rho}}(\Omega) - \epsilon_1}+ g_1(\epsilon_3, \epsilon_2)+ \frac{\alpha g_0(\epsilon_1)}{\beta-1}}.
        \numberthis
        \label{eq:approx_EAT_with_testing}
    \end{align*}
    where 
    \begin{align}
        & z_\beta(\epsilon, \delta) := \frac{\beta+1}{\beta-1}\log\rndBrk{\rndBrk{1+ \sqrt{(1-\delta)\epsilon}}^{\frac{\beta}{\beta+1}} + \rndBrk{\frac{\sqrt{(1-\delta)\epsilon}}{\delta^\beta}}^{{\frac{1}{\beta+1}}}} \\
        & K_\alpha := \frac{1}{6(2-\alpha)^3 \ln(2)} 2^{(\alpha-1)\rndBrk{2\log |A| + (\text{Max}(f) - \text{Min}_{\Sigma}(f))} }\ln^3\rndBrk{2^{\rndBrk{2\log |A| + (\text{Max}(f) - \text{Min}_{\Sigma}(f))}} + e^2}
    \end{align}
    and $g_1(x,y) = - \log(1- \sqrt{1-x^2}) - \log(1-y^2)$.
    \label{th:approx_EAT_with_testing}
\end{theorem}

\begin{proof}
    Just as in the proof of Theorem \ref{th:approx_EAT}, we define
    \begin{align}
        \cM^\delta_k := (1- \delta) \cM'_k + \delta \cM_k
    \end{align}
    for every $k$ and the state
    \begin{align}
        \sigma_{A_1^n B_1^n X_1^n E} := \cM^\delta_n \circ \cdots \circ \cM^\delta_1(\rho^{(0)}_{R_0 E}). \label{eq:aux_st_defn_wtest}
    \end{align}
    so that for $\beta>1$ and $\epsilon_1 >0$, we have
    \begin{align*}
        D^{\epsilon_1}_{\max}(\rho_{A_1^n B_1^n X_1^n E}||\sigma_{A_1^n B_1^n X_1^n  E})
        &\leq n z_\beta(\epsilon, \delta) + \frac{g_0(\epsilon_1)}{\beta-1}.\numberthis
    \end{align*}
    Define $d_{\beta} := n z_\beta(\epsilon, \delta) + \frac{g_0(\epsilon_1)}{\beta-1}$. The bound above implies that there exists a state $\tilde{\rho}_{A_1^n B_1^n X_1^n E}$, which is also classical on $X_1^n$ such that 
    \begin{align}
        P\rndBrk{\rho_{A_1^n B_1^n X_1^n E}, \tilde{\rho}_{A_1^n B_1^n X_1^n E}} \leq \epsilon_1
    \end{align}
    and 
    \begin{align}
        \tilde{\rho}_{A_1^n B_1^n X_1^n E} \leq 2^{d_\beta} \sigma_{A_1^n B_1^n X_1^n E}. 
        \label{eq:EAT_tilde_rho_op_ineq}
    \end{align}
    The registers $X_1^n$ for $\tilde{\rho}$ can be chosen to be classical, since the channel measuring $X_1^n$ only decreases the distance between $\tilde{\rho}$ and $\rho$, and the new state produced would also satisfy Eq. \ref{eq:EAT_tilde_rho_op_ineq}. As the registers $X_1^n$ are classical for both $\sigma$ and $\tilde{\rho}$, we can condition these states on the event $\Omega$. We will call the probability of the event $\Omega$ for the state $\sigma$ and $\tilde{\rho}$ $P_{\sigma}(\Omega)$ and $P_{\tilde{\rho}}(\Omega)$ respectively. Using Lemma \ref{lemm:dist_cond_states} and the Fuchs-van de Graaf inequality, we have 
    \begin{align}
        P\rndBrk{\rho_{A_1^n B_1^n X_1^n E|\Omega} , \tilde{\rho}_{A_1^n B_1^n X_1^n E|\Omega}} \leq 2 \sqrt{\frac{\epsilon_1}{P_{{\rho}}(\Omega)}}.
    \end{align}
    Conditioning Eq. \ref{eq:EAT_tilde_rho_op_ineq} on $\Omega$, we get 
    \begin{align}
        P_{\tilde{\rho}}(\Omega) \tilde{\rho}_{A_1^n B_1^n X_1^n E|\Omega} \leq 2^{d_\beta} P_{\sigma}(\Omega) \sigma_{A_1^n B_1^n X_1^n E|\Omega}. 
        \label{eq:EAT_tilde_rho_op_ineq_cond}
    \end{align}
    Together, the above two equations imply that 
    \begin{align}
        D^{\epsilon_2}_{\max}(\rho_{A_1^n B_1^n X_1^n E|\Omega} ||\sigma_{A_1^n B_1^n X_1^n E|\Omega}) \leq n z_\beta(\epsilon, \delta) + \frac{g_0(\epsilon_1)}{\beta-1} + \log\frac{P_{\sigma}(\Omega)}{P_{\tilde{\rho}}(\Omega)}
    \end{align}
    for $\epsilon_2 := 2\sqrt{\frac{\epsilon_1}{P_{{\rho}}(\Omega)}}$. \\

    For $\epsilon_3 >0$ and $\alpha \in (1, 2)$, we can plug the above in the bound provided by Lemma \ref{lemm:Hmin_rho_to_Halpha_sigma_using_Dmax} to get
    \begin{align}
        H_{\min}^{\epsilon_2+\epsilon_3}(A_1^n|B_1^n E)_{\rho_{|\Omega}} &\geq \tilde{H}^{\uparrow}_{\alpha}(A_1^n|B_1^n E)_{\sigma_{|\Omega}} - \frac{\alpha}{\alpha-1} n z_{\beta}(\epsilon, \delta) \nonumber \\
        & \qquad - \frac{1}{\alpha-1}\rndBrk{\alpha\log\frac{P_{\sigma}(\Omega)}{P_{\tilde{\rho}}(\Omega)}+ g_1(\epsilon_3, \epsilon_2)+ \frac{\alpha g_0(\epsilon_1)}{\beta-1}}.
        \label{eq:Hmin_rho_to_Halpha_sigma_wtest}
    \end{align}
    Now, note that using Eq. \ref{eq:test_maps} and \cite[Lemma B.7]{Dupuis20} we have 
    \begin{align}
        \tilde{H}^{\uparrow}_{\alpha}(A_1^n|B_1^n E)_{\sigma_{|\Omega}} = \tilde{H}^{\uparrow}_{\alpha}(A_1^n X_1^n|B_1^n E)_{\sigma_{|\Omega}}.
    \end{align}
    For every $k$, we introduce a register $D_k$ of dimension $|D_k| = \lceil 2^{\max(f)- \min(f)} \rceil$ and a channel $\mathcal{D}_k : X_k \rightarrow X_k D_k$ as 
    \begin{align}
        \mathcal{D}_k(\omega) := \sum_x \braket{x|\omega|x} \ket{x}\bra{x}\otimes \tau_x
    \end{align}
    where for every $x$, the state $\tau_x$ is a mixture between a uniform distribution on $\{1, 2, \cdots\ , \lfloor 2^{\max(f)- f(\delta_x)} \rfloor\}$ and a uniform distribution on $\{1, 2, \cdots\ , \lceil 2^{\max(f)- f(\delta_x)} \rceil\}$, so that
    \begin{align}
        H(D_k)_{\tau_x} = \max(f)- f(\delta_x)
    \end{align}
    where $\delta_x$ is the distribution with unit weight at element $x$. \\

    Define the channels $\bar{\mathcal{M}}_k := \mathcal{D}_k \circ \cM_k$, $\bar{\mathcal{M}}'_k := \mathcal{D}_k \circ \cM'_k$ and $\bar{\mathcal{M}}^\delta_k := \mathcal{D}_k \circ \cM^\delta_k = (1-\delta)\bar{\mathcal{M}}'_k + \delta \bar{\mathcal{M}}_k$ and the state 
    \begin{align}
        \bar{\sigma}_{A_1^n B_1^n X_1^n D_1^n E} := \bar{\mathcal{M}}^\delta_n \circ \cdots \bar{\mathcal{M}}^\delta_1(\rho^{(0)}_{R_0 E})
    \end{align}
    Note that $\bar{\sigma}_{A_1^n B_1^n X_1^n E} = {\sigma}_{A_1^n B_1^n X_1^n E}$. \cite[Lemma 4.5]{Metger22} implies that this satisfies
    \begin{align*}
        \tilde{H}^{\uparrow}_{\alpha}(A_1^n X_1^n|B_1^n E)_{{\sigma}_{|\Omega}} &= \tilde{H}^{\uparrow}_{\alpha}(A_1^n X_1^n|B_1^n E)_{\bar{\sigma}_{|\Omega}} \\
        &\geq \tilde{H}^{\uparrow}_{\alpha}(A_1^n X_1^n D_1^n |B_1^n E)_{\bar{\sigma}_{|\Omega}} - \max_{x_1^n \in \Omega} H_{\alpha}(D_1^n)_{\bar{\sigma}_{|x_1^n}}
        \numberthis
        \label{eq:Halpha_AX_to_AXD_bd}
    \end{align*}
    For $x_1^n \in \Omega$, we have 
    \begin{align*}
        H_{\alpha}(D_1^n)_{\bar{\sigma}_{|x_1^n}} &\leq H (D_1^n)_{\bar{\sigma}_{|x_1^n}} \\
        &\leq \sum_{k=1}^n H(D_k)_{\tau_{x_k}} \\
        &= \sum_{k=1}^n \max(f)- f(\delta_{x_k}) \\
        &= n \max(f) - n f(\text{freq}(x_1^n)) \\
        &\leq n \max(f) - n h.
        \numberthis
        \label{eq:Hapha_on_Omega}
    \end{align*}
    We can get rid of the conditioning on the right-hand side of Eq. \ref{eq:Halpha_AX_to_AXD_bd} by using \cite[Lemma B.5]{Dupuis20}
    \begin{align}
        \tilde{H}^{\uparrow}_{\alpha}(A_1^n X_1^n D_1^n |B_1^n E)_{\bar{\sigma}_{|\Omega}} \geq \tilde{H}^{\uparrow}_{\alpha}(A_1^n X_1^n D_1^n |B_1^n E)_{\bar{\sigma}} - \frac{\alpha}{\alpha-1}\log \frac{1}{P_{\sigma}(\Omega)}.
        \label{eq:rem_cond_Halpha_AXD_bd}
    \end{align}
    We now show that the channels $\bar{\mathcal{M}}'_k$ satisfy the second condition in Theorem \ref{th:approx_EAT}. For an arbitrary $k \in [n]$ and a sequence of channels $\cN_i \in \{\bar{\mathcal{M}}_i, \bar{\mathcal{M}}'_i\}$ for every $1\leq i < k$, let 
    \begin{align*}
        \eta_{A_1^k B_1^k X_1^k D_1^k E} = \bar{\mathcal{M}}'_k \circ \cN_{k-1} \cdots \circ \cN_1 (\rho^{(0)}_{R_0 E}).
    \end{align*}
    For this state, we have 
    \begin{align*}
        I(A_1^{k-1} D_1^{k-1}: B_k | B_1^{k-1} E)_{\eta} &= I(A_1^{k-1} : B_k | B_1^{k-1} E)_{\eta} + I(D_1^{k-1}: B_k | A_1^{k-1} B_1^{k-1} E)_{\eta} \\
        &= 0
    \end{align*}
    where $I(A_1^{k-1} : B_k | B_1^{k-1} E)_{\eta} = 0$ because of the condition in Eq. \ref{eq:approx_map_Markov_ch_wtest}, and $I(D_1^{k-1}: B_k | A_1^{k-1} B_1^{k-1} E)_{\eta} = 0$ since $X_1^{k-1}$ and hence $D_1^{k-1}$ are determined by $A_1^{k-1} B_1^{k-1}$. This implies that for this state $A_1^{k-1} D_1^{k-1} \leftrightarrow B_1^{k-1} E \leftrightarrow B_k$. Thus, the maps $\bar{\mathcal{M}}_k$ and $\bar{\mathcal{M}}'_k$ satisfy the conditions required for applying Theorem \ref{th:approx_EAT}. Specifically, we can use the bounds in Eq. \ref{eq:Halpha_sigma_to_theta} and \ref{eq:Halpha_theta_EAT_bd_in_Halpha} for bounding $\alpha$-conditional R\'enyi entropy in Eq. \ref{eq:rem_cond_Halpha_AXD_bd}
    \begin{align*}
        \tilde{H}&^{\uparrow}_{\alpha}(A_1^n X_1^n D_1^n |B_1^n E)_{\bar{\sigma}} \\
        &\geq \tilde{H}^{\uparrow}_{\alpha}(A_1^n D_1^n |B_1^n E)_{\bar{\sigma}} \\
        & \geq \sum_{k=1}^n \inf_{\omega_{R_{k-1} \tilde{R}_{k-1}}} \tilde{H}^{\downarrow}_{\alpha}(A_k D_k | B_k \tilde{R}_{k-1})_{\bar{\mathcal{M}}'_k(\omega)} - \frac{\alpha}{\alpha-1} n \log\rndBrk{1 + \delta \rndBrk{2^{\frac{\alpha-1}{\alpha}2\log(|A||D||B|)}-1}}.
        \numberthis
        \label{eq:Halpha_AXD_lb}
    \end{align*}
    The analysis in the proof of \cite[Proposition V.3]{Dupuis19} shows that the first term above can be bounded as 
    \begin{align}
        \inf_{\omega_{R_{k-1} \tilde{R}_{k-1}}} & \tilde{H}^{\downarrow}_{\alpha}(A_k D_k | B_k \tilde{R}_{k-1})_{\bar{\mathcal{M}}'_k(\omega)} \nonumber \\
        & \geq \text{Max}(f) - \frac{(\alpha-1)\ln(2)}{2}\rndBrk{\log(2 |A|^2 + 1) + \sqrt{2+ \text{Var}(f)}}^2 - (\alpha-1)^2 K_\alpha
        \label{eq:Halpha_single_term_bd}
    \end{align}
    Combining Eq. \ref{eq:Halpha_AX_to_AXD_bd}, \ref{eq:Hapha_on_Omega}, \ref{eq:rem_cond_Halpha_AXD_bd}, \ref{eq:Halpha_AXD_lb} and \ref{eq:Halpha_single_term_bd}, we have
    \begin{align*}
        \tilde{H}^{\uparrow}_{\alpha}(A_1^n X_1^n|B_1^n E)_{\bar{\sigma}_{|\Omega}} & \geq nh - \frac{(\alpha-1)\ln(2)}{2}\rndBrk{\log(2 |A|^2 + 1) + \sqrt{2+ \text{Var}(f)}}^2 - n(\alpha-1)^2 K_\alpha \\
        &\qquad \qquad - \frac{\alpha}{\alpha-1} n \log\rndBrk{1 + \delta \rndBrk{2^{\frac{\alpha-1}{\alpha}2\log(|A||D||B|)}-1}} - \frac{\alpha}{\alpha-1}\log \frac{1}{P_{\sigma}(\Omega)}.
        \numberthis
    \end{align*}
    Plugging this into Eq. \ref{eq:Hmin_rho_to_Halpha_sigma_wtest}, we get
    \begin{align*}
        H_{\min}^{\epsilon_2+\epsilon_3}(A_1^n|B_1^n E)_{\rho_{|\Omega}} &\geq nh - \frac{(\alpha-1)\ln(2)}{2}\rndBrk{\log(2 |A|^2 + 1) + \sqrt{2+ \text{Var}(f)}}^2 - n(\alpha-1)^2 K_\alpha \\
        &\qquad - \frac{\alpha}{\alpha-1} n \log\rndBrk{1 + \delta \rndBrk{4^{\frac{\alpha-1}{\alpha}(\log(|A||B|) + \text{max}(f) - \text{min}(f)+1)}-1}} \\
        & \qquad - \frac{\alpha}{\alpha-1} n z_{\beta}(\epsilon, \delta) - \frac{1}{\alpha-1}\rndBrk{\alpha\log\frac{1}{P_{{\rho}}(\Omega) - \epsilon_1}+ g_1(\epsilon_3, \epsilon_2)+ \frac{\alpha g_0(\epsilon_1)}{\beta-1}}.
        \numberthis
        \label{eq:approx_eat_test_final_bd}
    \end{align*}
    where we have used $P_{\tilde{\rho}}(\Omega)\geq P_{{\rho}}(\Omega) - \epsilon_1$ since $\frac{1}{2}\norm{\rho - \tilde{\rho}}_1 \leq P(\rho,\tilde{\rho})\leq \epsilon_1$. Note that the probability of $\Omega$ under the auxiliary state $\sigma$ cancels out.
\end{proof}

\subsection{Limitations and further improvements}
\label{sec:approx_EAT_lim_and_impr}

As we pointed out previously, the dependence of the entropy loss per round on $\epsilon$ is very poor (behaves as $\sim \epsilon^{1/24}$) in Theorem \ref{th:approx_EAT}. The classical version of this theorem has a much better dependence of $O(\sqrt{\epsilon})$ on $\epsilon$ (see Theorem \ref{th:approx_cl_EAT}). The reason for the poor performance of the quantum version is that our bound on the channel divergence (Lemma \ref{lemm:ch_div_bd}) is very weak compared to the bound we can use classically. It should be noted, however, that if Lemma \ref{lemm:ch_div_bd} were to be improved in the future, one could simply plug the new bound into our proof and derive an improvement for Theorem \ref{th:approx_EAT}. \\

A better bound on the channel divergence would have an additional benefit. It could simplify the proof and the Markov chain assumption in our theorem. In particular, it would be much easier to carry out the proof if the mixed channels $\cM^{\delta}_k$ were defined as $(1-\delta)\cM'_k + \delta \tau_{A_k B_k} \otimes \tr_{A_k B_k}\circ\cM_k$ (which is what is done classically), where $\tau_{A_k B_k}$ is the completely mixed state on registers $A_k B_k$. Here, instead of mixing the channel $\cM'_k$ with $\cM_k$, we mix it with $\tau_{A_k B_k} \otimes \tr_{A_k B_k}\circ\cM_k$, which also keeps $D_{\max}(\cM_k || \cM^{\delta}_k)$ small enough. Moreover, this definition ensures that the registers $B_k$ produced by the map $\cM_k^\delta$ always satisfy the Markov chain conditions. If it were possible to show that the divergence between the real state $\cM_n \circ \cdots \circ \cM_1 (\rho^{(0)}_{R_0 E})$ and the auxiliary state $\cM^\delta_n \circ \cdots \circ \cM^\delta_1 (\rho^{(0)}_{R_0 E})$ is small for this definition of $\cM_k^\delta$, then one could directly use the entropy accumulation theorem for lower bounding the entropy for the auxiliary state. We cannot do this in our proof as this definition of the mixed channel $\cM^\delta_k$ also increases the distance from the original channel $\cM_k$ to $\epsilon+2\delta$ and this makes the upper bound in Lemma \ref{lemm:Dsharp_close_st_bd} large (finite even in the limit $\epsilon \rightarrow 0$).\\

It seems that it should be possible to weaken the assumptions for approximate entropy accumulation. The classical equivalent of this theorem (Theorem \ref{th:approx_cl_EAT}) for instance can be proven very easily and requires a much weaker approximation assumption. It would be interesting if one could remove the ``memory'' registers $R_k$ from the assumptions required for approximate entropy accumulation, since these are not typically accessible to the users in applications. \\

Another troubling feature of the approximate entropy accumulation theorem seems to be that it assumes that the size of the side information registers $B_k$ is constant. One might wonder if this is necessary, since continuity bounds like the Alicki-Fannes-Winter (AFW) inequality do not depend on the size of the side information. It turns out that a bound on the side information size is indeed necessary in this case. We show a simple classical example to demonstrate this in Appendix \ref{sec:size_of_B_approx_EAT}. The necessity of such a bound also rules out a similar approximate extension of the generalised entropy accumulation theorem (GEAT). \\

%% file: Appendix.tex
\appendix

\addcontentsline{toc}{section}{APPENDICES}
\section*{APPENDICES}

\section{Entropic triangle inequalities cannot be improved much}

In this section, we will construct a classical counterexample to show that it is not possible to improve Lemma \ref{lemm:Hmin_rho_to_Halpha_sigma_using_Dmax} to get a result like
\begin{align}
    H^{\epsilon'}_{\min}(A|B)_\rho \geq H^{\epsilon}_{\min}(A|B)_\eta - O(D^{\epsilon''}_{\max}(\rho || \eta))
    \label{eq:better_smooth_min_tri_ineq}
\end{align}
where $\epsilon, \epsilon' >0$ and the constant in front of $D^{\epsilon''}_{\max}(\rho || \eta)$ is independent of the dimensions $|A|$ and $|B|$. \\

Consider the probability distribution $p_{AB}$ where $B$ is chosen to be equal to $1$ with probability $1- \epsilon$ and $0$ with probability $\epsilon$, and $A_1^n$ is chosen to be a random $n$-bit string if $B=1$ otherwise $A_1^n$ is chosen to be the all $0$ string. Let $E$ be the event that $B=0$. Then, we have 
\begin{align*}
    p_{AB|E} \leq \frac{1}{p(E)} p_{AB} = \frac{1}{\epsilon} p_{AB}
\end{align*}
or equivalently $D_{\max}(p_{AB|E} || p_{AB}) \leq \log \frac{1}{\epsilon}$. In this case, we have $H^{\epsilon}_{\min}(A|B)_{p} = n$ (where we are smoothing in the trace distance) and $H^{\epsilon'}_{\min}(A|B)_{p_{|E}} = \log \frac{1}{1-\epsilon'} = O(1)$ (independent of $n$). If Eq. \ref{eq:better_smooth_min_tri_ineq}, were true then we would have 
\begin{align*}
    n - O\rndBrk{\log \frac{1}{\epsilon}} & \leq H^{\epsilon}_{\min}(A|B)_{P} - O(D^{\epsilon''}_{\max}(p_{AB|E} || p_{AB})) \\
    & \leq H^{\epsilon'}_{\min}(A|B)_{p_{|E}} =O(1)
\end{align*}
which would lead to a contradiction because $n$ is a free parameter and we can let $n \rightarrow \infty$. \\

The same example can be used to show that it is not possible to improve Corollary \ref{cor:H_rho_to_Halpha_sigma_using_D} to an equation of the form
\begin{align*}
    H(A|B)_{\rho} \geq H(A|B)_{\eta} - O(D(\rho || \eta)).
\end{align*}
For $\rho = P_{|E}$ and $\eta = P$, such a bound would imply that
\begin{align*}
    0 \geq (1-\epsilon)n - \log \frac{1}{\epsilon}
\end{align*}
which is not true for large $n$. 

\section{Bounds for $D^{\#}_{\alpha}$ of the form in Lemma \ref{lemm:Dsharp_close_st_bd} necessarily diverge in the limit $\alpha=1$}
\label{sec:better_Dsharp_bd_disc}

Classically, we have the following bound for R\'enyi entropies. 
\begin{lemma}
    Suppose $\epsilon \in (0,1]$, $d \geq \epsilon^{1/2}$, and $p$ and $q$ are two distributions over an alphabet $\mathcal{X}$ such that $\frac{1}{2} \norm{p-q}_1 \leq \epsilon$ and $D_{\max}(p||q) \leq d < \infty$, for $\alpha>1$ we have 
    \begin{align}
        D_{\alpha}(p||q) \leq \frac{1}{\alpha-1} \log\rndBrk{(1+\sqrt{\epsilon})^{\alpha-1}(1-2\sqrt{\epsilon}) + 2^{d(\alpha-1)+1}\sqrt{\epsilon}}.
        \label{eq:cl_Dalpha_close_st_bd}
    \end{align}
    In the limit, $\alpha \rightarrow 1$, we get the bound 
    \begin{align}
        D(p||q) \leq (1-2\sqrt{\epsilon}) \log(1+ \sqrt{\epsilon}) + 2\sqrt{\epsilon} d.
        \label{eq:cl_D_close_st_bd}
    \end{align}
\end{lemma}
\begin{proof}
    Classically, we have that the set $S:= \{x \in \mathcal{X}: p(x) \leq (1+\sqrt{\epsilon}) q(x) \}$ is such that $p(S) \geq 1- 2\sqrt{\epsilon}$ using Lemma \ref{lemm:small_tr_norm_implies_small_Dmax}. Thus, for $\alpha>1$ we have 
    \begin{align*}
        \sum_{x \in \mathcal{X}} p(x) \rndBrk{\frac{p(x)}{q(x)}}^{\alpha-1} &= 
        \sum_{x \in S}p(x) \rndBrk{\frac{p(x)}{q(x)}}^{\alpha-1} + \sum_{x \not\in S}p(x) \rndBrk{\frac{p(x)}{q(x)}}^{\alpha-1} \\
        &\leq \sum_{x \in S} (1+ \sqrt{\epsilon})^{\alpha-1} p(x) + \sum_{x \not\in S} 2^{d(\alpha-1)}p(x) \\
        &=(1+ \sqrt{\epsilon})^{\alpha-1} p(S) + 2^{d(\alpha-1)}p(S^c) \\
        &\leq (1+ \sqrt{\epsilon})^{\alpha-1} (1- 2\sqrt{\epsilon}) +  2^{d(\alpha-1)+1}\sqrt{\epsilon} 
    \end{align*}
    where in the second line we used the definition of set $S$ and the fact that $D_{\max}(p||q) \leq d$, in the last line we use the fact that since $d \geq \sqrt{\epsilon} \geq \log(1+ \sqrt{\epsilon})$, the convex sum is maximised for the largest possible value of $p(S^c)$, which is $2\sqrt{\epsilon}$. The bound now follows. 
\end{proof}

We observed in Sec. \ref{sec:div_bd_approx_eq_st} that the bound in Lemma \ref{lemm:Dsharp_close_st_bd} for $D^{\#}_{\alpha}$ tends to $\infty$ as $\alpha \rightarrow 1$ for a fixed $\epsilon>0$. One may wonder if a bound like Eq. \ref{eq:cl_D_close_st_bd} exists for $\lim_{\alpha \rightarrow 1} D^\#_{\alpha}(\rho || \sigma) = \hat{D}(\rho || \sigma)$ \cite{Bergh21}. We show in the following that such a bound is not possible. \\

Suppose, that for all $\epsilon \in [0,a)$ (a small neighborhood of $0$), $1 \leq d < \infty$, states $\rho$ and $\sigma$, which satisfy $\frac{1}{2}\norm{\rho - \sigma}_1 \leq \epsilon$ and $\rho \leq 2^d \sigma$, the following bound holds
\begin{align}
    \hat{D}(\rho || \sigma) \leq f(\epsilon, d)
    \label{eq:supposed_hatD_bd}
\end{align}
where $f(\epsilon, d)$ is such that $\lim_{\epsilon \rightarrow 0} f(\epsilon, d) = f(0,d)=0$ for every $1\leq d <\infty$. Note that the upper bound in Eq. \ref{eq:cl_D_close_st_bd} is of this form. It is known that for pure states $\rho$, $\hat{D}(\rho || \sigma) = D_{\max}(\rho || \sigma)$. We will use this to construct a contradiction. 
\begin{lemma}\footnote{This Lemma was pointed out to us by Omar Fawzi.}
    For a pure state $\rho = \ket{\rho}\bra{\rho}$ and a state $\sigma$, we have 
    \begin{align*}
        \hat{D}(\rho || \sigma) = D_{\max}(\rho || \sigma) = \bra{\rho} \sigma^{-1} \ket{\rho}.
    \end{align*}
    \label{lemm:hatD_pure_st}
\end{lemma}
\begin{proof}
    First, we can evaluate $\hat{D}$ as 
    \begin{align*}
        \hat{D}(\rho || \sigma) &= \tr\rndBrk{\rho \log \rndBrk{\rho^{\frac{1}{2}} \sigma^{-1} \rho^{\frac{1}{2}}}} \\
        &= \tr\rndBrk{\ket{\rho}\bra{\rho} \log \rndBrk{\ket{\rho}\bra{\rho} \sigma^{-1} \ket{\rho}\bra{\rho}}} \\
        &= \tr\rndBrk{\ket{\rho}\bra{\rho} \log (\bra{\rho} \sigma^{-1} \ket{\rho}) \ket{\rho}\bra{\rho}}\\
        &= \log \bra{\rho} \sigma^{-1} \ket{\rho}.
    \end{align*}
    Next, we have that 
    \begin{align*}
        D_{\max}(\rho || \sigma) &= \log \norm{\sigma^{-\frac{1}{2}} \rho \sigma^{-\frac{1}{2}}}_{\infty} \\
        &= \log \norm{\sigma^{-\frac{1}{2}} \ket{\rho}\bra{\rho} \sigma^{-\frac{1}{2}}}_{\infty} \\
        &= \log \tr \rndBrk{\sigma^{-\frac{1}{2}} \ket{\rho}\bra{\rho} \sigma^{-\frac{1}{2}}}\\
        &= \log \bra{\rho} \sigma^{-1} \ket{\rho}.
    \end{align*}
\end{proof}
To obtain a contradiction, let $\epsilon \in [0, a^2)$. Define the states 
\begin{align*}
    \rho &:= \ket{0}\bra{0} = \begin{pmatrix}
        1 \ &0 \\
        0 \ &0
    \end{pmatrix}\\
    \sigma'_\epsilon &:= (\sqrt{1- \epsilon} \ket{0} + \sqrt{\epsilon}\ket{1}) (\sqrt{1- \epsilon} \ket{0} + \sqrt{\epsilon}\ket{1})^\dagger \\
    &= \begin{pmatrix}
        1 - \epsilon \ & \sqrt{\epsilon(1- \epsilon)} \\
        \sqrt{\epsilon(1- \epsilon)} \ & \epsilon
    \end{pmatrix}\\
    \sigma_\epsilon &:= (1- \delta) \sigma'_\epsilon + \delta \rho \\
    &= \begin{pmatrix}
        (1 - \epsilon)(1- \delta) + \delta \ & (1- \delta)\sqrt{\epsilon(1- \epsilon)} \\
        (1- \delta)\sqrt{\epsilon(1- \epsilon)} \ & (1- \delta)\epsilon
    \end{pmatrix}
\end{align*}
where $\{\ket{0}, \ket{1}\}$ is the standard basis and $\delta \in (0,1)$ is a parameter, which will be chosen later. Observe that $F(\rho, \sigma_\epsilon) = {\langle e_0, \sigma_\epsilon e_0 \rangle}= {1 - \epsilon (1-\delta)}$, which implies that $\frac{1}{2}\norm{\rho- \sigma_\epsilon}_1 \leq \sqrt{\epsilon} \in [0,a)$. For these definitions, we have
\begin{align*}
    \sigma_\epsilon^{-1} = \frac{1}{(1-\delta)\delta\epsilon}\begin{pmatrix}
        (1- \delta)\epsilon \ & -(1- \delta)\sqrt{\epsilon(1- \epsilon)} \\
        -(1- \delta)\sqrt{\epsilon(1- \epsilon)} \ & (1 - \epsilon)(1- \delta) + \delta
    \end{pmatrix}
\end{align*}
which implies that $\hat{D}(\rho || \sigma_\epsilon) = \log\frac{1}{\delta}$ using Lemma \ref{lemm:hatD_pure_st}. We can fix $\delta = \frac{1}{10}$. Note that $\hat{D}(\rho || \sigma_\epsilon) >0$ is independent of $\epsilon$. Now observe that if the bound in Eq. \ref{eq:supposed_hatD_bd} were true, then as $\epsilon \rightarrow 0$, $\hat{D}(\rho || \sigma_\epsilon) = \log(10) \rightarrow 0$, which leads us to a contradiction. Thus, we cannot have bounds of the form in Eq. \ref{eq:supposed_hatD_bd} (also see \cite{Bluhm22}). Consequently, any kind of bound on $\hat{D}_{\alpha}$ or $D^\#_\alpha$ which results in a bound of the form in Eq. \ref{eq:supposed_hatD_bd} as $\alpha \rightarrow 1$, for example, the bound in Eq. \ref{eq:cl_Dalpha_close_st_bd}, is also not possible at least close to $\alpha=1$.\\

It should be noted that the reason we can have bounds of the form in Lemma \ref{lemm:Dsharp_close_st_bd}, despite the fact that no good bound on $\hat{D} = \lim_{\alpha \rightarrow 1} D^\#_\alpha$ can be produced is that $D^\#_\alpha$, unlike the conventional generalizations of the R\'enyi divergence, is \textbf{not monotone} in $\alpha$ \cite[Remark 3.3]{Fawzi21}(otherwise the above counterexample would also give a no-go argument for $D^\#_\alpha$).

\section{Transforming lemmas for EAT from $\tilde{H}^{\downarrow}_\alpha$ to $\tilde{H}^{\uparrow}_\alpha$}

We have to redo the Lemmas used in \cite{Dupuis20} using $\tilde{H}^{\uparrow}_\alpha$ because we were only able to prove the dimension bound we need ($\tilde{H}^{\uparrow}_\alpha (A|BC)\geq \tilde{H}^{\uparrow}_\alpha(A|B) - 2\log|C|$) in terms of $\tilde{H}^{\uparrow}_\alpha$ 

\begin{lemma}[{\cite[Lemma 3.1]{Dupuis20}}]
    For $\rho_{A_1 A_2 B}$ and $\sigma_B$ be states and $\alpha \in (0, \infty)$, we have the chain rule
    \begin{align}
        \tilde{D}_\alpha (\rho_{A_1 B} || \Id_{A_1} \otimes \sigma_B) - \tilde{D}_\alpha (\rho_{A_1 A_2 B} || \Id_{A_1 A_2} \otimes \sigma_B)= \tilde{H}_\alpha^{\downarrow}(A_2 | A_1 B)_{\nu}
    \end{align}
    where the state $\nu_{A_1 A_2 B}$ is defined as 
    \begin{align*}
        & \nu_{A_1 B} := \frac{\rndBrk{\rho_{A_1 B}^{\frac{1}{2}} \sigma_B^{-\alpha'} \rho_{A_1 B}^{\frac{1}{2}}}^\alpha}{\tr\rndBrk{\rho_{A_1 B}^{\frac{1}{2}} \sigma_B^{-\alpha'} \rho_{A_1 B}^{\frac{1}{2}}}^\alpha}\\
        & \nu_{A_1 A_2 B} := \nu_{A_1 B}^{\frac{1}{2}} \rho_{A_2 | A_1 B} \nu_{A_1 B}^{\frac{1}{2}}
    \end{align*}
    and $\alpha':= \frac{\alpha-1}{\alpha}$.
    \label{lemm:EAT_ch_rule}
\end{lemma}
\begin{corollary}[Chain rule for $ \tilde{H}^{\downarrow}_\alpha$ {\cite[Theorem 3.2]{Dupuis20}}]
    For $\alpha \in (0, \infty)$, a state $\rho_{A_1 A_2 B}$, we have the chain rule
    \begin{align}
        \tilde{H}^{\downarrow}_\alpha (A_1 A_2 | B)_\rho = \tilde{H}^{\downarrow}_\alpha (A_1| B)_\rho + \tilde{H}^{\downarrow}_\alpha (A_2 | A_1 B)_\nu
    \end{align}
    where the state $\nu_{A_1 A_2 B}$ is defined as 
    \begin{align*}
        & \nu_{A_1 B} := \frac{\rndBrk{\rho_{A_1 B}^{\frac{1}{2}} \rho_B^{-\alpha'} \rho_{A_1 B}^{\frac{1}{2}}}^\alpha}{\tr\rndBrk{\rho_{A_1 B}^{\frac{1}{2}} \rho_B^{-\alpha'} \rho_{A_1 B}^{\frac{1}{2}}}^\alpha}\\
        & \nu_{A_1 A_2 B} := \nu_{A_1 B}^{\frac{1}{2}} \rho_{A_2 | A_1 B} \nu_{A_1 B}^{\frac{1}{2}}
    \end{align*}
    and $\alpha':= \frac{\alpha-1}{\alpha}$.
    \label{cor:Halpha_ch_rule}
\end{corollary}
We can modify {\cite[Theorem 3.2]{Dupuis20}}, which is in terms of $\tilde{H}^{\downarrow}_\alpha$, to the following, which is a chain rule in terms of $\tilde{H}^{\uparrow}_\alpha$. The chain rule in this Corollary was also observed in \cite{Dupuis20}.
\begin{corollary}[Chain rule for $ \tilde{H}^{\uparrow}_\alpha$]
    For $\alpha \in (0, \infty)$, a state $\rho_{A_1 A_2 B}$ and for any state $\sigma_B$ such that $\tilde{H}^{\uparrow}_\alpha (A_1| B)_\rho = -\tilde{D}_\alpha (\rho_{A_1 B} || \Id_{A_1} \otimes \sigma_B)$, we have 
    \begin{align}
        \tilde{H}^{\uparrow}_\alpha (A_1 A_2 | B)_\rho \geq \tilde{H}^{\uparrow}_\alpha (A_1| B)_\rho + \tilde{H}^{\downarrow}_\alpha (A_2 | A_1 B)_\nu
    \end{align}
    where the state $\nu_{A_1 A_2 B}$ is defined as 
    \begin{align*}
        & \nu_{A_1 B} := \frac{\rndBrk{\rho_{A_1 B}^{\frac{1}{2}} \sigma_B^{-\alpha'} \rho_{A_1 B}^{\frac{1}{2}}}^\alpha}{\tr\rndBrk{\rho_{A_1 B}^{\frac{1}{2}} \sigma_B^{-\alpha'} \rho_{A_1 B}^{\frac{1}{2}}}^\alpha}\\
        & \nu_{A_1 A_2 B} := \nu_{A_1 B}^{\frac{1}{2}} \rho_{A_2 | A_1 B} \nu_{A_1 B}^{\frac{1}{2}}
    \end{align*}
    and $\alpha':= \frac{\alpha-1}{\alpha}$. For $\alpha \in (0, \infty)$, state $\rho_{A_1 A_2 B}$ and any state $\sigma_B$ such that $\tilde{H}^{\uparrow}_\alpha (A_1 A_2| B)_\rho = -\tilde{D}_\alpha (\rho_{A_1 A_2 B} || \Id_{A_1 A_2} \otimes \sigma_B)$, we have
    \begin{align}
        \tilde{H}^{\uparrow}_\alpha (A_1 A_2 | B)_\rho \leq \tilde{H}^{\uparrow}_\alpha (A_1| B)_\rho + \tilde{H}^{\downarrow}_\alpha (A_2 | A_1 B)_\nu
    \end{align}
    where the state $\nu_{A_1 A_2 B}$ is defined the same as above.
    \label{cor:opt_Halpha_ch_rule}
\end{corollary}
\begin{proof}
    Let $\sigma_B$ be a state such that $\tilde{H}^{\uparrow}_\alpha (A_1| B)_\rho = -\tilde{D}_\alpha (\rho_{A_1 B} || \Id \otimes \sigma_B)$. Then, using Lemma \ref{lemm:EAT_ch_rule}, we have
    \begin{align*}
        \tilde{H}^{\uparrow}_\alpha (A_1 A_2 | B)_\rho &\geq - \tilde{D}_\alpha (\rho_{A_1 A_2 B} || \Id_{A_1 A_2} \otimes \sigma_B)\\
        &= - \tilde{D}_\alpha (\rho_{A_1 B} || \Id_{A_1} \otimes \sigma_B) + \tilde{H}_\alpha^{\downarrow}(A_2 | A_1 B)_{\nu} \\
        &= \tilde{H}^{\uparrow}_\alpha (A_1| B)_\rho + \tilde{H}^{\downarrow}_\alpha (A_2 | A_1 B)_\nu
    \end{align*}
    for $\nu_{A_1 A_2 B}$ defined as in the Lemma. Similarly, if $\tilde{H}^{\uparrow}_\alpha (A_1 A_2| B)_\rho = -\tilde{D}_\alpha (\rho_{A_1 A_2 B} || \Id_{A_1 A_2} \otimes \sigma_B)$, then
    \begin{align*}
        \tilde{H}^{\uparrow}_\alpha (A_1 A_2 | B)_\rho &= - \tilde{D}_\alpha (\rho_{A_1 A_2 B} || \Id_{A_1 A_2} \otimes \sigma_B)\\
        &= - \tilde{D}_\alpha (\rho_{A_1 B} || \Id_{A_1} \otimes \sigma_B) + \tilde{H}_\alpha^{\downarrow}(A_2 | A_1 B)_{\nu} \\
        &\leq \tilde{H}^{\uparrow}_\alpha (A_1| B)_\rho + \tilde{H}^{\downarrow}_\alpha (A_2 | A_1 B)_\nu
    \end{align*}
    for $\nu_{A_1 A_2 B}$ defined as in the Lemma.
\end{proof}
We transform \cite[Theorem 3.3]{Dupuis20} to a statement about $\tilde{H}^{\uparrow}_\alpha$ in the following.
\begin{lemma}
    Let $\alpha \in \left[\frac{1}{2}, \infty\right)$ and $\rho_{A_1 A_2 B_1 B_2}$ be a state which satisfies the Markov chain $A_1 \leftrightarrow B_1 \leftrightarrow B_2$. Then, we have 
    \begin{align}
        \tilde{H}^{\uparrow}_\alpha (A_1 A_2 | B_1 B_2)_\rho \geq \tilde{H}^{\uparrow}_\alpha (A_1| B_1)_\rho + \inf_{\nu} \tilde{H}^{\downarrow}_\alpha (A_2 | A_1 B_1 B_2)_\nu
    \end{align}
    where the infimum is taken over all states $\nu_{A_1 A_2 B_1 B_2}$ such that $\nu_{A_2 B_2 | A_1 B_1} = \rho_{A_2 B_2 | A_1 B_1}$.
\end{lemma}
\begin{proof}
    Since, $\rho$ satisfies the Markov chain $A_1 \leftrightarrow B_1 \leftrightarrow B_2$, there exists a decomposition of the system $B_1$ as \cite[Theorem 5.4]{Sutter18}
    \begin{align*}
        B_1 = \bigoplus_{j \in J} a_j \otimes c_j
    \end{align*}
    such that 
    \begin{align}
        \rho_{A_1 B_1 B_2} = \bigoplus_{j \in J} p(j) \rho_{A_1 a_j} \otimes \rho_{c_j B_2}.
    \end{align}
    Let $J' \subseteq J$ be the set $\{j \in J : p(j)>0\}$. Note, that we can replace $J$ by $J'$ in the above equation.\\
    % CPTP because we are measuring and appending a state and then tracing over c_j which is a space determined by the measurement- just store measurement result as j
    We can define the CPTP recovery map $\mathcal{R}_{B_1 \rightarrow B_1 B_2}$ for $\rho_{A_1 B_1 B_2}$ as
    \begin{align}
        \mathcal{R}_{B_1 \rightarrow B_1 B_2} (X) := \bigoplus_{j \in J} \tr_{c_j} \rndBrk{\Pi_{a_j} \otimes \Pi_{c_j} X \Pi_{a_j} \otimes \Pi_{c_j}} \otimes \rho_{c_j B_2}
    \end{align}
    where $\Pi_{a_j} \otimes \Pi_{c_j}$ is the projector on the subspace $a_j \otimes c_j$. This recovery channel satisfies
    \begin{align}
        \mathcal{R}_{B_1 \rightarrow B_1 B_2} (\rho_{A_1 B_1}) = \rho_{A_1 B_1 B_2}.
        \label{eq:rho_recovery}
    \end{align}
    We can now show that the optimisation for the conditional entropy $\tilde{H}^{\uparrow}_\alpha (A_1| B_1 B_2)_\rho$ can be restricted to states of the form $\mathcal{R}_{B_1 \rightarrow B_1 B_2} \rndBrk{\sigma_{B_1}}$. This follows as
    \begin{align*}
        \tilde{H}^{\uparrow}_\alpha (A_1| B_1 B_2)_\rho &= \sup_{\sigma_{B_1 B_2}} -\tilde{D}_\alpha (\rho_{A_1 B_1 B_2} || \Id_{A_1}  \otimes \sigma_{B_1 B_2}) \\
        &\leq \sup_{\sigma_{B_1 B_2}} -\tilde{D}_\alpha (\mathcal{R}_{B_1 \rightarrow B_1 B_2} \circ \tr_{B_2} \rndBrk{\rho_{A_1 B_1 B_2}} || \mathcal{R}_{B_1 \rightarrow B_1 B_2} \circ \tr_{B_2} \rndBrk{\Id_{A_1} \otimes \sigma_{B_1 B_2}})\\
        &= \sup_{\sigma_{B_1}} -\tilde{D}_\alpha (\rho_{A_1 B_1 B_2} || \Id_{A_1} \otimes \mathcal{R}_{B_1 \rightarrow B_1 B_2} \rndBrk{\sigma_{B_1}})\\
        &\leq \sup_{\sigma_{B_1 B_2}} -\tilde{D}_\alpha (\rho_{A_1 B_1 B_2} || \Id_{A_1}  \otimes \sigma_{B_1 B_2})\\
        &= \tilde{H}^{\uparrow}_\alpha (A_1| B_1 B_2)_\rho
    \end{align*}
    where the second line follows from the data processing inequality for $\tilde{D}_\alpha$ for $\alpha \geq \frac{1}{2}$, the supremum in the fourth line is over all states on the registers $B_1 B_2$,and the last line simply follows from the definition of $\tilde{H}^{\uparrow}_\alpha (A_1| B_1 B_2)_\rho$. As a result, it follows that 
    \begin{align}
        \tilde{H}^{\uparrow}_\alpha (A_1| B_1 B_2)_\rho = \sup_{\sigma_{B_1}} -\tilde{D}_\alpha (\rho_{A_1 B_1 B_2} || \Id_{A_1} \otimes \mathcal{R}_{B_1 \rightarrow B_1 B_2} \rndBrk{\sigma_{B_1}})
    \end{align}
    Let $\sigma_{B_1 B_2} = \mathcal{R}_{B_1 \rightarrow B_1 B_2} \rndBrk{\eta_{B_1}}$ be such that $\tilde{H}^{\uparrow}_\alpha (A_1| B_1 B_2)_\rho = -\tilde{D}_\alpha (\rho_{A_1 B_1 B_2} || \Id_{A_1}  \otimes \sigma_{B_1 B_2})$. Using Corollary \ref{cor:opt_Halpha_ch_rule}, for this choice of $\sigma_{B_1 B_2}$, we have that 
    \begin{align}
        \tilde{H}^{\uparrow}_\alpha (A_1 A_2 | B_1 B_2)_\rho \geq \tilde{H}^{\uparrow}_\alpha (A_1| B_1 B_2)_\rho + \tilde{H}^{\downarrow}_\alpha (A_2 | A_1 B_1 B_2)_\nu
        \label{eq:Halpha_ch_rule_for_Mk_ch}
    \end{align}
    where the state $\nu_{A_1 A_2 B_1 B_2}$ is defined as 
    \begin{align*}
        & \nu_{A_1 B_1 B_2} := \frac{\rndBrk{\rho_{A_1 B_1 B_2}^{\frac{1}{2}} \sigma_{B_1 B_2}^{-\alpha'} \rho_{A_1 B_1 B_2}^{\frac{1}{2}}}^\alpha}{\tr\rndBrk{\rho_{A_1 B_1 B_2}^{\frac{1}{2}} \sigma_{B_1 B_2}^{-\alpha'} \rho_{A_1 B_1 B_2}^{\frac{1}{2}}}^\alpha}\\
        & \nu_{A_1 A_2 B_1 B_2} := \nu_{A_1 B_1 B_2}^{\frac{1}{2}} \rho_{A_2 | A_1 B_1 B_2} \nu_{A_1 B_1 B_2}^{\frac{1}{2}}.
    \end{align*}
    We will now show that $\nu_{A_2 B_2 | A_1 B_1} = \rho_{A_2 B_2 | A_1 B_1}$. For this it is sufficient to show that 
    \begin{align*}
        \nu_{A_1 B_1}^{-\frac{1}{2}} \nu_{A_1 B_1 B_2 }^{\frac{1}{2}}= \rho_{A_1 B_1}^{-\frac{1}{2}} \rho_{A_1 B_1 B_2 }^{\frac{1}{2}}.
    \end{align*}
    We have that 
    \begin{align*}
        \sigma_{B_1 B_2} &= \mathcal{R}_{B_1 \rightarrow B_1 B_2} \rndBrk{\eta_{B_1}} \\
        &= \bigoplus_{j \in J} \tr_{c_j} \rndBrk{\Pi_{a_j} \otimes \Pi_{c_j} \eta_{B_1} \Pi_{a_j} \otimes \Pi_{c_j}} \otimes \rho_{c_j B_2}\\
        &= \bigoplus_{j \in J} q(j) \omega_{a_j} \otimes \rho_{c_j B_2}
    \end{align*}
    where we have defined the probability distribution $q(j):= \tr(\Pi_{a_j} \otimes \Pi_{c_j} \eta_{B_1})$ and states $\omega_{a_j} = \frac{1}{q(j)} \Pi_{a_j} \tr_{c_j} \rndBrk{\Pi_{c_j} \eta_{B_1} \Pi_{c_j}} \Pi_{a_j}$ for every $j \in J$. \\
    
    Since $\tilde{D}_\alpha (\rho_{A_1 B_1 B_2} || \Id_{A_1}  \otimes \sigma_{B_1 B_2})=-\tilde{H}^{\uparrow}_\alpha (A_1| B_1 B_2)_\rho \leq \log |A_1| < \infty$, we have that
    \begin{align*}
        & \rho_{A_1 B_1 B_2} \ll \Id_{A_1} \otimes \sigma_{B_1 B_2}\\
        \Rightarrow & \bigoplus_{j \in J'} p(j) \rho_{A_1 a_j} \otimes \rho_{c_j B_2} \ll \Id_{A_1} \otimes \bigoplus_{j \in J} q(j) \omega_{a_j} \otimes \rho_{c_j B_2} \\
        \Rightarrow & \text{for every } j \in J': \rho_{A_1 a_j} \ll \Id_{A_1} \otimes \omega_{a_j} \text{ and } q(j)>0. \numberthis
        \label{eq:omega_j_ll_cond}
    \end{align*}
    
    This decomposition can be used to evaluate $\nu_{A_1 B_1 B_2}$ as follows
    \begin{align*}
        \nu_{A_1 B_1 B_2} &= \frac{1}{N} \rndBrk{\rho_{A_1 B_1 B_2}^{\frac{1}{2}} \sigma_{B_1 B_2}^{-\alpha'} \rho_{A_1 B_1 B_2}^{\frac{1}{2}}}^\alpha \\
        &= \frac{1}{N} \rndBrk{\bigoplus_{j \in J'} p(j)^{\frac{1}{2}} \rho_{A_1 a_j}^{\frac{1}{2}} \otimes \rho_{c_j B_2}^{\frac{1}{2}} \bigoplus_{j \in J} q(j)^{-\alpha'} \omega^{-\alpha'}_{a_j} \otimes \rho^{-\alpha'}_{c_j B_2} \bigoplus_{j \in J'} p(j)^{\frac{1}{2}} \rho_{A_1 a_j}^{\frac{1}{2}} \otimes \rho_{c_j B_2}^{\frac{1}{2}}}^\alpha \\
        &= \frac{1}{N} \rndBrk{\bigoplus_{j \in J'} p(j) q(j)^{-\alpha'} \rho_{A_1 a_j}^{\frac{1}{2}} \omega^{-\alpha'}_{a_j} \rho_{A_1 a_j}^{\frac{1}{2}} \otimes \rho^{1-\alpha'}_{c_j B_2}}^\alpha \\
        &= \frac{1}{N} \bigoplus_{j \in J'} p(j)^\alpha q(j)^{1-\alpha} \rndBrk{\rho_{A_1 a_j}^{\frac{1}{2}} \omega^{-\alpha'}_{a_j} \rho_{A_1 a_j}^{\frac{1}{2}}}^\alpha \otimes \rho_{c_j B_2}
    \end{align*}
    for $N:= \tr\rndBrk{\rho_{A_1 B_1 B_2}^{\frac{1}{2}} \sigma_{B_1 B_2}^{-\alpha'} \rho_{A_1 B_1 B_2}^{\frac{1}{2}}}^\alpha$. Further, we have
    \begin{align*}
        \nu_{A_1 B_1}^{-\frac{1}{2}}& \nu_{A_1 B_1 B_2 }^{\frac{1}{2}} \\
        &= \frac{1}{N^{-\frac{1}{2}}} \bigoplus_{j \in J'} p(j)^{-\frac{\alpha}{2}}q(j)^{-\frac{1-\alpha}{2}} \rndBrk{\rho_{A_1 a_j}^{\frac{1}{2}} \omega^{-\alpha'}_{a_j} \rho_{A_1 a_j}^{\frac{1}{2}}}^{-\frac{\alpha}{2}}\otimes \rho_{c_j}^{-\frac{1}{2}}  \\
        & \qquad \qquad \qquad \cdot \frac{1}{N^{\frac{1}{2}}} \bigoplus_{j \in J'} p(j)^{\frac{\alpha}{2}}q(j)^{\frac{1-\alpha}{2}} \rndBrk{\rho_{A_1 a_j}^{\frac{1}{2}} \omega^{-\alpha'}_{a_j} \rho_{A_1 a_j}^{\frac{1}{2}}}^{\frac{\alpha}{2}} \otimes \rho_{c_j B_2}^{\frac{1}{2}} \\
        &= \bigoplus_{j \in J'} \rndBrk{\rho_{A_1 a_j}^{\frac{1}{2}} \omega^{-\alpha'}_{a_j} \rho_{A_1 a_j}^{\frac{1}{2}}}^{0} \otimes \rho_{c_j}^{-\frac{1}{2}} \rho_{c_j B_2}^{\frac{1}{2}} \\
        &= \bigoplus_{j \in J'} \rho_{A_1 a_j}^{0} \otimes \rho_{c_j}^{-\frac{1}{2}} \rho_{c_j B_2}^{\frac{1}{2}}
    \end{align*}
    where in the last line we have used that the projector $\rndBrk{\rho_{A_1 a_j}^{\frac{1}{2}} \omega^{-\alpha'}_{a_j} \rho_{A_1 a_j}^{\frac{1}{2}}}^{0}$ is equal to the projector $\rho_{A_1 a_j}^{0}$ for every $j \in J'$ (here $P^0$ is the projector onto the image of positive semidefinite operator $P$). This can be seen since for every $j \in J'$ we first have
    \begin{align}
        \im\rndBrk{\rho_{A_1 a_j}^{\frac{1}{2}} \omega^{-\alpha'}_{a_j} \rho_{A_1 a_j}^{\frac{1}{2}}} \subseteq \im \rndBrk{\rho_{A_1 a_j}}.
        \label{eq:proj_equal_first}
    \end{align}
    Second, we have that Eq. \ref{eq:omega_j_ll_cond} above implies that $\omega_{a_j}^0 \rho_{A a_j}^0 = \rho_{A a_j}^0$ for every $j \in J'$. Now, for $j \in J'$ we have the following inequality 
    \begin{align*}
        \rndBrk{\rho_{A_1 a_j}^{\frac{1}{2}} \omega^{-\alpha'}_{a_j} \rho_{A_1 a_j}^{\frac{1}{2}}} &\geq m \rndBrk{\rho_{A_1 a_j}^{\frac{1}{2}} \omega^{0}_{a_j} \rho_{A_1 a_j}^{\frac{1}{2}}} \\
        &= m \rho_{A_1 a_j}
    \end{align*}
    where $m>0$ is the minimum non-zero eigenvalue of $\omega^{-\alpha'}_{a_j}$. Finally, raising the above to the power of $0$ (this action is operator monotone)
    \begin{align}
        \rndBrk{\rho_{A_1 a_j}^{\frac{1}{2}} \omega^{-\alpha'}_{a_j} \rho_{A_1 a_j}^{\frac{1}{2}}}^0 &\geq \rho_{A_1 a_j}^0.
        \label{eq:proj_equal_sec}
    \end{align}
    Eq. \ref{eq:proj_equal_first} and \ref{eq:proj_equal_sec} together imply that for $j \in J'$
    \begin{align*}
        \rndBrk{\rho_{A_1 a_j}^{\frac{1}{2}} \omega^{-\alpha'}_{a_j} \rho_{A_1 a_j}^{\frac{1}{2}}}^0 = \rho_{A_1 a_j}^0.
    \end{align*}
    Finally, we have that 
    \begin{align*}
        \rho_{A_1 B_1}^{-\frac{1}{2}} \rho_{A_1 B_1 B_2 }^{\frac{1}{2}} &= \bigoplus_{j \in J'} p(j)^{-\frac{1}{2}} \rho_{A_1 a_j}^{-\frac{1}{2}}\otimes \rho_{c_j}^{-\frac{1}{2}} \bigoplus_{j \in J'} p(j)^{\frac{1}{2}} \rho_{A_1 a_j}^{\frac{1}{2}}\otimes \rho_{c_j B_2}^{\frac{1}{2}} \\
        &= \bigoplus_{j \in J'} \rho_{A_1 a_j}^{0}\otimes \rho_{c_j}^{-\frac{1}{2}} \rho_{c_j B_2}^{\frac{1}{2}}. 
    \end{align*}
    This proves that 
    \begin{align}
        \nu_{A_1 B_1}^{-\frac{1}{2}} \nu_{A_1 B_1 B_2 }^{\frac{1}{2}}= \rho_{A_1 B_1}^{-\frac{1}{2}} \rho_{A_1 B_1 B_2 }^{\frac{1}{2}}
        \label{eq:int_eq_cond_st_equal}
    \end{align}
    and hence 
    \begin{align*}
        \nu_{A_2 B_2 | A_1 B_1} &= \nu_{A_1 B_1}^{-\frac{1}{2}} \nu_{A_1 B_1 B_2 }^{\frac{1}{2}} \nu_{A_2 | A_1 B_1 B_2} \nu_{A_1 B_1 B_2 }^{\frac{1}{2}} \nu_{A_1 B_1}^{-\frac{1}{2}} \\
        &= \rho_{A_1 B_1}^{-\frac{1}{2}} \rho_{A_1 B_1 B_2 }^{\frac{1}{2}} \rho_{A_2 | A_1 B_1 B_2} \rho_{A_1 B_1 B_2 }^{\frac{1}{2}} \rho_{A_1 B_1}^{-\frac{1}{2}} \\
        &= \rho_{A_2 B_2 | A_1 B_1}
    \end{align*}
    where we have used the fact that $\nu_{A_2 | A_1 B_1 B_2} = \rho_{A_2 | A_1 B_1 B_2}$ and Eq. \ref{eq:int_eq_cond_st_equal}. We can now modify Eq. \ref{eq:Halpha_ch_rule_for_Mk_ch} to get
    \begin{align*}
        \tilde{H}^{\uparrow}_\alpha (A_1 A_2 | B_1 B_2)_\rho \geq \tilde{H}^{\uparrow}_\alpha (A_1| B_1 B_2)_\rho + \inf_{\nu} \tilde{H}^{\downarrow}_\alpha (A_2 | A_1 B_1 B_2)_\nu
    \end{align*}
    where the infimum is over states $\nu$ such that $\nu_{A_2 B_2 | A_1 B_1} = \rho_{A_2 B_2 | A_1 B_1}$. We can use the data processing inequality to get 
    \begin{align*}
        \tilde{H}^{\uparrow}_\alpha (A_1| B_1 B_2)_\rho &= \tilde{H}^{\uparrow}_\alpha (A_1| B_1 B_2)_{\mathcal{R}_{B_1 \rightarrow B_1 B_2}(\rho_{A B_1})} \\
        &\geq \tilde{H}^{\uparrow}_\alpha (A_1| B_1)_{\rho}.
    \end{align*}
    Together with the above inequality this proves the Lemma. 
\end{proof}
We will use the following modification of \cite[Corollary 3.5]{Dupuis20}. 
\begin{corollary}
    Let $\cM_{R \rightarrow A_2 B_2}$ be a channel and $\rho_{A_1 A_2 B_1 B_2}= \cM(\rho'_{A_1 B_1 R})$ such that the Markov chain $A_1 \leftrightarrow B_1 \leftrightarrow B_2$ holds. Then, we have 
    \begin{align}
        \tilde{H}^{\uparrow}_\alpha (A_1 A_2 | B_1 B_2)_\rho \geq \tilde{H}^{\uparrow}_\alpha (A_1| B_1)_\rho + \inf_{\omega} \tilde{H}^{\downarrow}_\alpha (A_2 | A_1 B_1 B_2)_{\cM(\omega)}
    \end{align}
    where the infimum is taken over all states $\omega_{A_1 B_1 R}$. Moreover, if $\rho'_{A_1 B_1 R}$ is pure then we can restrict the optimisation to pure states. 
    \label{cor:Markoc_ch_opt_Halpha_bd}
\end{corollary}
\begin{proof}
    The proof is the same as \cite[Corollary 3.5]{Dupuis20}. We include it here for the sake of completeness. \\
    It is sufficient to show that for every state $\nu$ such that $\nu_{A_2 B_2 | A_1 B_1} = \rho_{A_2 B_2 | A_1 B_1}$, there exists an $\omega_{A_1 B_1 R}$ such that $\nu_{A_1 A_2 B_1 B_2} = \cM(\omega)$. For such a $\nu$, we can define 
    \begin{align*}
        \omega_{R A_1 B_1} = \nu_{A_1 B_1}^{\frac{1}{2}} \rho_{A_1 B_1}^{-\frac{1}{2}} \rho'_{A_1 B_1 R} \rho_{A_1 B_1}^{-\frac{1}{2}} \nu_{A_1 B_1}^{\frac{1}{2}}
    \end{align*}
    which can be seen to be a valid state and also satisfy $\nu_{A_1 A_2 B_1 B_2} = \cM(\omega)$. 
\end{proof}

\section{Dimension bounds for conditional R\'enyi entropies}

\begin{lemma}[Dimension bound]
    For $\alpha \in [\frac{1}{2}, \infty]$, a state $\rho_{A_1 A_2 B}$, the following bounds hold for the sandwiched conditional entropies
    \begin{align*}
        \tilde{H}^{\downarrow}_\alpha (A_1| B)_\rho - \log|A_2| &\leq \tilde{H}^{\downarrow}_\alpha (A_1 A_2 | B)_\rho \leq \tilde{H}^{\downarrow}_\alpha (A_1| B)_\rho + \log|A_2| \\
        \tilde{H}^{\uparrow}_\alpha (A_1| B)_\rho - \log|A_2| &\leq \tilde{H}^{\uparrow}_\alpha (A_1 A_2 | B)_\rho \leq \tilde{H}^{\uparrow}_\alpha (A_1| B)_\rho + \log|A_2|. 
    \end{align*}
    For $\alpha \in [0,2]$ and a state $\rho_{A_1 A_2 B}$, the following bounds hold for the Petz conditional entropies
    \begin{align*}
        \bar{H}^{\downarrow}_\alpha (A_1 A_2 | B)_\rho &\leq \bar{H}^{\downarrow}_\alpha (A_1| B)_\rho + \log|A_2| \\
        \bar{H}^{\uparrow}_\alpha (A_1 A_2 | B)_\rho &\leq \bar{H}^{\uparrow}_\alpha (A_1| B)_\rho + \log|A_2|.
    \end{align*}
    \label{lemm:dim_bd}
\end{lemma}
\begin{proof}
    For the sandwiched conditional entropies, we simply use the corresponding chain rules (Corollary \ref{cor:Halpha_ch_rule} or Corollary \ref{cor:opt_Halpha_ch_rule}) along with the fact that for all states $\nu$, $\tilde{H}^{\downarrow}_\alpha (A_2| A_1 B)_\nu \in [-\log |A_2|, \log |A_2|]$ \cite[Lemma 5.2]{TomamichelBook16}.\\
    For the Petz conditional entropies, we will make use of the Jensen's inequality for operators \cite[Theorem V.2.3]{Bhatia97}. Suppose, $\{\ket{e_i}\}_{i=1}^{|X|}$ is an orthogonal basis for the space $X$. Then, we have for a positive operator $P_{XY}$ and $\alpha \in [0,1]$
    \begin{align*}
        \tr_{X} P_{XY}^\alpha &= \sum_{i=1}^{|X|} \Id_Y \otimes \bra{e_{i}}_X P_{XY}^\alpha \Id_Y \otimes \ket{e_{i}}_X \\
        &\leq |X| \rndBrk{\sum_{i=1}^{|X|} \frac{1}{|X|} \Id_Y \otimes \bra{e_{i}}_X P_{XY} \Id_Y \otimes \ket{e_{i}}_X}^\alpha \\
        &= |X|^{1-\alpha} P_Y^\alpha \numberthis \label{eq:partial_tr_op_ineq_alpha_leq1}
    \end{align*}
    where in the second step we have used the operator Jensen's inequality with the operators $\curlyBrk{\frac{1}{\sqrt{|X|}} \Id_Y \otimes \ket{e_{i}}_X}_{i=1}^{|X|}$ along with the fact that the map $X \mapsto X^{\alpha}$ is operator concave. For $\alpha \in [1,2]$ and positive operator $P_{XY}$, we can use the same argument as above and the fact that $X \mapsto X^{\alpha}$ is operator convex in this regime and derive
    \begin{align}
        \tr_{X} P_{XY}^\alpha &\geq |X|^{1-\alpha} P_Y^\alpha.
        \label{eq:partial_tr_op_ineq_alpha_geq1}
    \end{align}
    To prove the dimension bound, observe that for a positive state $\sigma_B$ and $\alpha\in [0,2]$, we have 
    \begin{align*}
        -\bar{D}_{\alpha}(\rho_{A_1 A_2 B}|| \Id_{A_1 A_2}\otimes \sigma_B) &= \frac{1}{1-\alpha} \log \tr \rndBrk{\rho_{A_1 A_2 B}^{\alpha} \sigma_B^{1-\alpha}} \\
        &= \frac{1}{1-\alpha} \log \tr \rndBrk{\tr_{A_2}\rndBrk{\rho_{A_1 A_2 B}^{\alpha}} \sigma_B^{1-\alpha}}\\
        &\leq \frac{1}{1-\alpha} \log \tr \rndBrk{|A_2|^{1-\alpha} \rho_{A_1 B}^{\alpha} \sigma_B^{1-\alpha}}\\
        &=  -\bar{D}_{\alpha}(\rho_{A_1 B}|| \Id_{A_1}\otimes \sigma_B) + \log|A_2|.
    \end{align*}
    We can now take a supremum over $\sigma_B$ to prove the dimension bound for $\bar{H}^{\uparrow}_\alpha$ or choose $\sigma_B = \rho_B$ to prove the dimension bound for $\bar{H}^{\downarrow}_\alpha$.
\end{proof}
The following Lemma was originally proven in \cite[Proposition 8]{Lennert13}. We reproduce the proof argument here.
\begin{lemma}
    For $\alpha \in [\frac{1}{2}, \infty]$, a state $\rho_{ABC}$, we have 
    \begin{align}
        \tilde{H}^{\uparrow}_\alpha (A| BC)_\rho \geq \tilde{H}^{\uparrow}_\alpha(AC|B)_\rho - \log|C|
        \label{eq:Halpha_dim_bd_tomamichel1}
    \end{align}
    and for $\alpha \in [0, 2]$
    \begin{align}
        \bar{H}^{\uparrow}_\alpha (A| BC)_\rho \geq \bar{H}^{\uparrow}_\alpha(AC|B)_\rho - \log|C|
        \label{eq:Halpha_dim_bd_tomamichel2}
    \end{align}
    \label{lemm:Halpha_dim_bd_tomamichel}
\end{lemma}
\begin{proof}
    By the definition of the sandwiched conditional entropy, we have
    \begin{align*}
        \tilde{H}^{\uparrow}_\alpha (A| BC) &= \sup_{\eta_{BC} \in D(BC)} -\tilde{D}_\alpha(\rho_{ABC}|| \Id_{A} \otimes \eta_{BC})\\
        &\geq \sup_{\eta_{B} \in D(B)} -\tilde{D}_\alpha\rndBrk{\rho_{ABC}|| \Id_{A} \otimes \frac{\Id_C}{|C|} \otimes \eta_{B}} \\
        &= \sup_{\eta_{B} \in D(B)} -\tilde{D}_\alpha\rndBrk{\rho_{ABC}|| \Id_{AC} \otimes \eta_{B}} - \log|C| \\
        &= \tilde{H}^{\uparrow}_\alpha (AC| B) - \log|C|
    \end{align*}
    where we simply restrict the supremum in the second line to states of the form $\eta_{BC} = \eta_{B} \otimes \frac{\Id_C}{|C|}$ to derive the inequality. The same proof also works with $\bar{H}^{\uparrow}_\alpha$ entropy.
\end{proof}
The following lemma was originally proven in \cite[Proposition 3.3.5]{Leditzky16-thesis}.
\begin{lemma}[Dimension bound for conditioning register]
    For $\alpha \in [\frac{1}{2}, \infty]$ and a state $\rho_{ABC}$ we have 
    \begin{align}
        \tilde{H}^{\uparrow}_\alpha (A| B C)_\rho \geq \tilde{H}^{\uparrow}_\alpha (A| B)_\rho - 2\log|C|.
        \label{eq:cond_reg_dim_bd}
    \end{align}
    Further, if the register $C$ is classical, then we have
    \begin{align}
        \tilde{H}^{\uparrow}_\alpha (A| B C)_\rho \geq \tilde{H}^{\uparrow}_\alpha (A| B)_\rho - \log|C|.
        \label{eq:cond_reg_dim_bd_cl}
    \end{align}
    \label{lemm:cond_reg_dim_bd}
\end{lemma}
\begin{proof}
    This bound can be proven by combining Lemma \ref{lemm:dim_bd} and Lemma \ref{lemm:Halpha_dim_bd_tomamichel}. In the case that $C$ is classical, we have the inequality $\tilde{H}^{\uparrow}_\alpha (A C| B)_\rho \geq \tilde{H}^{\uparrow}_\alpha (A| B)_\rho$ \cite[Lemma 5.3]{TomamichelBook16}. 
\end{proof}

\section{Necessity for constraints on side information size for approximate AEP and EAT and its implication for approximate GEAT}
\label{sec:size_of_B_approx_EAT}

It turns out that it is necessary to place some sort of bound on the size of the side information for an approximate entropy accumulation theorem of the form in Theorem \ref{th:approx_EAT}. The following classical example demonstrates this. This example also demonstrates the necessity for a bound on the size of the side information in an approximate asymptotic equipartition of the form in Theorem \ref{th:weak_approx_AEP}.\\

Let there be $n$ rounds. For $k \in [n]$, the map $\cM_k : A_1^{k-1} \rightarrow A_k B_k C_k$. This map sets the variables as follows:
\begin{enumerate}
    \item Measure $A_1^{k-1}$ in the standard basis. 
    \item Let $A_k \in_R \{0,1\}$ be a randomly chosen bit.
    \item Let $C_k = 0$ with probability $\frac{\epsilon}{2}$ and $C_k = 1$ otherwise.
    \item In the case that $C_k = 1$, let $B_k \in_R \{0,1\}^n$ be a randomly chosen $n$-bit string. Otherwise, let $B_k = A_1^k R_k$, where $R_k$ is an $(n-k)$ bit randomly chosen string from $\{0,1\}$. 
\end{enumerate}
\begin{sloppypar}
Let $\cM'_k$ be the map which always chooses $B_k$ to be a random $n$-bit string. It is easy to see that in this case, we have $H_{\min}(A_1^n | B_1^n C_1^n)_{\cM'_n \circ \cdots \circ \cM'_1(1)} = n$ whereas $H_{\min}(A_1^n | B_1^n C_1^n)_{\cM_n \circ \cdots \circ \cM_1(1)} = O(1)$ even though for every $k \in [n]$, the maps $\cM_k$ are $\epsilon-$close in diamond norm distance to the maps $\cM'_k$. This proves that a bound on the size of the side registers is indeed necessary for approximate entropy accumulation. We show these facts formally in the following. 
\end{sloppypar}
\begin{lemma}
    Suppose $\Phi: R \rightarrow A$ and $\Phi': R \rightarrow A$ are two channels which take a register $R$ and measure it in the standard basis and map the resulting classical register $C$ to the classical register $A$. Then, for every $\rho_{R R'}$, we have 
    \begin{align}
        \norm{\Phi(\rho_{R R'}) - \Phi'(\rho_{R R'})}_{1} \leq \norm{P^{\Phi}_{AC} - P_{AC}^{\Phi'}}_1
    \end{align}
    where $P^{\Phi}_{AC}$ and $P^{\Phi'}_{AC}$ are the classical distributions produced when the maps $\Phi$ and $\Phi'$ are applied to the state $\rho_{R R'}$ respectively. 
\end{lemma}
\begin{proof}
    Let $\{\ket{c}\bra{c}\}_c$ represent the measurement in the standard basis. Since, both the channels first measure register $R$ in the standard basis, they produce the state
    \begin{align*}
        \rho_{CR'} &= \sum_{c} \ket{c}\bra{c}_C \otimes \tr_R \rndBrk{\ket{c}\bra{c}_R\rho_{RR'}} \\
        &= \sum_c p(c) \ket{c}\bra{c}_C \otimes \rho_{R'|c}
    \end{align*}
    where we have defined $p(c):= \tr \rndBrk{\ket{c}\bra{c}_R\rho_{R}}$ and $\rho_{R'|c} := \frac{1}{p(c)}\tr_R \rndBrk{\ket{c}\bra{c}_R\rho_{RR'}}$. Now, the action of channel $\Phi$ on register $C$ can be represented using the conditional probability distribution $p^{\Phi}_{A|C}$ and the action of channel $\Phi'$ on register $C$ can be similarly represented using $p^{\Phi'}_{A|C}$. We can define the states
    \begin{align*}
        \rho^{\Phi}_{ACR'} := \sum_{ac} p^{\Phi}_{A|C}(a|c) p(c) \ket{a,c}\bra{a,c} \otimes \rho_{R'|c} \\
        \rho^{\Phi'}_{ACR'} := \sum_{ac} p^{\Phi'}_{A|C}(a|c) p(c) \ket{a,c}\bra{a,c} \otimes \rho_{R'|c}.
    \end{align*}
    Note that $\tr_C \rndBrk{\rho^{\Phi}_{ACR'}} = \Phi(\rho_{RR'})$ and $\tr_C \rndBrk{\rho^{\Phi'}_{ACR'}} = \Phi'(\rho_{RR'})$. Further, we can view the $R'$ register of $\rho^{\Phi}_{ACR'}$ and $\rho^{\Phi'}_{ACR'}$ as being created by a channel which measures the register $C$ and outputs the state $\rho_{R'|c}$ in the register $R'$. Therefore, we have
    \begin{align*}
        \norm{\Phi(\rho_{RR'}) - \Phi'(\rho_{RR'})}_1 &\leq \norm{\rho^{\Phi}_{ACR'} - \rho^{\Phi'}_{ACR'}}_1 \\
        &\leq \norm{\rho^{\Phi}_{AC} - \rho^{\Phi'}_{AC}}_1 \\
        &= \norm{P^{\Phi}_{AC} - P_{AC}^{\Phi'}}_1.
    \end{align*}
\end{proof}
We can use the above lemma to evaluate the distance between the channels $\cM_k$ and $\cM'_k$. Using the above lemma, it is sufficient to suppose that the input of the channels are classical. We can suppose that the registers $A_1^{k-1}$ are classical and distributed as $P_{A_1^{k-1}}$. Let $P_{A_1^k B_k C_k}$ be the output of $\cM_k$ on this distribution and $Q_{A_1^k B_k C_k}$ be the output of applying $\cM'_k$. Then, we have 
\begin{align*}
    \norm{P_{A_1^k B_k C_k} - Q_{A_1^k B_k C_k}}_1 &= \sum_{a_1^k, c_k} P(a_1^{k-1})P(a_k) P(c_k) \norm{P_{B_k| a_1^k, c_k} - Q_{B_k}}_1 \\
    &= \sum_{a_1^k} P(a_1^{k-1})P(a_k) \rndBrk{\rndBrk{1-\frac{\epsilon}{2}}\norm{P_{B_k| a_1^k, c_k=1} - Q_{B_k}}_1 + \frac{\epsilon}{2} \norm{P_{B_k| a_1^k, c_k=0} - Q_{B_k}}_1} \\
    &\leq \sum_{a_1^k} P(a_1^{k-1})P(a_k) \epsilon \\
    &= \epsilon 
\end{align*}
where in the first line we have used the fact that $A_k$ and $C_k$ are chosen independently with the same distribution in both the maps and the fact that $B_k$ is chosen independently in $\cM'_k$, for the third line we have used the fact that $B_k$ is independent and has the same distribution as $Q_{B_k}$ when $c_k=1$. Since, this is true for all input distributions, we have $\norm{\cM_k - \cM'_k}_{\diamond} \leq \epsilon$. \\

Now, let $R_{A_1^n B_1^n C_1^n}$ be the probability distribution created when the maps $\cM_k$ are applied sequentially $n$ times and $S_{A_1^n B_1^n C_1^n}$ be the probability distribution created when the maps $\cM'_k$ are applied sequentially $n$ times. Since, $B_k$ and $C_k$ are independent of $A_k$ in the distribution $S$, we have 
\begin{align*}
    H_{\min}(A_1^n | B_1^n C_1^n)_S = n.
\end{align*}
We will show that $H_{\min}^{\epsilon'} (A_1^n | B_1^n C_1^n)_R = O(1)$ as long as $\epsilon' \leq \frac{1}{4}$. Let $l := \frac{2}{\epsilon} \log\frac{1}{\epsilon'}$. Let $E$ be the event that there exists a $k > n- l$ such that $C_k=0$. For our choice of $l$, we have $p(E) \geq 1- \epsilon'$.

\begin{lemma}
    Let $P_{AB}$ be a subnormalised probability distribution such that $A= f(B)$ for some function $f$ (that is, $P(a,b)>0$ only if $ a= f(b)$). Then, $H^{\epsilon}_{\min} (A|B)_P \leq \log \frac{1}{\tr(P) - \sqrt{2\epsilon}}$. 
    \label{lemm:det_P_Hmin_eps_bd}
\end{lemma}
\begin{proof}
    Let $P'_{AB}$ be a distribution $\epsilon$-close to $P$ in purified distance. Then, it is $\sqrt{2\epsilon}$ close to $P$ in trace distance. We have that 
    \begin{align*}
        2^{-H_{\min}(A|B)_{P'}} &= P'_{\text{guess}}(A|B) \\
        &\geq \sum_{b} P'_{AB}(f(b), b) \\
        &\geq \sum_{b} P_{AB}(f(b), b) - \sqrt{2\epsilon} \\
        &= \tr(P) - \sqrt{2\epsilon}
    \end{align*}
    which implies that $H_{\min}(A|B)_{P'} \leq \log \frac{1}{\tr(P) - \sqrt{2\epsilon}}$. Since, this is true for every distribution $\epsilon$-close to $P$, it also holds for $H^{\epsilon}_{\min} (A|B)_P$. 
\end{proof}
We then have that
\begin{align*}
    H^{\epsilon'}_{\min}(A_1^n | B_1^n C_1^n)_{R} &\leq H^{\epsilon'}_{\min}(A_1^n | B_1^n C_1^n \wedge E)_{R} \\
    &\leq H^{\epsilon'}_{\min}(A_1^{n-l} | B_1^n C_1^n \wedge E)_{R} + l \\
    &\leq \log \frac{1}{p(E)- \sqrt{2\epsilon'}} + l \\
    &\leq \log \frac{1}{1- \epsilon'- \sqrt{2\epsilon'}} + l \\
    &\leq l + \log 8/3 = O(1)
\end{align*}
where in the first line we have used \cite[Lemma 10]{Tomamichel17} in the first line, dimension bound (can be proven using Lemma \ref{lemm:dim_bd}) in the second line, Lemma \ref{lemm:det_P_Hmin_eps_bd} in the third line and the fact that $p(E)\geq 1- \epsilon'$. \\

Also, note that the example given here satisfies
\begin{align*}
    \norm{P_{A_1^k B_1^k C_1^k} - P_{A_1^{k-1} B_1^{k-1} C_1^{k-1}} P_{A_k B_k C_k}}_1 \leq \epsilon
\end{align*}
for every $k$. This also proves that a bound on the size of the side information registers ($B_k C_k$ here), as we have in Theorem \ref{th:weak_approx_AEP}, is necessary for an approximate version of AEP.\\

Further, this example also rules out the possibility of a natural approximate extension to the generalised entropy accumulation theorem (GEAT) \cite{Metger22} where the maps $\cM_k \approx_\epsilon \cM'_k$ and the maps $\cM'_k$ satisfy the non-signalling conditions because one can write the entropy accumulation scenario in the form of a generalised entropy accumulation scenario where Eve's information contains the side information $B_1^k E$ in each step. Thus, it would not be possible to prove a meaningful bound on the smooth min-entropy without some sort of bound on the information transferred between the adversary's register $E_i$ and the register $R_i$. 

\section{Classical approximate entropy accumulation}

We present a simple proof for the approximate entropy accumulation theorem for classical distributions. This result also requires a much weaker assumption than Theorem \ref{th:approx_EAT}.

\begin{theorem}
    \label{th:approx_cl_EAT}
    Let $p_{A_1^n B_1^n E}$ be a classical distribution such that for every $k \in [n]$, and $a_1^{k-1}, b_1^{k-1}$ and $e$ 
    \begin{align}
        \norm{p_{A_k B_k | a_1^{k-1}, b_1^{k-1}, e} - q^{(k)}_{A_k B_k | a_1^{k-1}, b_1^{k-1}, e}}_{\infty} \leq \epsilon
    \end{align}
    where $\norm{v}_{\infty}:= \max_{i} |v(i)|$ and the $q^{(k)}_{B_k | a_1^{k-1}, b_1^{k-1}, e} = q^{(k)}_{B_k | b_1^{k-1}, e}$ or equivalently $q^{(k)}$ satisfies the Markov chain $A_k \leftrightarrow B_1^{k-1}E \leftrightarrow B_k$. Also, let $|A_k| = |A|$, $|B_k| = |B|$ for every $k \in [n]$.\\

    Then, for $\epsilon' \in  (0,1)$ and $\alpha \in \rndBrk{1, 1+ \frac{1}{\log(1+2|A|)}}$, we have that 
    \begin{align*}
        H_{\min}^{ \epsilon'}(A_1^n |B_1^n E)_{p} &\geq \sum_{k=1}^n \inf_{q} H(A_k | B_k A_1^{k-1}B_1^{k-1}E)_{q^{(k)}_{A_k B_k | A_1^{k-1}B_1^{k-1}E}q_{A_1^{k-1}B_1^{k-1}E}} \\
        &\quad\quad - n(\alpha-1)\log^2(2|A| + 1)- \frac{\alpha}{\alpha-1} n \log(1 + \epsilon |A||B|) -\frac{g_0(\epsilon')}{\alpha-1}. \numberthis
    \end{align*}
    where $g_0(x) := - \log(1 - \sqrt{1-x^2})$. The infimums are taken over all possible input probability distributions.
\end{theorem}
For $\alpha = 1 + \sqrt{\epsilon}$ (assuming $\sqrt{\epsilon} \leq 1+ \frac{1}{\log(1+2|A|)}$), and using $\alpha \leq 2$ and $\log(1+x) \leq x$ as long as $x \geq 0$, the above bound gives us
\begin{align*}
    H_{\min}^{ \epsilon'}(A_1^n |B_1^n E)_{p} &\geq \sum_{k=1}^n \inf_{q} H(A_k | B_k A_1^{k-1}B_1^{k-1}E)_{q^{(k)}_{A_k B_k | A_1^{k-1}B_1^{k-1}E}q_{A_1^{k-1}B_1^{k-1}E}} \\
    &\quad\quad - n\sqrt{\epsilon}\rndBrk{\log^2(2|A| + 1)- 2 |A||B|} -\frac{g_0(\epsilon')}{\alpha-1} \numberthis
\end{align*}
\begin{proof}
    For every $k \in [n]$, we modify $q^{(k)}_{A_k B_k | A_1^{k-1} B_1^{k-1} E}$ to create the distributions $r^{(k)}_{A_k B_k | A_1^{k-1} B_1^{k-1} E}$, which are defined as follows
    \begin{enumerate}
        \item Choose a random variable $C_k$ from $\{0,1\}$ with probabilities $\rndBrk{\frac{|A||B|\epsilon}{1+|A||B|\epsilon}, \frac{1}{1+|A||B|\epsilon}}$. 
        \item If $C_k = 1$, then choose random variables $A_k, B_k$ using $q^{(k)}_{A_k B_k | A_1^{k-1} B_1^{k-1} E}$ else choose $A_k, B_k$ randomly with probability $\frac{1}{|\mathcal{A}||\mathcal{B}|}$.
    \end{enumerate}
    That is, we have 
    \begin{align*}
        r^{(k)}_{A_k B_k | A_1^{k-1} B_1^{k-1} E} := \frac{1}{1+|A||B|\epsilon} q^{(k)}_{A_k B_k | A_1^{k-1} B_1^{k-1} E} + \frac{|A||B|\epsilon}{1+|A||B|\epsilon} u_{A_k B_k}
    \end{align*}
    where $u_{A_k B_k}$ is the uniform distribution on the registers $A_k$ and $B_k$. \\

    For every $k, a_1^{k-1}, b_1^{k-1}$, and $e$, we have
    \begin{align*}
        & \norm{p_{A_k B_k | a_1^{k-1}, b_1^{k-1}, e} - q^{(k)}_{A_k B_k | a_1^{k-1}, b_1^{k-1}, e}}_\infty \leq \epsilon\\
        \Rightarrow\ & p_{A_k B_k | a_1^{k-1}, b_1^{k-1}, e} \leq q^{(k)}_{A_k B_k | a_1^{k-1}, b_1^{k-1}, e} + \epsilon \Id_{A_k B_k} \\
        \Rightarrow\ & p_{A_k B_k | a_1^{k-1}, b_1^{k-1}, e} \leq q^{(k)}_{A_k B_k | a_1^{k-1}, b_1^{k-1}, e} + \epsilon |A||B| u_{A_k B_k} \\
        \Rightarrow\ & p_{A_k B_k | a_1^{k-1}, b_1^{k-1}, e} \leq (1+|A||B|\epsilon) r^{(k)}_{A_k B_k | A_1^{k-1} B_1^{k-1} E}
    \end{align*}
    Define the distribution 
    \begin{align}
        r_{A_1^n B_1^n E} = \prod_{k=1}^n r^{(k)}_{A_k B_k | A_1^{k-1} B_1^{k-1} E} p_E.
    \end{align}
    Note that for every $k, a_1^{k-1}, b_1^k$, and $e$, we have 
    \begin{align*}
        r_{B_k | A_1^{k-1} B_1^{k-1} E}(b_k | a_1^{k-1} b_1^{k-1} e) &= \frac{1}{1+|A||B|\epsilon} q^{(k)}_{B_k | A_1^{k-1} B_1^{k-1} E}(b_k | a_1^{k-1} b_1^{k-1} e) + \frac{\epsilon}{1+|A||B|\epsilon} \\
        &= \frac{1}{1+|A||B|\epsilon} q^{(k)}_{B_k | B_1^{k-1} E}(b_k | b_1^{k-1} e) + \frac{\epsilon}{1+|A||B|\epsilon},
    \end{align*}
    which implies
    \begin{align*}
        r_{B_k | B_1^{k-1} E}(b_k | b_1^{k-1} e) &= \sum_{\bar{a}_1^{k-1}} r_{A_1^{k-1} | B_1^{k-1} E}(\bar{a}_1^{k-1}| b_1^{k-1} e ) r_{B_k | A_1^{k-1} B_1^{k-1} E}(b_k | \bar{a}_1^{k-1} b_1^{k-1} e)\\
        &= \sum_{\bar{a}_1^{k-1}} r_{A_1^{k-1} | B_1^{k-1} E}(\bar{a}_1^{k-1}| b_1^{k-1} e )\rndBrk{\frac{1}{1+|A||B|\epsilon} q^{(k)}_{B_k | B_1^{k-1} E}(b_k | b_1^{k-1} e) + \frac{\epsilon}{1+|A||B|\epsilon}}\\
        &= r_{B_k | A_1^{k-1} B_1^{k-1} E}(b_k | a_1^{k-1} b_1^{k-1} e).
    \end{align*}
    Thus, for every $k \in [n]$, $r$ satisfies the Markov chain $A_1^{k-1} \leftrightarrow B_1^{k-1} E \leftrightarrow B_k$. Further, we have 
    \begin{align*}
        p_{A_1^n B_1^n E}(a_1^n, b_1^n, e) &= \prod_{k=1}^n p_{A_k B_k | A_1^{k-1}, B_1^{k-1}, E} (a_k, b_k | a_1^{k-1}, b_1^{k-1}, e) p_E(e)\\
        &\leq (1+\epsilon|A||B|)^n \prod_{k=1}^n r^{(k)}_{A_k B_k | A_1^{k-1}, B_1^{k-1}, E} (a_k, b_k | a_1^{k-1}, b_1^{k-1}, e) p_E(e) \\
        &= (1+\epsilon|A||B|)^n r_{A_1^n B_1^n E}(a_1^n, b_1^n, e)
    \end{align*}
    which shows that $D_{\max}(p_{A_1^n B_1^n E} || r_{A_1^n B_1^n E}) \leq n \log(1 + \epsilon|A||B|)$.\\

    \begin{figure}
        \centering
        \includegraphics[scale=0.2]{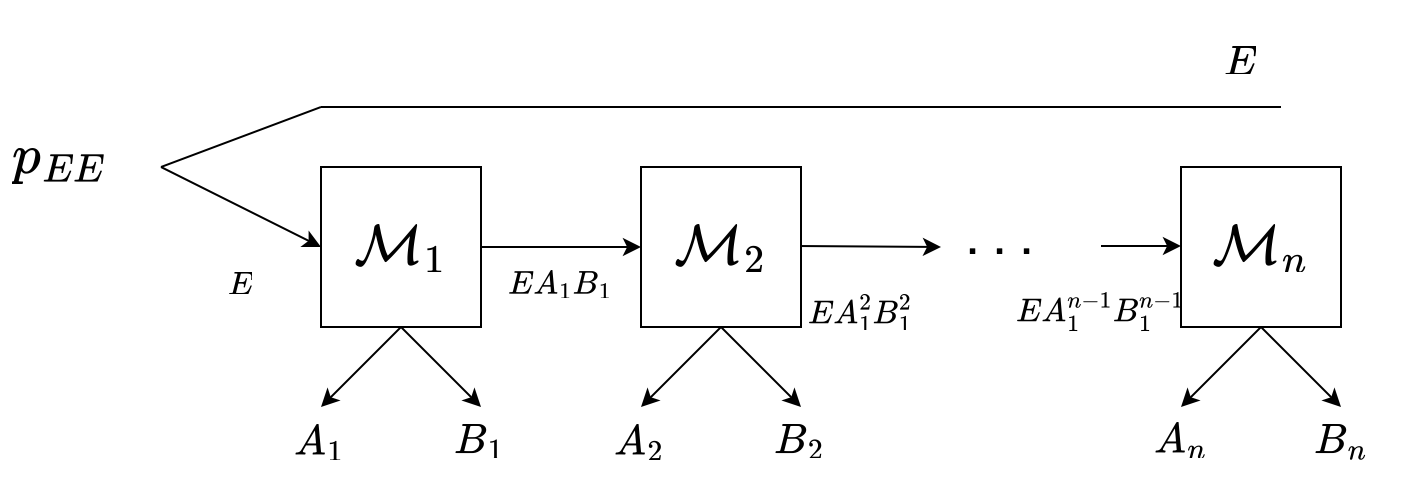}
        \caption{Setting for classical EAT}
        \label{fig:cl_approx_EAT_diag}
    \end{figure}

    The distribution $r_{A_1^n B_1^n E}$ can be viewed as the result of a series of maps as in Fig. \ref{fig:cl_approx_EAT_diag}. We can now use the EAT chain rule \cite[Corollary 3.5]{Dupuis20} along with \cite[Lemma B.9]{Dupuis20} $n$-times to bound the entropy of this auxiliary distribution. We get 
    \begin{align*}
        \tilde{H}_{\alpha}^{\uparrow}(A_1^n | B_1^n E)_r &\geq \sum_{k=1}^n \inf_{q_{A_1^{k-1}B_1^{k-1}E}} \tilde{H}_{\alpha}^{\downarrow}(A_k | B_k A_1^{k-1}B_1^{k-1}E)_{r^{(k)}_{A_k B_k | A_1^{k-1}B_1^{k-1}E}q_{A_1^{k-1}B_1^{k-1}E}} \\
        &\geq  \sum_{k=1}^n \inf_{q_{A_1^{k-1}B_1^{k-1}E}} H(A_k | B_k A_1^{k-1}B_1^{k-1}E)_{r^{(k)}_{A_k B_k | A_1^{k-1}B_1^{k-1}E}q_{A_1^{k-1}B_1^{k-1}E}} - n(\alpha-1)\log^2(2|A| + 1) \\
        &\geq \sum_{k=1}^n \bigg( \inf_{q_{A_1^{k-1}B_1^{k-1}E}}  \frac{1}{1+|A||B|\epsilon} H(A_k | B_k A_1^{k-1}B_1^{k-1}E)_{q^{(k)}_{A_k B_k | A_1^{k-1}B_1^{k-1}E}q_{A_1^{k-1}B_1^{k-1}E}} \\
        & \qquad+ \frac{\epsilon}{1+|A||B|\epsilon} \log|A| \bigg) - n(\alpha-1)\log^2(2|A| + 1) \\
        &\geq \sum_{k=1}^n \inf_{q_{A_1^{k-1}B_1^{k-1}E}} H(A_k | B_k A_1^{k-1}B_1^{k-1}E)_{q^{(k)}_{A_k B_k | A_1^{k-1}B_1^{k-1}E}q_{A_1^{k-1}B_1^{k-1}E}} - n(\alpha-1)\log^2(2|A| + 1)
    \end{align*}
    for $\alpha \in \rndBrk{1, 1+\frac{1}{\log(1+2|A|)}}$. In the third line, we have used the concavity of the von Neumann entropy along with the definition of $r^{(k)}_{A_k B_k | A_1^{k-1}B_1^{k-1}E}$. Using Lemma \ref{lemm:Hmin_rho_to_Halpha_sigma_using_Dmax}, we have 
    \begin{align*}
        H_{\min}^{\epsilon'}(A_1^n |B_1^n E)_{p} &\geq \tilde{H}_{\alpha}^{\uparrow}(A_1^n |B_1^n E)_{r} - \frac{\alpha}{\alpha-1}D_{\max}(p_{A_1^n B_1^n E}||r_{A_1^n B_1^n E}) -\frac{g_1(\epsilon', 0)}{\alpha-1}\\
        &\geq \sum_{k=1}^n \inf_{q} H(A_k | B_k A_1^{k-1}B_1^{k-1}E)_{q^{(k)}_{A_k B_k | A_1^{k-1}B_1^{k-1}E}q_{A_1^{k-1}B_1^{k-1}E}} \\
        &\quad\quad - n(\alpha-1)\log^2(2|A| + 1)- \frac{\alpha}{\alpha-1} n \log(1 + \epsilon |A||B|) -\frac{g_0(\epsilon')}{\alpha-1}.
    \end{align*}    
\end{proof}

\section{Lemma to bound distance after conditioning}

The following Lemma relates the distance of two states conditioned on an event to the distance between them without conditioning. 
\begin{lemma}
    \label{lemm:dist_cond_states}
    Suppose $\rho_{XA} = \sum_{x \in \mathcal{X}} p(x) \ket{x}\bra{x}\otimes \rho_{A|x}$ and $\tilde{\rho}_{XA} = \sum_{x \in \mathcal{X}} \tilde{p}(x) \ket{x}\bra{x} \otimes \tilde{\rho}_{A|x}$ are classical-quantum states such that $\frac{1}{2}\norm{\rho_{XA}- \tilde{\rho}_{XA}}_1 \leq \epsilon$. Then, for $x \in \mathcal{X}$ such that $p(x)>0$, we have 
    \begin{align}
        \frac{1}{2}\norm{\rho_{A|x}- \tilde{\rho}_{A|x}}_1 \leq \frac{2\epsilon}{p(x)}
    \end{align} 
\end{lemma}
\begin{proof}
    \begin{align*}
        & \frac{1}{2}\norm{\rho_{XA}- \tilde{\rho}_{XA}}_1 = \frac{1}{2}\sum_{x \in \mathcal{X}} \norm{p(x)\rho_{A|x} - \tilde{p}(x) \tilde{\rho}_{A|x}}_1 \leq \epsilon
    \end{align*}
    This implies that for $x \in \mathcal{X}$
    \begin{align*}
        \frac{1}{2}\norm{p(x)\rho_{A|x} - \tilde{p}(x) \tilde{\rho}_{A|x}}_1 \leq \epsilon 
    \end{align*}
    and 
    \begin{align*}
        \frac{1}{2}\vert p(x) - \tilde{p}(x) \vert \leq \epsilon.
    \end{align*}
    Using these inequalities, we have
    \begin{align*}
        \frac{1}{2}\norm{\rho_{A|x}- \tilde{\rho}_{A|x}}_1 &\leq \frac{1}{2}\norm{\rho_{A|x}- \frac{\tilde{p}(x)}{p(x)}\tilde{\rho}_{A|x}}_1 + \frac{1}{2}\left\vert 1- \frac{\tilde{p}(x)}{p(x)} \right\vert \norm{\tilde{\rho}_{A|x}}_1 \\
        &= \frac{1}{p(x)}\frac{1}{2}\norm{p(x)\rho_{A|x}- \tilde{p}(x)\tilde{\rho}_{A|x}}_1 + \frac{1}{p(x)}\frac{1}{2}\left\vert p(x)- \tilde{p}(x) \right\vert \\
        &\leq \frac{2\epsilon}{p(x)}.
    \end{align*}
\end{proof}